\documentclass[a4paper,11pt]{article} 
\usepackage[a4paper,hmargin=2.8cm,vmargin=3cm]{geometry}
\usepackage[utf8x]{inputenc} 
\usepackage[english]{babel}
\usepackage{amsmath,amsthm,amsfonts,amssymb,mathrsfs,bbm}
\usepackage{authblk}
\usepackage{lineno}
\usepackage{slashed}
\usepackage{comment}
\usepackage[dvipsnames]{xcolor}

\definecolor{BeauBlue}{rgb}{0, 0.2, .9}
\definecolor{BeauOrange}{rgb}{.8, .1, 0}
\usepackage{hyperref}
\hypersetup{
	colorlinks = true,
	linkcolor = BeauBlue,
	urlcolor = cyan,
	citecolor = BeauOrange,
}

\numberwithin{equation}{section}

\usepackage{tikz}
\usetikzlibrary{calc,shapes,patterns,decorations.pathreplacing}

\makeatletter
\DeclareRobustCommand\widecheck[1]{{\mathpalette\@widecheck{#1}}}
\def\@widecheck#1#2{%
    \setbox\z@\hbox{\m@th$#1#2$}%
    \setbox\tw@\hbox{\m@th$#1%
       \widehat{%
          \vrule\@width\z@\@height\ht\z@
          \vrule\@height\z@\@width\wd\z@}$}%
    \dp\tw@-\ht\z@
    \@tempdima\ht\z@ \advance\@tempdima2\ht\tw@ \divide\@tempdima\thr@@
    \setbox\tw@\hbox{%
       \raise\@tempdima\hbox{\scalebox{1}[-1]{\lower\@tempdima\box
\tw@}}}%
    {\ooalign{\box\tw@ \cr \box\z@}}}
\makeatother

\makeatletter
\newcommand{\leqnos}{\tagsleft@true\let\veqno\@@leqno}
\newcommand{\reqnos}{\tagsleft@false\let\veqno\@@eqno}
\reqnos
\makeatother

\newcommand{\triple}[1]{{\left\vert\kern-0.25ex\left\vert\kern-0.25ex\left\vert #1 
    \right\vert\kern-0.25ex\right\vert\kern-0.25ex\right\vert}}

\newcommand{\dd}{\mathrm{d}}

\newcommand{\R}{\mathbb{R}}
\newcommand{\RR}{\mathcal{R}}
\newcommand{\N}{\mathbb{N}}
\newcommand{\tr}{\mathrm{tr}}
\newcommand{\op}{\mathrm{op}}

\newcommand{\veps}{\varepsilon}
\newcommand{\W}{\mathcal{W}}

\newcommand{\U}{\mathcal{U}}
\newcommand{\F}{\mathcal{F}}
\newcommand{\cN}{\mathcal{N}}

\newcommand{\h}{\mathfrak{h}}

\def\tr{\mathrm{tr}}

\def\bR{\mathbb{R}}
\def\bN{\mathbb{N}}

\def\cH{\mathcal{H}}
\def\cN{\mathcal{N}}

\newtheorem{theorem}{Theorem}[section] % reset theorem numbering for each chapter
\newtheorem{proposition}[theorem]{Proposition} 
\newtheorem{corollary}[theorem]{Corollary}
\newtheorem{lemma}[theorem]{Lemma}

\newtheorem{remark}[theorem]{Remark}
\newtheorem{assumption}[theorem]{Assumption}

\title{Effective Dynamics of Local Observables for Extended Fermi Gases in the High-Density Regime}

\author[1]{Luca Fresta}
\author[2]{Marcello Porta}
\author[3]{Benjamin Schlein}
\affil[1]{Hausdorff Center for Mathematics,
University of Bonn, Endenicher Allee 60
53115 Bonn, Germany}
\affil[2]{SISSA, Via Bonomea 265, 34136 Trieste, Italy}
\affil[3]{Institute of Mathematics, University of Zurich, Winterthurerstrasse 190, 8057 Zurich, Switzerland}

\begin{document}

\maketitle

\begin{abstract}
We give a rigorous derivation of the Hartree equation for the many-body dynamics of pseudo-relativistic Fermi systems at high density $\varrho \gg 1$, on arbitrarily large domains, at zero temperature. With respect to previous works, we show that the many-body evolution can be approximated by the Hartree dynamics locally, proving convergence of the expectation of observables that are supported in regions with fixed volume, independent of $\varrho$. The result applies to initial data describing fermionic systems at equilibrium confined in arbitrarily large domains, under the assumption that a suitable local Weyl-type estimate holds true. The proof relies on the approximation of the initial data through positive temperature quasi-free states, that satisfy strong local semiclassical bounds, which play a key role in controlling the growth of the local excitations of the quasi-free state along the many-body dynamics.
\end{abstract}

\tableofcontents     

\section{Introduction}

We are interested in the time evolution of extended Fermi gases, at high density $\varrho$. We consider systems of $N$ fermions with a relativistic dispersion relation (pseudo-relativistic fermions), interacting through a smooth, rapidly decaying, two-body potential $V : \bR^3 \to \bR$. The Hamilton operator generating the dynamics has the form 
\begin{equation}\label{eq:ham0} H_{N} = \sum_{j=1}^N \sqrt{1-\veps^2 \Delta_{x_j}} + \veps^3 \sum_{i<j}^{N} V (x_i -x_j) \end{equation} 
and, according to fermionic statistics, it acts on the Hilbert space $L^2_a (\bR^{3N})$, the subspace of $L^2 (\bR^{3N})$ consisting of functions that are antisymmetric with respect to permutations. In (\ref{eq:ham0}) and throughout the paper, we shall choose $\veps = O(\varrho^{-1/3})$. This choice guarantees that, for physically relevant states, both kinetic and potential energy per particle are of order one, in the limit of large $\varrho$. 

Initially, the $N$ particles are trapped by an external potential in a set $\Lambda$ with volume $|\Lambda| = N / \varrho$. At zero temperature, the Fermi gas relaxes into the ground state of the trapped Hamiltonian. At high density, we can expect that the many-body interaction can be effectively replaced by an averaged one-body potential and that the ground state of the trapped Hamiltonian can be approximated by a Slater determinant, minimizing the corresponding Hartree-Fock energy (or the reduced Hartree-Fock energy, since the exchange term is expected to be subleading in the limit of large $\varrho$). 

Motivated by these observations, we are going to study the solution of the many-body Schr\"odinger equation 
\begin{equation}\label{eq:schr} i \veps \partial_t \psi_t = H_{N} \psi_t =  \Big[  \sum_{j=1}^N \sqrt{1-\veps^2 \Delta_{x_j}} + \veps^3 \sum_{ i<j }^{N} V (\hat x_i - \hat x_j) \Big] \psi_t \end{equation} 
for initial data $\psi_{t=0}$ that are close to a Slater determinant, with reduced one-particle density matrix having the form $\omega_\mu = \chi (H \leq \mu)$, for a one-particle Hamiltonian $H = \sqrt{1-\veps^2 \Delta} + V_\text{ext}$ (where the external potential can also include the contribution of the direct term in the interaction) and with the chemical potential $\mu \in \bR$ chosen so that $\tr  \,\omega_{\mu} = N$. Our goal will be to show that, for large $\varrho$, the solution of (\ref{eq:schr}) remains close to a Slater determinant, with reduced one-particle density matrix evolved through the time-dependent Hartree equation
\begin{equation}\label{eq:hartree0} i\veps \partial_t \omega_t = \big[ \sqrt{1-\veps^2 \Delta} + \veps^{3} (V * \varrho_t) , \omega_t \big] 
\end{equation} 
with $\varrho_t (x) = \omega_t (x;x)$ and with initial datum $\omega_{t=0} = \omega_\mu$. 

In the last decades, there have been substantial efforts in the mathematical physics community to show that the many-body Schr\"odinger evolution of Fermi gases can be approximated by the Hartree dynamics. Most results have been obtained for particles with non-relativistic dispersion in the mean-field regime, where $\varrho = N$ (or equivalently, where the $N$ particles are initially trapped in a volume of order one, so that each particle interacts with all other particles in the system). In this setting, closeness to the Hartree evolution was first established in \cite{EESY} for analytic interaction potentials and for short times (convergence to the Vlasov equation, which approximates the Hartree evolution in the semiclassical limit, has been  known since \cite{Sp,NS}). This result has been extended to a larger class of regular interaction potentials and to arbitrary fixed times in \cite{BPS} (and later in \cite{PP}), for initial data describing Slater determinants (or perturbations thereof) with reduced one-particle density matrices satisfying certain semiclassical commutator estimates. For mixed quasi-free initial data, which are relevant at positive temperature, convergence towards Hartree dynamics in the mean-field regime was established in \cite{BJPSS}, again under the assumption that the initial data exhibit an appropriate semiclassical structure. Results for singular interactions have been later obtained in \cite{PRSS,CLS}. A norm approximation for the many-body dynamics of a homogeneous Fermi gas in terms of a quasi-free bosonic evolution has been obtained in \cite{BNPSS}, using rigorous bosonization ideas developed to study the correlation energy of mean-field fermions in \cite{BNPSSa,BNPSSb,BPSScorr} (for an alternative approach to the correlation energy, see also \cite{CHN,CHN2,CHN3}). For mean-field fermions with a relativistic dispersion, whose many-body evolution is described by (\ref{eq:schr}), with $\veps = N^{-1/3}$, the Hartree equation (\ref{eq:hartree0}) was derived in \cite{BPS2}, adapting the ideas of \cite{BPS}. Other mean-field type scalings have been considered in \cite{BGGM,FK,BBPPT}. 

It is a natural question to understand how to extend these results to the thermodynamic setting. In \cite{FPS}, we recently established convergence towards Hartree dynamics also for extended Fermi gases, where $\varrho$ is large but now independent of $N$. In that work, we considered fermions with both non-relativistic and relativistic dispersions. In the relativistic case, we studied the solution of (\ref{eq:schr}) for many-body initial data close to a Slater determinant, with reduced one-particle density matrix $\omega$ satisfying the local trace-norm bounds 
\begin{equation}\label{eq:semi} 
\begin{split} 
\Big\| \frac{1}{1+ |\hat x - z|^{4n}} \omega \Big\|_\text{tr} \leq C &\veps^{-3} , \qquad \Big\| \frac{1}{1+|\hat x - z|^{4n}} [ e^{i p \cdot \hat{x} }, \omega ] \Big\|_\text{tr}  \leq C \veps^{-2} (1+|p|),\\& \Big\| \frac{1}{1+|\hat x - z|^{4n}} [ \veps \nabla , \omega ] \Big\|_\text{tr} \leq C \veps^{-2} , \end{split} 
\end{equation} 
for all $z \in \bR^3$ and $n$ large enough. While the first estimate in (\ref{eq:semi}) implies that the local density of particles around the point $z \in \bR^3$ is at most of order $\veps^{-3}$, for all $z \in \bR^3$, the other two bounds guarantee that $\omega$ exhibits a local semiclassical structure, in the sense that its integral kernel $\omega (x;y)$ varies on the length scale $\veps$, in the $(x-y)$ direction, and on a scale of order one, in the $(x+y)$ direction. Assuming the bounds (\ref{eq:semi}) to hold at time $t=0$ and propagating them along the Hartree dynamics (\ref{eq:hartree0}) (showing, in other words, that the solution $\omega_t$ of (\ref{eq:hartree0}) still satisfies (\ref{eq:semi}), though with a worse constant $C_t = C \exp (c |t|)$), in \cite{FPS} we proved that the solution $\psi_t$ of (\ref{eq:schr}) remains close to a Slater determinant, with reduced one-particle density matrix determined by the Hartree equation (\ref{eq:hartree0}). More precisely, if $\gamma^{(1)}_{t}$ denotes the one-particle reduced density matrix associated with $\psi_t$, we showed that 
\begin{equation}\label{eq:FPS} \| \gamma^{(1)}_t - \omega_t \|_\text{HS} \leq C \exp (c \exp (c |t|)) \, \veps^{1/2}  \sqrt{N}  \end{equation} 
for all $t \in \bR$, which is small if compared with $\| \gamma^{(1)}_t \|_\text{HS} , \| \omega_t \|_\text{HS} \simeq \sqrt{N}$.

Equation (\ref{eq:FPS}) establishes convergence towards the Hartree dynamics in a global sense. Namely, it shows that the total number of excitations of the Slater determinant at time $t \in \bR$ is small, compared to the total number of particles $N$: while the number of excitations is still proportional to $N$, the ratio tends to zero as $\veps \to 0$. A more refined question is whether it is possible to establish convergence towards the Hartree dynamics locally, when testing against a local observable. This is the question we address in the present paper, giving a positive answer. In our main theorem, Theorem \ref{thm:main} below, we consider the expectation in the state $\psi_t$, solution of (\ref{eq:schr}), of one-particle observables $\sum_{j=1}^N \mathcal{O}_z (\hat{x}_j)$, where $\mathcal{O}_{z}(x)$ is a smooth and fast decaying function centered at $z$, essentially measuring the number of fermions in a region with volume of order one around the point $z \in \bR^3$, and we prove that it is close to the expectation in the Slater determinant with reduced density $\omega_t$, solving (\ref{eq:hartree0})  with initial datum $\omega_{t=0} = \omega_\mu$. More precisely, we show that: 
\begin{equation}\label{eq:mainres}  \big| \tr \, \mathcal{O}_{z} (\gamma^{(1)}_t - \omega_t) \big| \leq C \exp (c \exp (c|t|)) \veps^{-3+\delta} \end{equation} 
for some $0<\delta \leq 1$, which is again small compared to $\tr \, \mathcal{O}_{z} \gamma^{(1)}_t , \tr \, \mathcal{O}_{z} \omega_t \simeq \veps^{-3}$. Remarkably, this result only requires $\psi_0$ to be close to the Slater determinant with reduced density $\omega_0$ in a ball centered at $z$ with radius of order $\veps^{-\delta} |t|$.

The proof of (\ref{eq:mainres}) is more involved than the proof of (\ref{eq:FPS}) obtained in \cite{FPS}. In particular, it requires stronger control on the local density and on the local semiclassical structure of the solution $\omega_t$ of the Hartree equation, which translates into stronger assumptions on the initial data. Instead of (\ref{eq:semi}), we need, roughly speaking, bounds of the form 
\begin{equation}\label{eq:semi2} 
\begin{split} 
\Big\| \frac{1}{1+| \hat x - z |^{4n}} \omega  \frac{1}{1+| \hat x - z' |^{2n}} \Big\|_{\text{tr}} &\leq \frac{C \veps^{-3}}{1+|z-z'|^{2n - 4}}\;,\\
\Big\|  \frac{1}{1+| \hat x - z |^{4n}} \big[ e^{i p \cdot \hat{x}} , \omega \big]  \frac{1}{1+| \hat x - z' |^{2n}} \Big\|_\text{tr}  &\leq \frac{C \veps^{-2} (1+|p|)}{1+|z-z'|^{2n}}\;, \\  
\Big\|  \frac{1}{1+| \hat x - z |^{4n}} \big[ \veps \nabla , \omega \big]  \frac{1}{1+| \hat x - z' |^{2n}} \Big\|_\text{tr} &\leq \frac{C \veps^{-2}}{1+|z-z'|^{2n}} ,
\end{split}
\end{equation}
for $n$ large enough and for $z_1, z_2 \in \bR^3$, capturing also the decay of correlations in space. Because of the lack of regularity of the characteristic function, it is not clear whether the initial Slater determinants we are interested in, with one-body reduced density matrix of the form $\omega_\mu = \chi (H \leq \mu)$, satisfies (\ref{eq:semi2}). To circumvent this issue, we use an approximation argument, approximating the Slater determinant described by $\omega_\mu$ by the quasi-free thermal state with reduced density $\omega_{\mu, \beta} =1/ (1+e^{\beta (H-\mu)})$, at inverse temperature $\beta = O (\veps^{-1})$.
The proof of (\ref{eq:mainres}) consists therefore of three main parts. First, in Theorem \ref{thm:T>0}, we prove that the many-body evolution of a mixed, approximately quasi-free state, whose reduced density satisfies (\ref{eq:semi2}), fulfills (\ref{eq:mainres}). Second, in Proposition \ref{prop:localsc}, we show that $\omega_{\mu,\beta}$ satisfies the strong local bounds (\ref{eq:semi2}); more precisely, since we are forced to work with mixed states, we need some additional estimates, listed in Assumption \ref{ass:sc}, controlling also the semiclassical structure of $\sqrt{\omega}$ and $\sqrt{1-\omega}$, and the overlap $\sqrt{\omega} \sqrt{1-\omega}$. This part of our analysis relies on a suitable assumption on the one-body Hamiltonian $H = \sqrt{1-\veps^2 \Delta} + V_\text{ext}$, which can be viewed as a local version of the Weyl law, and which is expected to hold true for a large class of confining external potentials $V_\text{ext}$, relevant for extended Fermi gases. Finally, in Section \ref{sec:proofmain}, we show that, for $\beta = O (\veps^{-1})$, the many-body evolution of initial data that are close to a Slater determinant with reduced density $\omega_\mu$ can  be approximated, in a local sense, by the many-body evolution of an approximately quasi-free mixed state with reduced density $\omega_{\mu,\beta}$ (and that the solution of the Hartree equation with initial datum $\omega_\mu$ can be approximated by the solution with initial datum $\omega_{\mu,\beta}$). 

In \cite{FPS}, where we established global convergence towards Hartree dynamics for extended Fermi gases, we also considered the case of particle with non-relativistic dispersion. In fact, the result we prove in \cite{FPS} for non-relativistic fermions was a bit weaker, as it only holds for sufficiently short times. In the current paper, we focus only on the relativistic case. Establishing local convergence (in the sense of (\ref{eq:mainres})) for non-relativistic particles is a more challenging problem, and we do not address it here. The reason why the non-relativistic problem is more difficult is ultimately due to the unbounded group velocity of the particles, which makes it hard to control the solutions of the Hartree equation and of the many-body dynamics locally. This is also the reason for the short-time limitation of the result about the non-relativistic case in \cite{FPS}: because of the unboundedness of the velocity of the particles, it is difficult to rule out the excessive concentration of a large number of particles in a small region of space. Concretely, when computing the variation of the expectation value of the local observable $\mathcal{O}_z$, we have to control its commutator with the dispersion law of the particles. For a relativistic dispersion, this commutator can be controlled by another bounded observable, localized close to $z$. In the non-relativistic case, on the other hand, we find $[ \Delta , \mathcal{O}_z ] = \nabla \cdot \nabla \mathcal{O}_z + \nabla \mathcal{O}_z \cdot \nabla$ which is not bounded. The absence of local conservation laws makes it difficult to control this contribution, and the conservation of the total energy is not useful to control local quantities. 

The paper is organized as follows. In Section \ref{sec: effective-zero-temp} we introduce the fermionic Fock space, the formalism of second quantization and the fermionic Bogoliubov transformations, and we state our main result, Theorem \ref{thm:main}. The theorem describes the evolution of initial data close to Slater determinants defined as ground states of the second quantization of a one-particle Hamiltonian $H$, satisfying a local Weyl-type estimate, as stated in Assumption \ref{ass:Weyl}. In Section~\ref{sec:T>0} we introduce positive-temperature states, which will be needed as a regularization of the initial zero-temperature state. We will describe them as pure states on a doubled Fock space, via the Araki-Wyss representation of mixed quasi-free states. The main technical result of this section is Theorem \ref{thm:T>0}, which proves convergence from many-body dynamics to the Hartree evolution for initial data at temperature $O(\varepsilon)$, exhibiting the local semiclassical structure of Assumption \ref{ass:locsc}. In Section \ref{sec:Tproof} we prove Theorem \ref{thm:T>0}; the proof is based on the propagation of the local semiclassical structure, Proposition \ref{prop: semiclassical_prop-main}, and on the control of the growth of fluctuations between many-body and Hartree dynamics at a local scale, Proposition~\ref{thm:bound-localized-fluctuation}. Then, in Section \ref{sec:proofmain} we use the positive-temperature result of Theorem \ref{thm:T>0} to prove our main result, Theorem \ref{thm:main}, via an approximation argument. In Appendix \ref{sec: appen-apriori} we prove a priori bounds for the density of the Fermi-Dirac distribution in the semiclassical scaling. In Appendix \ref{app:gibbs} we prove the validity of the local semiclassical estimates of Assumption \ref{ass:locsc} for the Fermi-Dirac distribution associated with the Hamiltonian $H = \sqrt{1 - \varepsilon^{2}\Delta} + V_{\text{ext}}$ at temperature $O(\varepsilon)$, as a consequence of Assumption \ref{ass:Weyl}. Finally, in Appendix \ref{app:DG} we discuss the relation of Assumption \ref{ass:Weyl} with other known results in semiclassical analysis, such as the sharp version of the pointwise Weyl law.

\paragraph{Acknowledgements.} The work of L. F. has been supported by the German Research Foundation (DFG) under Germany's Excellence Strategy - GZ 2047/1, Project-ID 390685813, and under SFB 1060 - Project-ID 211504053. M. P. acknowledges financial support by the European Research Council through the ERC-StG MaMBoQ, n. 802901. B. S. acknowledges financial support from the Swiss National Science Foundation through the Grant “Dynamical and energetic properties of Bose-Einstein condensates”, from the NCCR SwissMAP and from the European Research Council through the ERC-AdG CLaQS. The work of M. P. has been carried out under the auspices of the GNFM of INdAM. L.~F.~and M.~P.~gratefully acknowledge hospitality from the University of Z\"urich. We thank Gaultier Lambert for useful discussions on local semiclassical estimates.

\section{Effective dynamics of zero-temperature states}
\label{sec: effective-zero-temp}

\subsection{Fock space formalism}
\label{sec: second-quantisation}

We henceforth set $\mathfrak{h}:=L^{2}(\R^{3})$. The antisymmetric (or fermionic) Fock space $\mathcal{F}(\mathfrak{h})$ associated with $\mathfrak{h}$ is defined as:
\begin{equation*}
\mathcal{F}(\mathfrak{h}):= \mathbb{C}  \oplus \bigoplus _{n\geq 1} \mathfrak{h}^{\wedge n} \;,
\end{equation*}
where, for $n\in \mathbb{N}$, $\mathfrak{h}^{\wedge n}:=\mathfrak{h} \wedge \cdots \wedge \mathfrak{h}$ denotes the Hilbert space given by the $n$-fold antisymmetric tensor product of $\mathfrak{h}$. Vectors in $\mathcal{F}(\mathfrak{h})$ are sequences $\Psi = (\psi^{(n)})_{n \in \bN}$ with $\psi^{(0)} \in \mathbb{C}$ and $\psi^{(n)} \in \mathfrak{h}^{\wedge n}$ for $n \geq 1$. A distinguished vector is the vacuum $\Omega := (1,0,\dots,0,\dots)$. As a Hilbert space, $\mathcal{F}(\mathfrak{h})$ is equipped with the scalar product
\begin{equation*}
\langle \Psi_{1}, \Psi_{2}\rangle = \overline{\psi_{1}^{(0)}} \psi_{2}^{(0)} +  \sum_{n \in \mathbb{N}} \langle \psi_{1}^{(n)}, \psi_{2}^{(n)} \rangle_{\mathfrak{h}^{\otimes n}}
\end{equation*}
for any $\Psi_{1}= (\psi_{1}^{(n)})_{n \in \bN}, \Psi_{2}= (\psi_{2}^{(n)})_{n \in \bN} \in \mathcal{F}(\mathfrak{h})$. 

For $f \in \mathfrak{h}$, we introduce the creation operator $a^{*}(f)$ and the annihilation operator $a(f)$, whose actions on $\Psi = ( \psi^{(n)})_{n \in \bN}\in \mathcal{F}(\mathfrak{h})$ are defined by 
\[ \begin{split} 
(a^* (f) \Psi)^{(n)} (x_1, \dots ,x_n) &= \frac{1}{\sqrt{n}} \sum_{j=1}^n (-1)^j f(x_j) \psi^{(n-1)} (x_1, \dots , x_{j-1}, x_{j+1}, \dots , x_n) \\
(a(f) \Psi)^{(n)} (x_1, \dots , x_n) &= \sqrt{n+1} \int \bar{f} (x) \, \psi^{(n+1)} (x, x_1, \dots, x_n) dx \end{split} \]
for all $n \in \bN$. It is not difficult to see that these operators are the adjoint of each other, and that they satisfy the canonical anticommutation relations:
\begin{equation}\label{eq:CAR}
\{a(f), a(g) \} = \{a^{*}(f), a^{*}(g) \} = 0 \;,\qquad \{a^{*}(f), a(g) \} = \langle g,f \rangle_{\mathfrak{h}}\;,\qquad \forall f,g \in \mathfrak{h} \;,
\end{equation}
with $\{A,B \}:=AB+BA$. These relations imply that  creation and annihilation operators are bounded, 
\begin{equation*}
\|a(f) \|, \| a^* (f) \|  \leq \| f\|_{\mathfrak{h}} \;.
\end{equation*}
It is also convenient to introduce operator-valued distributions $a^{*}_{x}$, $a_{x}$, for $x \in \bR^3$. They allow us to write 
\begin{equation}\label{eq:opval}
a(f) = \int dx\, a_{x} \overline{f(x)}\;,\qquad a^{*}(f) = \int dx\, a^{*}_{x} f(x)\;.
\end{equation}
For any (closable) operator $O$ on $\frak{h}$, we denote by $d\Gamma(O)$ its second quantization, as an operator on $\mathcal{F} (\mathfrak{h})$. That is:
\begin{equation}
d\Gamma(O) = 0 \oplus \bigoplus_{n\geq 1} \Big(\sum_{j=1}^{n} 1^{\otimes j} \otimes O\otimes 1^{n-j-1}\Big)\qquad \text{on $\mathcal{F}(\mathfrak{h})$.}
\end{equation}
In terms of the creation and annihilation operators, we have:
\begin{equation}\label{eq:op-second-quantization}
d \Gamma(O) = \sum_{i,j} \langle f_i, O f_j  \rangle_{\frak{h}} a^{*}(f_{i}) a(f_{j})\;.
\end{equation}
If the operator $O$ has an integral kernel $O(x;y)$, we can also write 
\begin{equation*}
d\Gamma(O) = \int dxdy\, O(x;y) a^{*}_{x} a_{y}\;.
\end{equation*}
In this case, we also introduce the notation 
\[ d\Gamma^+ (O) = \int dx dy \, O (x;y) \, a^*_x a^*_y\;, \qquad d\Gamma^- (O) = \int dx dy \, O (x;y) \, a_x a_y\;. \]
Finally, on $\mathcal{F} (\mathfrak{h})$, we define the Hamilton operator
\begin{equation*}
\mathcal{H} = 0 \oplus \bigoplus_{n\geq 1} H_{n}\;,
\end{equation*}
where $H_{n}$ is the operator on $\frak{h}^{\wedge n}$ given by:
\begin{equation*}
H_{n} = \sum_{i=1}^{n} \sqrt{1-\varepsilon^{2}\Delta_{i}} + \varepsilon^{3} \sum_{i<j}^{n} V(\hat x_{i} - \hat x_{j})\;.
\end{equation*}
Recall that, as discussed in the introduction, we will consider initial data for which the density of particles $\varrho$ is of order $\varepsilon^{-3}$. Equivalently,
\begin{equation}\label{eq:calH} 
\mathcal{H} = d\Gamma(\sqrt{1-\varepsilon^{2}\Delta}) + \frac{\varepsilon^{3}}{2} \int dxdy\, V(x-y) a^{*}_{x} a^{*}_{y} a_{y} a_{x}\;.
\end{equation}
To simplify our analysis, we are going to assume that $V \in \mathcal{S} (\bR^3)$, with 
\[ \mathcal{S} (\bR^3) = \{ f \in C^\infty (\bR^3 ; \bR) : \sup_{x \in \bR^3} |x^\alpha \partial^\beta f (x)| < \infty \text{ for all multi-indices $\alpha, \beta \in \bN^3$} \} \]
denoting the space of real-valued Schwarz functions on $\bR^3$. 
 
An important class of vectors in $\mathcal{F} (\mathfrak{h})$ are Slater determinants, having the form
\begin{equation}\label{eq:Slater}
(f_{1} \wedge \cdots \wedge f_{n})(x_{1},\ldots, x_{n}) = \frac{1}{\sqrt{n!}}\sum_{\pi\in S_{n}} \text{sgn}(\pi) f_{1}(x_{\pi(1)}) \cdots f_{n}(x_{\pi(n)})\;,
\end{equation}
with $\{ f_{i} \}_{i=1}^n$ an orthonormal system in $\frak{h}$. The Slater determinant (\ref{eq:Slater}) describes a state with exactly $n$ particles; it can   also be written as $a^* (f_1) \dots a^* (f_n) \Omega$. If $\{ f_i \}_{i\in \bN}$ is an orthonormal basis for $\mathfrak{h}$, then the set $\{ f_{i_1} \wedge \cdots \wedge f_{i_n} : n \in \bN, i_1 < i_2 < \dots < i_n\}$ forms an orthonormal basis of $\mathcal{F} (\mathfrak{h})$. 

Slater determinants are quasi-free states; the expectation of an arbitrary product of creation and annihilation operator in the state (\ref{eq:Slater}) can be computed through Wick's theorem, using the one-particle reduced density matrix 
\begin{equation}\label{eq:omega0} \omega = \tr_{2, \dots , n} |f_{1} \wedge \cdots \wedge f_{n} \rangle \langle f_{1} \wedge \cdots \wedge f_{n}| = \sum_{i=1}^n |f_i \rangle \langle f_i| \end{equation} 
associated with $f_1 \wedge \dots \wedge f_n$, which is just the orthogonal projection onto the $n$-dimensional subspace of $\mathfrak{h}$ spanned by $f_1, \dots , f_n$. 

We will make use of the fact that Slater determinants in $\mathcal{F} (\mathfrak{h})$ can be generated by the action of unitary transformations on the vacuum. More precisely, given an arbitrary orthogonal projection of the form (\ref{eq:omega0}), we can find a unitary Bogoliubov transformation $R_\omega$ on $\mathcal{F} (\mathfrak{h})$ such that 
\begin{equation} \label{eq:Romega} 
R_\omega  \Omega = a^* (f_1) \dots a^* (f_n) \Omega \end{equation}
is the Slater determinant (\ref{eq:Slater}) and 
\begin{equation}\label{eq:Romega2} R^*_\omega a (f) R_\omega =  a (uf) + a^* (v \bar{f}) \end{equation} 
with $u = 1-\omega$ the projection on the orthogonal complement of the space spanned by $f_1, \dots , f_n$, and $v = \sum_{i=1}^n |\bar{f}_i \rangle \langle f_i|$.

\subsection{Main result}
\label{subsec: main result}

Our main theorem describes the many-body evolution of initial data that are close, in a local sense, to a Slater determinant describing the ground state of a non-interacting pseudo-relativistic Fermi gas trapped by an external potential, at high density. Let us specify more precisely the class of initial data we are going to consider. For a confining, non-negative and smooth potential $V_\text{ext}$, we consider the one-particle Hamiltonian 
\begin{equation}\label{eq:H-noint} H = \sqrt{1-\veps^2 \Delta} + V_\text{ext} \end{equation}
acting on the Hilbert space $\mathfrak{h}$. We will require precise local semiclassical properties of spectral functions of $H$ (describing equilibrium states of $d\Gamma (H)$). To derive these properties, we will assume a bound for the trace norm of the product of a function of $H$, localized in a small interval of size $O (\veps)$ around a fixed $\mu > 0$ (which will play the role of the chemical potential), and a function of the position operator $\hat{x}$, localized around a point $z \in \bR^3$. For $2n \in \N$, $z,x \in \R^{3}$, we define the spatial weight 
\begin{equation}\label{eq: polynomial weight}
\W_{z}^{(n)}(x):=\big(1+ |x-z|^{4n}\big)^{-1} \; .
\end{equation}
Note that the functions $\W_{z}^{(n)}$ satisfy $\W_{z}^{(m)} \leq 2\W_{z}^{(n)}$ for all $m>n$, and this will be repeatedly used throughout this work. Another property we shall frequently use is that $\W_{z_{1}}^{(n)}(x) \W_{z_{2}}^{(n)}(x) \leq C_{n} \W^{(n)}_{z_{1}}(z_{2})$.

Moreover, we introduce the notation $\W_{z}^{(n)} = \W_{z}^{(n)}( \hat{x} )$ for the multiplication operator on $\mathfrak{h}$ associated with (\ref{eq: polynomial weight}).  
\begin{assumption}[Assumption on $(H,\mu)$]\label{ass:Weyl} Let $\beta = O (\varepsilon^{-1})$ and let $H$ be defined as in (\ref{eq:H-noint}), with $V_\text{ext} \in C^\infty (\bR^3)$, $V_\text{ext} \geq 0$ and $\sup_{x \in \bR^3} |\partial^\alpha V_\text{ext} (x)| < \infty$, for all $\alpha \in \bN^3$ with $|\alpha | \geq 2$ (i.e. $V_\text{ext}$ must grow at most quadratically at infinity). We assume that there exists a domain $\Lambda \subset \R^{3}$ such that 
\begin{equation}\label{eq:localWeyl}
\sup_{z \in \Lambda} \Big\| \frac{1}{(\beta (H - \mu))^{2m} + 1} \mathcal{W}^{(n)}_{z} \Big\|_{\tr} \leq C\varepsilon^{-2}
\end{equation}
for all $n,m \in \bN$ large enough. 
\end{assumption}
\begin{remark} 
\item[(i)] It is not difficult to see that the bound (\ref{eq:localWeyl}) holds true for $V_{\text{ext}} = 0$. 
\item[(ii)] A similar estimate is known to hold for $V_{\text{ext}}\neq 0$, for non-relativistic fermions; see {\it e.g.} \cite{DG}, replacing the localized functions of $H$ and of $\hat x$ by smooth compactly supported functions. We believe that the estimates could be extended to the pseudorelativistic setting considered here. In general, the constant $C$ depends on the potential $V_{\text{ext}}$. It would be interesting to show that the constant can be made independent of the size of the classically confined region $\{x \in \mathbb{R}^{3}: V_{\text{ext}}(x) - \mu \leq 0\}$, which is arbitrarily large for extended systems.
\item[(iii)] In Appendix \ref{app:DG}, we will show that the estimate (\ref{eq:localWeyl}) is actually implied by an analogous bound in Hilbert-Schmidt norm. Furthermore, in Appendix \ref{app:DG} we also show that the Hilbert-Schmidt bound is implied by the validity of the pointwise sharp Weyl law, which is known to hold for non-relativistic fermions, see {\it e.g.} \cite{DG0}. 
\end{remark}

We are now ready to state our main result.

\begin{theorem}[Derivation of the Hartree Equation on a Local Scale]\label{thm:main} Let $(H,\mu)$ satisfy Assumption \ref{ass:Weyl} in a bounded domain $\Lambda \subset \bR^3$. Let $\omega_\mu =  \chi (H \leq \mu)$ denote the Fermi projection associated with $H$ and with the chemical potential $\mu$. We consider the Fock space vector $\psi = R_{\omega_\mu} \xi$, where $R_{\omega_\mu}$ is the Bogoliubov transformation defined in (\ref{eq:Romega}), (\ref{eq:Romega2}), generating the Slater determinant with reduced density $\omega_\mu$, and where we assume that the excitation vector $\xi \in \mathcal{F} (\mathfrak{h})$ is such that 
\begin{equation}\label{eq:assxi}
\sup_{z \in \R^{3} }\langle \xi, d\Gamma(\mathcal{W}^{(n)}_{z}) \xi \rangle \leq C\varepsilon^{-3}\;,
\qquad
\sup_{z \in \Lambda }\langle \xi, d\Gamma(\mathcal{W}^{(n)}_{z}) \xi \rangle \leq C\varepsilon^{-3+\delta}\; , 
\end{equation}
for a constant $C > 0$, independent of $\Lambda$, for some $0< \delta \leq 1$ and for all $n \in \bN$ large enough. Let $\cH$ denote the Hamilton operator (\ref{eq:calH}) on $\mathcal{F} (\frak{h})$, with $V \in \mathcal{S} (\bR^3)$, consider the evolution $\psi_{t} = e^{-i\mathcal{H} t/\varepsilon} \psi$, and let $\gamma^{(1)}_{t}$ denote the reduced one-particle density matrix of $\psi_{t}$. Let $\omega_{t}$ be the solution of the time-dependent Hartree equation 
\begin{equation}\label{eq:hartree-main} 
i \veps \partial_{t} \omega_{t} = \Big[ \sqrt{1-\varepsilon^{2} \Delta } + \varepsilon^{3} (V * \varrho_{t})\,,\, \omega_{t} \Big]\;,  
\end{equation} 
with $\varrho_t (x) = \omega_t (x;x)$ and with initial datum $\omega_{t=0} = \omega_\mu$. 
 Let $\mathcal{O}\in L^{1}(\mathbb{R}^{3})$ such that
\begin{equation}\label{eq: assumptions-O}
\big\| (1 + |\cdot|) D^{\alpha} \hat{\mathcal{O}}\big \|_{1} \leq C ,
\end{equation}
for $|\alpha| \leq \ell$ and with $\ell\in \mathbb{N}$ large enough, and let $\mathcal{O}_{z}(x) = \mathcal{O}(x - z)$. Then, there exist constants $C,c > 0$,  independent of $\Lambda$, such that 
\begin{equation}\label{eq:mainT>0}
\Big| \tr\, \mathcal{O}_{z} ( \gamma^{(1)}_{t} - \omega_{t} ) \Big| \leq C   \exp(c \exp(c |t|)) \, \veps^{-3+\delta} 
\end{equation}
for all $n\in \bN$ large enough, for all $t\in \R$, and for all $z \in \Lambda$ such that $B_{\veps^{-\delta} |t|}(z) \subset \Lambda$.  
\end{theorem}
The proof of Theorem \ref{thm:main} is based on an approximation argument of the zero-temperature pure state with reduced density $\omega_{\mu}$ by a mixed state with reduced density given by the Fermi-Dirac distribution $\omega_{\mu, \beta} = 1/(1+e^{\beta (H-\mu)})$, at low temperature $T=1/\beta = O(\varepsilon)$. Informally, the reason for introducing this approximation is that at positive temperature the kernel of the reduced one-particle density matrix has better decay properties; in particular, the improved locality of the positive temperature state, combined with Assumption \ref{ass:Weyl}, will allow us to show that $\omega_{\mu,\beta}$ exhibit a local semiclassical structure, which is preserved by the Hartree equation and which will play a crucial role to prove local convergence towards the Hartree dynamics. At the same time, we will prove that, for $1/\beta = O(\varepsilon)$, the positive-temperature state and the zero-temperature state are close in the sense of expectation of local observables.

\section{Effective dynamics of low-temperature states}\label{sec:T>0}

In this section, we establish convergence towards the Hartree dynamics for initial data approximating equilibrium states at low but positive temperature $T = O(\varepsilon)$. We formulate our result in the general setting of mixed states. We will be interested in the evolution of (approximately) quasi-free mixed initial data. To describe such initial data, we will use the Araki-Wyss representation \cite{AW}, switching to a doubled Fock space. To introduce these tools, we will follow the discussion of \cite{BJPSS}. 

\subsection{Mixed states}
\label{sec: mixed-states}
A density matrix $\rho$ is a non-negative trace class operator on $\mathcal{F}(\frak{h})$, with $\tr_{\mathcal{F}(\frak{h})}\, \rho = 1$. Through the spectral theorem, it is always possible to write the density matrix $\rho$ as 
\begin{equation*}
\rho = \sum_{n} \lambda_{n} |\psi_{n} \rangle \langle \psi_{n} |\;,
\end{equation*}
with $\lambda_{n} \geq 0$, $\sum_{n} \lambda_{n} = 1$ and with an appropriate orthonormal sequence $\psi_n \in \mathcal{F} (\frak{h})$. The expectation value of a bounded operator $\mathcal{O}$ on $\mathcal{F}(\mathfrak{h})$ in the state $\rho$ is determined by 
\begin{equation*}
\tr_{\mathcal{F}(\mathfrak{h})} \, \mathcal{O} \rho = \sum_{n} \lambda_{n} \langle \psi_{n}, \mathcal{O} \psi_{n} \rangle\;.
\end{equation*}
The density matrix $\rho$ is said to describe a pure state if $\rho = |\psi \rangle \langle \psi |$, for a normalized $\psi \in \mathcal{F} (\frak{h})$. In this case, $\rho$ provides an equivalent representation of the state described by the vector $\psi$. If $\rho$ is not an orthogonal projection, but a convex combination of orthogonal projections, then it is said to describe a mixed state.  

Given the density matrix $\rho$, we define the operator $\widetilde{\kappa}$ on $\mathcal{F}(\frak{h})$ as
\begin{equation*}
\widetilde{\kappa} = \sum_{n} \varepsilon_{n} |\psi_{n} \rangle \langle \phi_{n} |\;,
\end{equation*}
where $|\varepsilon_{n}|^{2} = \lambda_{n}$ and $\{ \phi_{n} \}_{n}$ is another arbitrary orthonormal set. Clearly,
\begin{equation*}
\widetilde{\kappa} \widetilde{\kappa}^{*} = \rho\;. 
\end{equation*}
Observe that the set of Hilbert--Schmidt operators on $\mathcal{F}(\frak{h})$, denoted by $\mathcal{L}^{2}(\mathcal{F}(\frak{h}))$, is isomorphic to $\mathcal{F} \otimes \mathcal{F}$, the isomorphism being simply defined by $|\psi \rangle \langle \phi |\to \psi \otimes \overline{\phi}$ and extended by linearity to the whole $\mathcal{L}^{2}(\mathcal{F}(\frak{h}))$. Therefore, the mixed state with density matrix $\rho$ is described on $\mathcal{F}(\frak{h}) \otimes \mathcal{F}(\frak{h})$ by the vector:
\begin{equation}\label{eq:kappadef}
\kappa = \sum_{n} \varepsilon_{n} \psi_{n} \otimes \overline{\phi}_{n}\;.
\end{equation}
Thus, the expectation value of the operator $\mathcal{O}$ on $\mathcal{F}(\frak{h})$ can be rewritten as:
\begin{equation}\label{eq:exp-kappa}
\tr_{\mathcal{F}(\mathfrak{h})}\, \mathcal{O} \rho = \tr_{\mathcal{F}(\mathfrak{h})}\, \mathcal{O}\widetilde{\kappa} \widetilde{\kappa}^{*} = \langle \kappa, (\mathcal{O} \otimes 1) \kappa \rangle_{\mathcal{F}(\frak{h}) \otimes \mathcal{F}(\frak{h})}\;.
\end{equation}
We further observe that the doubled Fock space $\mathcal{F}(\frak{h}) \otimes \mathcal{F}(\frak{h})$ is isomorphic to the Fock space $\mathcal{F}(\frak{h} \oplus \frak{h})$; the unitary $U$ that realizes the isomorphism is called the exponential law, and it is defined by the relations
\begin{equation*}
U ( \Omega_{\mathcal{F}(\frak{h})} \otimes \Omega_{\mathcal{F}(\frak{h})} ) = \Omega_{\mathcal{F}(\frak{h} \oplus \frak{h})}
\end{equation*}
and
\begin{equation*}
U (a(f) \otimes 1)U^{*} = a(f\oplus 0) =: a_{l}(f)\;,\qquad U( (-1)^{\mathcal{N}} \otimes a(f) ) U^{*} = a(0\oplus f) =: a_{r}(f)\; ,
\end{equation*}
where $\cN = d\Gamma (1)$ is the number of particles operator on $\mathcal{F} (\frak{h})$. We call the operators $a_{l}(f)$ and $a_{r}(f)$ respectively the left and right representation of $a(f)$. We also denote $a^{*}_{\sigma}(f) = (a_{\sigma}(f))^{*}$, $\sigma = l,r$. We will also use the representation of the creation and annihilation operators in terms of the operator-valued distributions $a^{*}_{\sigma,x}$, similarly to (\ref{eq:opval}). Furthermore, it is useful to introduce the left and right representations of the second quantization of an operator $O$ on $\frak{h}$, defining 
\begin{equation}\label{eq:sec-U}
\begin{split}
d\Gamma_{l}(O) :=U ( d \Gamma(O) \otimes 1 ) U^{*} &= d\Gamma(O\oplus 0) \\
d\Gamma_{r}(O):=U ( 1\otimes d \Gamma(O)) U^{*} &= d\Gamma(0\oplus O) \;.
\end{split}
\end{equation}
If the operator $O$ has an integral kernel, we can write 
\begin{equation*}
d\Gamma_{\sigma}(O) = \int dxdy\, O(x;y) a^{*}_{x,\sigma} a_{y,\sigma}\;,\qquad \sigma = l,r\;.
\end{equation*}
Furthermore, we define 
\[ d\Gamma^+_{\sigma \sigma'} (O) = \int dx dy \, O (x;y) a^*_{\sigma ,x} a_{\sigma' , y}^* , \qquad d\Gamma^-_{\sigma \sigma'} (O) = \int dx dy \, O (x;y) a_{\sigma ,x} a_{\sigma' , y} \]

If $O$ is a trace-class operator over $\frak{h}$, then $d\Gamma_\sigma (O), d\Gamma^+_{\sigma \sigma'} (O), d\Gamma^-_{\sigma \sigma'} (O)$ are bounded operators. The proof of the following lemma is similar to the proof of \cite[Lemma 3.1]{BPS}. 
\begin{lemma} \label{lemma: bound-Fock-operator}
Let $J$ be a trace-class operator on $\mathfrak{h}$. Then, for any $\sigma,\sigma'= r,l$ the following bounds hold true:
\begin{equation*}
\big\|d\Gamma_{\sigma} (J) \big\| \leq 2  \left\| J \right\|_{\tr}  \;,\qquad \big\| d \Gamma^{\pm}_{\sigma \sigma'} (J) \big\|  \leq 2 \left\| J \right\|_{\tr}  \;.
\end{equation*}
\end{lemma}

Setting $\Psi = U \kappa \in \mathcal{F} (\frak{h} \oplus \frak{h})$, with $\kappa$ as in (\ref{eq:kappadef}), we can write, from (\ref{eq:exp-kappa}),  
\[ \tr_{\mathcal{F}(\frak{h})} \mathcal{O} \rho = \langle \Psi, \big[ U (\mathcal{O} \otimes 1) U^* \big] \Psi \rangle_{\mathcal{F} (\frak{h} \oplus \frak{h})}\;. \]
This allows us to represent the mixed state $\rho$ on $\mathcal{F} (\frak{h})$ as a pure state on $\mathcal{F} (\frak{h} \oplus \frak{h})$. In particular, the one-particle reduced density and the pairing density of $\rho$, which allow us to compute the expectation of observables that are quadratic in creation and annihilation operators, can be expressed in terms of $\Psi$, by 
\begin{equation}\label{eq:1partPsi} 
\begin{split} 
\tr_{\mathcal{F} (\frak{h})} a_y^* a_x \, \rho & = \langle \Psi, a_{l,y}^* a_{l,x} \Psi \rangle_{\mathcal{F} (\frak{h} \oplus \frak{h})} =: \gamma^{(1)}_\Psi (x;y) \\
\tr_{\mathcal{F} (\frak{h})} a_y a_x \, \rho & = \langle \Psi, a_{l,y} a_{l,x} \Psi \rangle_{\mathcal{F} (\frak{h} \oplus \frak{h})} =: \alpha^{(1)}_\Psi (x;y)  \;.\end{split} 
\end{equation} 
From (\ref{eq:sec-U}), we obtain 
\begin{equation} \label{eq:dG} 
\begin{split}  \tr_{\mathcal{F} (\frak{h})} d\Gamma (O) \rho &= \langle \Psi, d\Gamma_l (O) \Psi \rangle_{\mathcal{F} (\frak{h} \oplus \frak{h})} = \tr_\frak{h} \, O \gamma^{(1)}_\Psi \\ 
\tr_{\mathcal{F} (\frak{h})} d\Gamma^- (O) \rho &= \langle \Psi, d\Gamma^-_{ll} (O) \Psi \rangle_{\mathcal{F} (\frak{h} \oplus \frak{h})} = \tr_\frak{h} \, O \alpha^{(1)}_\Psi \\
 \tr_{\mathcal{F} (\frak{h})} d\Gamma^* (O) \rho &= \langle \Psi, d\Gamma^+_{ll} (O) \Psi \rangle_{\mathcal{F} (\frak{h} \oplus \frak{h})} = \overline{\tr_\frak{h} \, O \alpha^{(1)}_\Psi} \;.
\end{split} 
\end{equation} 

To conclude this section, let us discuss the time evolution of mixed states, after the mapping into the doubled Fock space. The time evolution of the density matrix $\rho$ is:
\begin{equation*}
\rho_{t} = e^{-i\mathcal{H} t / \varepsilon} \rho \, e^{i \mathcal{H} t /\varepsilon}\;,
\end{equation*}
with the Fock-space Hamiltonian $\mathcal{H}$ defined in (\ref{eq:calH}). Thus, it is natural to define the time evolution of the vector $\kappa \in \mathcal{F}(\frak{h}) \otimes \mathcal{F}(\frak{h}) $ as:
\begin{equation*}
\kappa_{t} = (e^{-i \mathcal{H} t / \varepsilon} \otimes e^{i \mathcal{H} t / \varepsilon}) \kappa\;.
\end{equation*}
Letting $\Psi_{t} = U \kappa_{t} \in \mathcal{F}(\frak{h} \oplus \frak{h})$, we find 
\[  \tr_{\mathcal{F}(\mathfrak{h})}\, \mathcal{O} \rho_{t} = \langle \Psi_{t}, U (\mathcal{O} \otimes 1) U^* \Psi_{t} \rangle \;. \]
We observe that 
\begin{equation*}
\Psi_{t} = U \kappa_{t} = e^{-i\mathcal{L} t / \varepsilon} \Psi\;,
\end{equation*}
where the Liouvillian $\mathcal{L}$ is defined as:
\begin{equation*}
\mathcal{L} = U ( \mathcal{H} \otimes 1 - 1 \otimes \mathcal{H} ) U^{*}\;.
\end{equation*}
In terms of the operator valued distributions, we find 
\begin{equation}\label{eq:liou}
\begin{split}
\mathcal{L} &= d\Gamma_{l}( \sqrt{1-\varepsilon^{2} \Delta } ) + \frac{\varepsilon^{3}}{2} \int dxdy\, V(x-y) a^{*}_{l, x} a^{*}_{l,y} a_{l,y} a_{l,x} \\ 
& \quad - d\Gamma_{r}( \sqrt{1-\varepsilon^{2} \Delta } ) - \frac{\varepsilon^{3}}{2} \int dxdy\, V(x-y) a^{*}_{r, x} a^{*}_{r,y} a_{r,y} a_{r,x}\;.
\end{split}
\end{equation}
%
%Thus, the time evolution of the expectation value of an operator $\mathcal{O}$ is:
%%
%\begin{equation*}
%\tr_{\mathcal{F}(\mathfrak{h})}\, \mathcal{O} \rho_{t} = \langle \Psi_{t}, (\mathcal{O} \otimes 1) \Psi_{t} \rangle = \langle \Psi, e^{i\mathcal{L} t /
 %\varepsilon} (\mathcal{O} \otimes 1) e^{-i\mathcal{L} t / \varepsilon} \Psi \rangle\;.
%\end{equation*}
%%

\subsection{Bogoliubov transformations and quasi-free states}\label{sec:bogpure}

A density matrix $\rho$ on $\mathcal{F} (\frak{h})$ describes a quasi-free state, if it satisfies the fermionic Wick's theorem, which states that 
\begin{equation}\label{eq:wick}
\tr_{\mathcal{F} (\frak{h})} \; a^{\sharp_{1}}_{x_{1}} a^{\sharp_{2}}_{x_{2}} \cdots a^{\sharp_{2j}}_{x_{2j}}  \rho  = \sum_{\sigma \in P_{2j}}  \text{sgn}(\sigma) \tr_{\mathcal{F} (\frak{h})}  \, a^{\sharp_{\sigma(1)}}_{x_{\sigma(1)}} a^{\sharp_{\sigma(2)}}_{x_{\sigma(2)}} \rho \cdots 
\tr_{\mathcal{F} (\frak{h})} a^{\sharp_{\sigma(2j-1)}}_{x_{\sigma(2j-1)}} a^{\sharp_{\sigma(2j)}}_{x_{\sigma(2j)}} \rho \;,
\end{equation}
where $a^{\sharp}(f)$ denotes either $a(f)$ or $a^{*}(f)$ and where 
\begin{equation*}
P_{2j} = \Big\{ \sigma \in S_{2j} \; \big|\; \sigma(2\ell-1) < \sigma(2\ell + 1)\;,\; \ell=1, \ldots, j-1\;,\; \sigma(2\ell-1) < \sigma(2\ell)\;,\; \ell = 1,\ldots, j \Big\}
\end{equation*}
is the set of possible pairing of $2j$ elements ($S_{2j}$ is the set of permutations). From (\ref{eq:wick}), we conclude that quasi-free states are completely characterized by their reduced density matrix and their pairing density, defined by the integral kernels 
\[ \begin{split} \gamma^{(1)}_\rho (x;y) &= \tr_{\mathcal{F} (\frak{h})} a_x^* a_y \rho \\ 
\alpha^{(1)}_\rho (x;y) &= \tr_{\mathcal{F} (\frak{h})} a_x a_y \rho \;.\end{split} \]
From the canonical anticommutation relations (\ref{eq:CAR}), it is easy to check that the hermitian operator $\gamma^{(1)}_\rho$ and the symmetric operator $\alpha^{(1)}_\rho$ satisfy the inequality \[ (\alpha^{(1)}_\rho)^* \alpha^{(1)}_\rho \leq \gamma^{(1)}_\rho (1- \gamma^{(1)}_\rho) .\] 

On the other hand, for every hermitian operator $\omega$ and symmetric operator $\alpha$ on $\frak{h}$, with $0 \leq \omega \leq 1$, $\tr \; \omega < \infty$ and $\alpha^* \alpha \leq \omega (1-\omega)$, there exists a unique quasi-free state $\rho$ on $\mathcal{F} (\frak{h})$ with $\omega$ as reduced density matrix and $\alpha$ as pairing density. The quasi-free state $\rho$ is pure if and only if $\alpha^* \alpha = \omega (1-\omega)$ holds as an equality, rather than just as an inequality; see {\it e.g.} \cite{BLS, Solovej} for mathematical reviews.

As discussed in Section \ref{sec: effective-zero-temp}, quasi-free states associated with an orthogonal projection $\omega$ and with $\alpha =0$ are Slater determinants. This class of quasi-free states is relevant to describe equilibrium states of non-interacting Fermi gases at zero temperature. 
In this section, we are interested in another class of quasi-free states, describing equilibrium states at positive temperature. We can still restrict our attention to $\alpha = 0$, but now $0 \leq \omega \leq 1$ will be a Fermi-Dirac distribution, not an orthogonal projection. The corresponding quasi-free state is mixed and it cannot be represented by a vector in the Hilbert space $\mathcal{F} (\frak{h})$. It can, however, be represented by a vector $\Psi = \mathcal{R}_\omega \Omega \in \mathcal{F} (\frak{h} \oplus \frak{h})$, where (in analogy with (\ref{eq:Romega2}) for Slater determinant), the Bogoliubov transformation $\mathcal{R}_\omega$ is a unitary operator, satisfying   
\begin{equation}\label{eq:Ra}
\mathcal{R}_\omega^{*} a_{l,x} \mathcal{R}_\omega = a_{l}(u_{x}) - a^{*}_{r}(\overline{v}_{x})\;,\qquad \mathcal{R}^{*}_\omega a_{r,x} \mathcal{R}_\omega = a_{r}(\overline{u}_{x}) + a^{*}_{l}(v_{x})
\end{equation}
with $u = \sqrt{1-\omega}$ and $v = \sqrt{\omega}$. Using these rules, together with the anticommutation relations (\ref{eq:CAR}), it is easy to check that $\Psi$ is indeed a quasi-free state and, recalling (\ref{eq:1partPsi}), that $\gamma_\Psi = \omega$ and $\pi_\Psi = 0$. The vector $\Psi$ is known as the Araki-Wyss representation of the mixed quasi-free state with reduced density $\omega$ and vanishing pairing density. From (\ref{eq:Ra}), we can also compute the action of the Bogoliubov transformation $\mathcal{R}_\omega$ on operators of the form $d\Gamma_\sigma (O)$, which will play an important role in our analysis. 
\begin{lemma}\label{lemma: Bogolubov-observables}
Let $0 \leq \omega \leq 1$ and $\mathcal{R}_\omega$ denote the Bogoliubov transformation defined in (\ref{eq:Ra}), generating the quasi-free state with reduced density $\omega$ and with vanishing pairing density. Let $u = \sqrt{1-\omega}$ and $v = \sqrt{\omega}$. Then, for any operator $O$ on $\mathfrak{h}$, we have the identities
\begin{equation*}
\begin{split}
\mathcal{R}^{*}_\omega d \Gamma_{l}(O) \mathcal{R}_\omega & = \tr ( O v^{*} v) +  d \Gamma_{l}(u O u^{*}) - d \Gamma _{r} (\overline{v} O  \overline{v}^{*}) 
- d \Gamma^{+}_{l r}(u O v^{*}) -d \Gamma^{-}_{r l}(v O u^{*})  \;,
\\
\mathcal{R}^{*}_\omega d \Gamma_{r}(O) \mathcal{R_\omega} & = \tr ( O \overline{v}^{*} \overline{v}) +  d \Gamma_{l}(\overline{u} O \overline{u}^{*}) - d \Gamma _{l} (v O v^{*}) 
+ d \Gamma^{+}_{r l}(\overline{u} O \overline{v}^{*}) +d \Gamma^{-}_{l r}(\overline{v} O \overline{u}^{*})  \;.
\\
\end{split}
\end{equation*}
\end{lemma}

\subsection{Local convergence to the Hartree dynamics at positive temperature}
\label{sec: conve-hartree-mixed}

We will consider the time evolution of initial data that can be approximated by a mixed quasi-free state. We will assume a uniform estimate for the local density of the quasi-free state.
\begin{assumption}[Bounded Density]
\label{ass:bd} 
Let $0 \leq \omega \leq 1$ and let $ \mathcal{W}_z^{(1)}$ denote the spatial weight, localizing around the point $z \in \bR^3$, introduced in (\ref{eq: polynomial weight}). We say that $\omega$ has bounded density over $\bR^3$ if there exists $C > 0$ such that 
\[ \tr \, \W^{(1)}_z \omega \leq C \veps^{-3} \]
for all $z \in \bR^3$.
\end{assumption}
\begin{remark}
Note that by monotonicity in $n$, if Assumption \ref{ass:bd}  holds true, we also have $\tr \, \W_z^{(n)} \omega \leq C \veps^{-3}$, for all $n>1 $.
\end{remark}
Furthermore, we will assume that the initial quasi-free state exhibits a local semiclassical structure in a domain $\Lambda \subset \bR^3$, in the sense specified by the following assumption.
\begin{assumption}[Local Semiclassical Structure]\label{ass:locsc}
\label{ass:sc} Let $0 \leq \omega \leq 1$. Let $\pi$ indicate either $\omega$, $\sqrt{\omega}$ or $1-\sqrt{1-\omega}$.
%%
%\begin{equation}\label{eq:def-P_n}
%\mathcal{W}^{(n)}_{z}(x) = \frac{1}{1 + |z - x|^{4n}} 
%\end{equation}
%%
We say that $\omega$ exhibits a local semiclassical structure in $\Lambda \subset \bR^3$ if, for all $n \in \bN$ large enough, there exist a constant $C_{n} > 0$, independent of $\Lambda$, such that
\begin{align} 
 %\| \mathcal{W}^{(n)}_{z_1} \omega \|_\text{tr} &\leq C \varepsilon^{-3} \label{eq:sc0}  \\
 (1+|z - z'|^{2n - 4}) \| \mathcal{W}^{(n)}_{z} \omega \mathcal{W}_{z'}^{(n/2)} \|_\tr  &\leq C_{n} \varepsilon^{-3} \label{eq:sc1} \\
(1 + |z - z'|^{2n})  \big \| \W_{z}^{(n)} \big[ \veps \nabla, \pi \big]  \W_{z'}^{(n/2)} \big\|_{\tr}  &\leq C_{n} \veps^{-2}\label{eq:sc4}\\
%(1 + |z - z'|^8)  \big \| \W_{z}^{(n)} \big[ \veps \nabla, \sqrt{\mathbbm{1}-\omega} \big]  \W_{z'}^{(n)} \big\|_{\tr}  &\leq %C \veps^{-2}\;.\label{eq:sc5}\\
(1 + |z - z'|^{2n}) \big \| \W_{z}^{(n)} \big[ e^{i p \cdot \hat{x}}, \pi \big] \W_{z'}^{(n/2)}\big\|_{\tr} &\leq C_{n} \veps^{-2}(1+|p|) \label{eq:sc2} \\
 (1+|z- z'|^{2n}) \| \mathcal{W}^{(n)}_{z} \sqrt{\omega} \sqrt{1-\omega} \mathcal{W}_{z'}^{(n/2)} \|_\tr &\leq C_{n} \veps^{-2}  \label{eq:sc1b} 
%\\
%(1 + |z - z'|^8)  \big \| \W_{z}^{(n)} \big[ e^{i p \cdot \hat{x}}, \sqrt{\mathbbm{1}-\omega} \big] \W_{z'}^{(n)}\big\|_{\tr} %&\leq C \veps^{-2} (1+|p|) \label{eq:sc3}
\end{align} 
for all $p \in \bR^3$, $z \in \Lambda$ and $z' \in \bR^3$. Here $\mathcal{W}_z^{(n)}$ denotes the spatial weight defined in (\ref{eq: polynomial weight}). 
\end{assumption} 
\begin{remark}
For typical $H$ of the form (\ref{eq:H-noint}) and chemical potential $\mu > 0$, the Fermi projection $\omega_\mu = \chi (H \leq \mu)$ is not expected to satisfy this assumption; this is why we are switching here to positive-temperature states. The problem in establishing the bounds (\ref{eq:sc1})-(\ref{eq:sc1b}) at zero temperature is already evident for the free Fermi gas ($V_{\text{ext}} = 0$). There, the reduced one-particle density matrix at density $\varepsilon^{-3}$ is:
\begin{equation}
\begin{split}
\omega_{\text{FFG}}(x,y) &= \frac{1}{\varepsilon^{3}} \varphi\Big( \frac{x-y}{\varepsilon} \Big)\;,\qquad \text{where} \\
\varphi(z) &= \frac{4\pi}{|z|^{2}} \Big( \frac{\sin(|z|)}{|z|} - \cos (|z|) \Big)\;.
\end{split}
\end{equation} 
The non-integrable decay in $|x-y|$ of $\omega_{\text{FFG}}(x,y)$ has to be compared with the fast decay assumed for the localized trace norms in the bounds (\ref{eq:sc1})-(\ref{eq:sc1b}).

%One way to estimate the trace norm of an operator $K$ is via the following bound, see {\it e.g.} \cite[Chapter 9]{DS}:
%
%\begin{equation}
%\| K \|_{\tr} \leq C\sum_{\alpha: |\alpha| \leq 7} \big\| \partial_{x,y}^{\alpha} K \big\|_{1}\;,
%\end{equation}
%
%where $\big\| \partial_{x,y}^{\alpha} K \big\|_{1}$ is the $L^{1}$ norm of the derivatives of the kernel $K(x,y)$. Let $K = \mathcal{W}^{(n)}_{z_1} \omega \mathcal{W}_{z_2}^{(n)}$ and consider the contribution to the right-hand side with $\alpha = 0$. Then, for $|z_{1}-z_{2}|\gg 1$, we have:
%
%\begin{equation}
%\int dxdy\, \mathcal{W}^{(n)}_{z_1}(x) |\omega_{\text{FFG}}(x,y)| \mathcal{W}_{z_2}^{(n)}(y) \leq C\frac{\varepsilon^{-1}}{|z_{1} - z_{2}|^{2}}\;,
%\end{equation}
%
%which is much worse than (\ref{eq:sc1}) for $|z_{1} - z_{2}| \gg \varepsilon^{-1/3}$; we will need the estimate (\ref{eq:sc1}) for arbitrarily large $|z_{1} - z_{2}|$, in particular to show that $\| \mathcal{W}^{(n)}_{z_1} \omega \mathcal{W}_{z_2}^{(n)} \|_\tr$ is integrable in $|z_{1} - z_{2}|$. A similar discussion applies to the bounds (\ref{eq:sc4})-(\ref{eq:sc1b}).
\end{remark}

The next theorem establishes the validity of the Hartree equation as an approximation for the time evolution of approximately quasi-free initial data, exhibiting a local semiclassical structure. 
\begin{theorem}[Derivation of the Hartree Equation on a Local Scale - Mixed States]\label{thm:T>0}  
Let $0 \leq \omega \leq 1$ be an operator over $\frak{h}$, satisfying Assumption \ref{ass:bd} (bounded density) and Assumption~\ref{ass:sc} (local semiclassical structure) in a domain $\Lambda \subset \bR^3$. Let $\Xi \in \mathcal{F}(\frak{h} \oplus \frak{h})$ be such that:
\begin{equation}\label{eq:assT>0}
\sup _{z \in \R^{3}} \,
\langle \Xi, d\Gamma_\sigma (\mathcal{W}^{(n)}_{z}) \Xi \rangle \leq C\varepsilon^{-3}\;, \qquad  \sup _{z \in \Lambda} \, 
\langle \Xi, d\Gamma_\sigma (\mathcal{W}^{(n)}_{z}) \Xi \rangle \leq C\varepsilon^{-3+\delta}\;,
\end{equation}
for all $n \in \bN$ large enough, for $\sigma = l,r$, for some $0 < \delta \leq 1$ and for some $C>0$ independent of $\Lambda$. Let $\mathcal{L}$ denote the  Liouvillian  (\ref{eq:liou}), with $V \in \mathcal{S} (\bR^3)$. Consider the time evolution 
\begin{equation*}
\Psi_{t} = e^{-i \mathcal{L} t / \varepsilon} \mathcal{R}_{\omega} \Xi\;, 
\end{equation*}
of the initial datum $\mathcal{R}_{\omega} \Xi$ approximating the quasi-free state with reduced density $\omega$ and vanishing pairing density ($\mathcal{R}_\omega$ denotes the Bogoliubov transformation introduced in (\ref{eq:Ra})). 

Let $\gamma_t^{(1)}$ denote the one-particle reduced density matrix associated with $\Psi_t$ and let $\omega_{t}$ be the solution of the pseudo-relativistic Hartree equation (\ref{eq:hartree-main}), with initial datum $\omega_{t=0} = \omega$. Let $\mathcal{O}$ be as in Theorem \ref{thm:main} and let $\mathcal{O}_{z}(x) = \mathcal{O}(x - z)$.
%
%\begin{equation*}
%i \veps \partial_{t} \omega_{t} = \Big[ \sqrt{1-\varepsilon^{2} \Delta} + \varepsilon^{3} (V * \varrho_{t})\,,\, \omega_{t} \Big]\;,
%\qquad \omega_{0} = \omega\;,
%\end{equation*}
%
Then, there exist constants $C, c>0$, independent of $\Lambda$, such that 
\begin{equation}
\Big| \tr\, \mathcal{O}_{z} ( \gamma^{(1)}_{t} - \omega_{t} ) \Big| \leq C \exp (c \exp (c|t|)) \, \varepsilon^{-3+\delta} \; , 
\end{equation}
for all $t\in \mathbb{R}$ and $z \in \Lambda$ with $B_{\veps^{-\delta} |t|}(z) \subset \Lambda$. 
\end{theorem}

An important observation, which allows us to apply Theorem~\ref{thm:T>0} to study the evolution of approximate Slater determinants and to prove Theorem~\ref{thm:main}, is the fact that equilibrium states of non-interacting Fermi gases with Hamiltonian $d\Gamma (H)$ (and $H$ as defined in (\ref{eq:H-noint})), at chemical potential $\mu > 0$ and at positive temperature $\beta^{-1} = O (\veps)$, do satisfy Assumption~\ref{ass:bd} and Assumption~\ref{ass:sc}, if $H$ and $\mu$ satisfy Assumption \ref{ass:Weyl}. 

\begin{proposition}[Bounded Density and Local Semiclassical Structure for Equilibrium States]\label{prop:localsc} 
Let $H$ be the Hamiltonian defined in (\ref{eq:H-noint}), $\beta = O (\varepsilon^{-1})$. Suppose that $(H,\mu )$ satisfy Assumption \ref{ass:Weyl} in the region $\Lambda \subset \bR^3$. Then, the Fermi-Dirac distribution 
\begin{equation}\label{eq:FD}
\omega_{\mu,\beta}  = \frac{1}{1 + e^{\beta (H - \mu)}}\;.
\end{equation}
at chemical potential $\mu$ and at inverse temperature $\beta$, satisfies Assumption \ref{ass:bd} and Assumption \ref{ass:sc} in the domain $\Lambda$.  
\end{proposition}

We postpone the proof of this proposition to Appendix \ref{sec: appen-apriori} (for the verification of Assumption \ref{ass:bd}) and to Appendix \ref{app:gibbs} (for the proof of Assumption \ref{ass:sc}).

\section{Proof of Theorem \ref{thm:T>0}}
\label{sec:Tproof} 

\subsection{Propagation of the local semiclassical structure}
\label{sec: prop-semiclassical}

The first important ingredient in the proof of Theorem \ref{thm:T>0} is the propagation of the bound on the density in Assumption \ref{ass:bd} and of the local semiclassical bounds in Assumption \ref{ass:sc} along the Hartree evolution. This is the content of the next proposition.  

\begin{proposition}[Propagation of the Local Semiclassical Structure at Positive Temperature]\label{prop: semiclassical_prop-main} 
Let $V \in \mathcal{S} (\bR^3)$. Suppose that $0 \leq \omega \leq 1$ satisfies Assumption \ref{ass:bd} and Assumption \ref{ass:sc} in a domain $\Lambda \subset \bR^3$. Let $\omega_{t}$ denote the solution of the Hartree equation
\begin{equation}\label{eq:hartree2}  
i \varepsilon \partial_t \omega_t = \big[ \sqrt{1-\veps^2 \Delta} + \veps^{3} (V* \varrho_t) , \omega_t ] \end{equation}
with $\varrho_t (x) = \omega_t (x;x)$ and initial datum $\omega_{t=0} = \omega$. Then, there exist $C,c > 0$ such that 
\begin{equation}\label{pro: a priori-propagation}
\sup_{z \in \R^{3}} \tr\, \W^{(1)}_{z} \omega_{t} \leq C \exp(c|t|) \veps^{-3} \; 
\end{equation}
for all $t \in \R$. Moreover, for all $n \in \bN$ large enough there are constants $C_{n},c_{n}>0$ independent of $\Lambda$, such that  
\begin{equation}\label{eq:prop-uv}
\begin{split}
\sup_{z \in \Lambda, z' \in \bR^3} (1+|z-z'|^{2n-4})  \| \mathcal{W}_z^{(n)} \omega_t \mathcal{W}_{z'}^{(n/2)}\|_\tr  & \leq C_{n} \exp(c_{n}|t|) \veps^{-3}   \\
%\sup_{z \in \Lambda, z' \in \bR^3} (1+|z-z'|^n)  \| \mathcal{W}_z^{(n)} \omega_t \mathcal{W}_{z'}^{(n)}\|_\tr  & \leq C \exp(c_{n}|t|) \veps^{-3} \; , \\
%\sup_{z \in \Lambda, z' \in \bR^3} (1+|z-z'|^n)  \| \mathcal{W}_z^{(n)} \sqrt{\omega_t} \mathcal{W}_{z'}^{(n)}\|_\tr  &\leq C \exp(c_{n}|t|) \veps^{-3} \\
%\sup_{z \in \Lambda, z' \in \bR^3} (1+|z-z'|^n)  \| \mathcal{W}_z^{(n)} [\sqrt{1-\omega_t} - 1] \mathcal{W}_{z'}^{(n)}\|_\tr  &\leq C \exp(c|t|) 
%\veps^{-3} \\
%\sup_{z \in \Lambda, z' \in \bR^3} (1+|z-z'|^n)  \| \mathcal{W}_z^{(n)} \sqrt{\omega_t} \sqrt{1-\omega_t} \mathcal{W}_{z'}^{(n)}\|_\text{tr}  &\leq C 
%\exp(c_{n}|t|) \veps^{-3} 
\end{split} \end{equation} 
and, with $\pi_t$ denoting either $\omega_t, \sqrt{\omega_t}$ or $\sqrt{1-\omega_t} $,
\begin{equation}\label{eq:eipx-u-v}
\begin{split}
\sup _{z\in \Lambda, z' \in \bR^3, p \in \bR^3}  \frac{(1+|z-z'|^{2n})}{1+|p|}  \Big\| \mathcal{W}_z^{(n)} \big[\pi_{t}, e^{i p \cdot \hat{x}}\big] \mathcal{W}_{z'}^{(n/2)} \Big\|_\tr &\leq C_{n} \exp(c_{n}|t|) \veps^{-2}\;, \\ 
\sup _{z\in \Lambda, z' \in \bR^3}  (1+|z-z'|^{2n})  \Big\| \mathcal{W}_z^{(n)} \big[\pi_{t}, \veps \nabla \big] \mathcal{W}_{z'}^{(n/2)} \Big\|_\tr &\leq C_{n} \exp(c_{n}|t|) \veps^{-2} 
%\sup _{z\in \Lambda, z' \in \bR^3, p \in \bR^3}  \frac{(1+|z-z'|^n)}{1+|p|}  \Big\| \mathcal{W}_z^{(n)} \big[\omega_{t}, \ee^{\ii p \cdot \hat{x}}\big] 
%\mathcal{W}_{z'}^{(n)} \Big\|_\tr &\leq C \exp(c|t|) \veps^{-2}\;, \\
%\sup _{z\in \Lambda, z' \in \bR^3, p \in \bR^3}  \frac{(1+|z-z'|^n)}{1+|p|}  \Big\| \mathcal{W}_z^{(n)} \big[\sqrt{\omega_{t}}, \ee^{\ii p \cdot \hat{x}}\big] %\mathcal{W}_{z'}^{(n)} \Big\|_\text{tr} &\leq C \exp(c|t|) \veps^{-2} \; \\
%\sup _{z\in \Lambda, z' \in \bR^3}  \frac{(1+|z-z'|^n)}{1+|p|}  \Big\| \mathcal{W}_z^{(n)} \big[ \sqrt{1-\omega_{t}}, \ee^{\ii p \cdot \hat{x}}
%\big] \mathcal{W}_{z'}^{(n)} \Big\|_\text{tr} &\leq C \exp(c|t|) \veps^{-2}
\end{split}
\end{equation}
for every $t \in \bR$. Finally, for all $n \in \bN$, there are constants $C_n, c_n > 0$ such that 
\begin{equation}\label{eq:overuv}
\sup_{z \in \Lambda, z' \in \bR^3} (1+|z-z'|^{2n}) \| \W_z^{(n)} \sqrt{\omega_t} \sqrt{1-\omega_t} \W_{z'}^{(n/2)} \|_\tr \leq C_{n} \exp (c_{n}|t|) \veps^{-2} 
\end{equation}  
for every $t \in \bR$. 
%as well as  
%\begin{equation}\label{eq:nabla-u-v}
%\begin{split}
%\sup _{z\in \Lambda, z' \in \bR^3, p \in \bR^3}  (1+|z-z'|^n)  \Big\| \mathcal{W}_z^{(n)} \big[\pi_{t}, \veps \nabla \big] \mathcal{W}_{z'}^{(n)} \Big\|%_\text{tr} &\leq C \exp(c|t|) \veps^{-2} 
%\sup _{z\in \Lambda, z' \in \bR^3, p \in \bR^3}  (1+|z-z'|^n)  \Big\| \mathcal{W}_z^{(n)} \big[\omega_{t}, \veps \nabla \big] \mathcal{W}_{z'}^{(n)} \Big\|%_\text{tr} &\leq C \exp(c|t|) \veps^{-2}\;, \\
%\sup _{z\in \Lambda, z' \in \bR^3, p \in \bR^3}  (1+|z-z'|^n) \Big\| \mathcal{W}_z^{(n)} \big[\sqrt{\omega_{t}}, \veps \nabla \big] \mathcal{W}_{z'}^{(n)} %\Big\|_\text{tr} &\leq C \exp(c|t|) \veps^{-2} \; \\
%\sup _{z\in \Lambda, z' \in \bR^3, p \in \bR^3}  (1+|z-z'|^n)  \Big\| \mathcal{W}_z^{(n)} \big[ \sqrt{1-\omega_{t}}, \veps \nabla \big] \mathcal{W}_{z'}%^{(n)} \Big\|_\text{tr} &\leq C \exp(c|t|) \veps^{-2}
%\end{split}
%\end{equation}
\end{proposition}

\begin{remark}\label{rem:propa} 
It follows from (\ref{eq:prop-uv}) that 
\begin{equation}\label{eq:1loc} \| \W_z^{(n)} \omega_t \|_\tr \leq C \exp (c|t|) \veps^{-3} \end{equation} 
for $n \in \bN$ large enough and for all $z \in \Lambda$. In fact, fixing $n_0 \geq 2$ so that (\ref{eq:prop-uv}) holds true,
\[ \frac{1}{C} \int_{\bR^3} \W^{(n_0 /2)}_{z'} (x) \, dz' = 1 \]
for some $C > 0$, we can estimate
\begin{equation*}
\begin{split}
 \| \W_z^{(n_0)} \omega_t \|_\tr & \leq \frac{1}{C} \int_{\bR^3} dz' \, \| \W_z^{(n_0)} \omega_t  \W^{(n_0 /2)}_{z'} \|_\tr 
 \\
 &\leq C \exp (c |t|)  \veps^{-3}  \int_{\bR^3}  dz' \frac{1}{1+|z-z'|^{2n_0}}  \leq C \exp (c|t|)\veps^{-3} \end{split}
\end{equation*}
By monotonicity, we get (\ref{eq:1loc}) for all $n \geq n_0$. 
Similar bounds can be also derived from (\ref{eq:eipx-u-v}) and (\ref{eq:overuv}). 
\end{remark} 

\begin{proof} 
We proceed analogously to \cite[Sect.~6]{FPS}. The constants in this proof depend on $n$, but we will drop this dependence for brevity. We consider the Hartree evolution $U$ defined by 
\[ i \veps \partial_t U(t;s) = h(t) U(t;s),\]
 with $U(s;s) = 1$ for all $s \in \bR$ and with $h(t) = \sqrt{1-\varepsilon^2 \Delta} + \veps^{3} (V * \varrho_t) (x)$. For $p \in \bR^3$, we also introduce the modified Hartree evolution $U_p$ defined by 
\[ i \veps\partial_t U_p (t;s) = h_p (t) U_p (t;s),\]
with $U_p (s;s) = 1$ and 
\[ h_p (t) = e^{ip \cdot x} h(t) e^{-ip\cdot x} = \sqrt{1 + \varepsilon^2 (-i \nabla + p)^2} + \veps^{3} (V* \varrho_t) (x) \]
Notice that, with this notation, $U = U_{p=0}$. An important property of the evolution $U_p$ is that it preserves locality, in the sense that 
\begin{equation}\label{eq:Up-W} U_p (t;s)^* \mathcal{W}^{(n)}_z U_p (t;s) \leq e^{c |t-s|} \mathcal{W}_z^{(n)} \, , \end{equation} 
for some $c>0$ depending on $n$ only. 
For $p=0$, the proof of (\ref{eq:Up-W}) is given in \cite[Prop.~6.2]{FPS}; the argument of  \cite[Prop.~6.2]{FPS} actually applies to general time-dependent Hamiltonians $H(t) = \sqrt{1 - \varepsilon^{2} \Delta} + V(x,t)$. Since that proof can be trivially extended to $p \not =0$, we skip the details. Because $\omega_t  = U (t;0)^* \omega U(t;0)$, Eq. (\ref{eq:Up-W}) immediately implies that 
\[ \tr \, \mathcal{W}_z \omega_t =  \tr \; U (t;0) \mathcal{W}_z U(t;0)^* \omega \leq e^{c|t|} \tr\, \mathcal{W}_z \omega \leq C e^{c|t|} \varepsilon^{-3} \;. \]
Similarly, also (\ref{eq:prop-uv}) and (\ref{eq:overuv}) follow from (\ref{eq:Up-W}) and from the corresponding bounds at time $t=0$ (which are contained in Assumption \ref{ass:sc}), because $\sqrt{\omega_t} \sqrt{1-\omega_t} = U(t;0)^* \sqrt{\omega} \sqrt{1-\omega} \, U(t;0)$. 

To prove (\ref{eq:eipx-u-v}), we observe that 
\[ i\varepsilon \partial_t \pi_t = \big[ h(t) , \pi_t ]\;. \]
Thus, we obtain 
\[ i\varepsilon \partial_\tau U(t;\tau) [e^{ip \cdot \hat{x}}, \pi_\tau] U_p (t;\tau)^* = U (t;\tau) e^{ip \cdot \hat{x}} \big[ \pi_\tau , \varepsilon  p \cdot A(p) \big] U_p (t;\tau)^* \]
where we defined 
\[ A (p) = \int_0^1 ds\, \frac{\varepsilon (-i \nabla + s p)}{\sqrt{1+\varepsilon^2 (-i\nabla +s p)^2}}\;. \]
Integrating over $\tau$, we find 
\begin{equation}\label{eq: gron-1-1}
\begin{split}
\big \| \mathcal{W}_z^{(n)}  \big[ e^{i p \cdot \hat{x}}, \pi_{t} \big]  \mathcal{W}_{z'}^{(n/2)} \big\|_\text{tr} 
\leq \; &\big \| \mathcal{W}_z^{(n)}  U(t;0)  \big[ e^{i p \cdot \hat{x}}, \pi_{0} \big]U_{p}(t;0)^{*} \mathcal{W}_{z'}^{(n/2)} \big\|_\text{tr} 
\\
& + |p| \int_{0}^{t}\big \| \mathcal{W}^{(n)}_z U(t;\tau) e^{i p \cdot \hat{x}} \big[ \pi_{\tau}, A(p) \big]  U_{p}(t;\tau)^{*} \mathcal{W}^{(n/2)}_{z'}  \big\|_\text{tr}    d \tau \;.
\end{split}
\end{equation}
With (\ref{eq:Up-W}) and proceeding as in \cite[Proof of Theorem 6.1, until Eq.~(6.8)]{FPS} to replace $[ \pi_\tau, A (p)]$ with $[\pi, \varepsilon \nabla]$, we arrive at (we focus for simplicity on $t > 0$) 
\[ 
\begin{split}
\big \| \mathcal{W}_z^{(n)}  \big[ e^{i p \cdot \hat{x}}, \pi_{t} \big]  \mathcal{W}_{z'}^{(n/2)} \big\|_\text{tr} 
\leq \; &e^{c t} \big \| \mathcal{W}_z^{(n)}  \big[ e^{i p \cdot \hat{x}}, \pi_{0} \big]  \mathcal{W}_{z'}^{(n/2)} \big\|_\text{tr} 
\\
& + C |p| \int_{0}^{t} d\tau\, e^{c(t-\tau)} \big \| \mathcal{W}^{(n)}_z \big[ \pi_{\tau}, \varepsilon \nabla \big] \mathcal{W}^{(n/2)}_{z'}  \big\|_\text{tr}  \;.
\end{split}\]
Dividing by $(1+|p|)$, multiplying by $(1+|z-z'|^{2n})$ and taking the supremum, we obtain  
\[ \begin{split}
e^{-ct} \sup_{z \in \Lambda , z' \in \bR^3, p \in \bR^3} &\frac{(1+ |z-z'|^{2n})}{1+|p|} \big \| \mathcal{W}_z^{(n)}  \big[ e^{i p \cdot \hat{x}}, \pi_{t} \big]  \mathcal{W}_{z'}^{(n/2)} \big\|_\text{tr} 
\\ \leq \; &\sup_{z \in \Lambda, z' \in \bR^3, p \in \bR^3} \frac{(1+ |z-z'|^{2n})}{1+|p|}  \big \| \mathcal{W}_z^{(n)}  \big[ e^{i p \cdot \hat{x}}, \pi_{0} \big]  \mathcal{W}_{z'}^{(n/2)} \big\|_\text{tr} 
\\
& +C \int_{0}^{t}  d \tau\, e^{-c\tau} \sup_{z \in \Lambda, z' \in \bR^3} (1+|z-z'|^{2n}) \big\| \mathcal{W}^{(n)}_z \big[ \pi_{\tau}, \varepsilon \nabla \big] \mathcal{W}^{(n/2)}_{z'}  \big\|  \;.
\end{split}\]
Proceeding as in \cite[Proof of Theorem 6.1, after Eq.~(6.9)]{FPS} we also obtain 
\[ \begin{split}
e^{-c t} \sup_{z \in \Lambda, z' \in \bR^3} &(1+|z-z'|^{2n}) \big \| \mathcal{W}_z^{(n)}  \big[ \varepsilon \nabla , \pi_{t} \big]  \mathcal{W}_{z'}^{(n/2)} \big\|_\text{tr} 
\\ \leq \; & 
\sup_{z \in \Lambda, z' \in \bR^3} (1+|z-z'|^{2n}) \big \| \mathcal{W}_z^{(n)}  \big[ \varepsilon \nabla , \pi_{0} \big]  \mathcal{W}_{z'}^{(n/2)} \big\|_\text{tr} 
\\
& + C \int_{0}^{t} d \tau\, 
e^{-c \tau} \sup_{z \in \Lambda, z' \in \bR^3, p \in \bR^3} \frac{(1+|z-z'|^{2n})}{1+|p|} \big \| \mathcal{W}^{(n)}_z \big[ \pi_{\tau}, e^{i p \cdot x}  \big] \mathcal{W}^{(n/2)}_{z'}  \big\|  \;.
\end{split}\]
The bounds (\ref{eq:eipx-u-v}) now follow from Gronwall's lemma and from the assumptions on the initial datum. 
\end{proof} 
\begin{remark} Strictly speaking, \cite[Theorem 6.1]{FPS} provides an estimate for the commutators of the solution of the time-dependent Hartree equation with $e^{ip\cdot x}$ for $|p| < \varepsilon^{-1}$. This restriction is important in the non-relativistic case, but it is actually irrelevant in the pseudo-relativistic case; the commutator estimate holds true for all $p\in \mathbb{R}^{3}$, as stated above. To see this, one observes that in \cite[Eq.~(6.10)]{FPS} there is actually no need to split $p$ in $|p| < \varepsilon^{-1}$ and $|p| \geq \varepsilon^{-1}$; it is enough to multiply and divide the trace norm by $(1+|p|)$ and take the supremum over $p$ in $\mathbb{R}^{3}$.
\end{remark}

As a corollary of Proposition \ref{prop: semiclassical_prop-main}, we can control commutators of $\omega_t$ with sufficiently regular and decaying  functions of $\hat x$. The proof of the next corollary is based on the bounds of Proposition \ref{prop: semiclassical_prop-main} and on the remarks thereafter. %It follows closely the proof of \cite[Corollary 4.3]{FPS} and we will skip the details.
\begin{corollary}[Localized Commutators with Regular Functions]\label{cor:commF} 
Let $V$ be as in Theorem~\ref{thm:main}. Suppose that $0 \leq \omega \leq 1$ satisfies Assumption \ref{ass:bd} and  Assumption \ref{ass:sc} in a region $\Lambda \subset \bR^3$. Let $\omega_{t}$ be the solution of the Hartree equation (\ref{eq:hartree2}). Let 
\begin{equation}\label{eq:Fass}
F(x) = \int dp\, e^{ip \cdot x} \hat F(p)\;,\qquad  \int dp\, (1 + |p|) \big| \partial_{p}^{\alpha} \hat F(p) \big| \leq C \end{equation}
for $|\alpha| \leq \ell$, with $\ell \in \bN$ sufficiently large. For any $z\in \R^{3}$ let $F_{z}(x) = F(x-z)$. Let $\pi_t$ denote either $\omega_t, \sqrt{\omega_t}$ or $\sqrt{1-\omega_t}$. Then, we have:
\begin{equation}\label{eq:commF}
\begin{split}
\sup _{z \in \Lambda, z' \in \R^{3}} (1+|z-z'|^{2n}) \Big\| \W_{z'}^{(n/2)} \big[ \pi_{t}, F_{z} (\hat x) \big] \Big\|_{\tr} & \leq C \exp(c |t|)  \veps^{-2}  
%\\
%\sup _{z \in \Lambda, z' \in \R^{3}} (1+|z-z'|^8) \Big\| \W_{z'}^{(n)} \big[ \pi_{t}, F_{z} (\hat x) \big] \Big\|_{\tr} & \leq C \exp(c|t|)  \veps^{-2} 
%\\ 
%\sup _{z,z' \in \R^{3}} \big( \W^{(k)}_{z}(z')\big)^{-1}\Big\| \W_{z}^{(k)} \big[ \sqrt{\omega_{t}}, F_{z'}(\hat x) \big]\Big\|_{\tr} & \leq C\exp(c|t|)  \veps^{-2} %\;,
%\\ 
%\sup _{z,z' \in \R^{3}} \big( \W^{(k)}_{z}(z')\big)^{-1}\Big\| \W_{z}^{(k)} \big[ \sqrt{1-\omega_{t}}, F_{z'}(\hat x) \big]\Big\|_{\tr} & \leq C\exp(c|t|)  
%\veps^{-2}  \;.
\end{split}
\end{equation}
for all $t \in \bR$ and for all $n \in \N$ large enough.
\end{corollary} 
\begin{proof}
%To begin, observe that it is enough to prove (\ref{eq:commF}) for $n=4$, since $\mathcal{W}^{(n/2)} \leq 2\mathcal{W}^{(2)}$ for $n\geq 4$. 
Let $\widetilde{F}_{z}(x) = \big(\W_{z}^{(n)}(x)\big)^{-1} F_{z}(x)$. We start by estimating:
\begin{equation}\label{eq:dimcor}
\begin{split}
\Big\| \W_{z'}^{(n/2)} \big[ \pi_{t}, F_{z} (\hat x) \big] \Big\|_{\tr} &= \Big\| \W_{z'}^{(n/2)} \big[ \pi_{t}, \W_{z}^{(n)} \widetilde{F}_{z} (\hat x) \big] \Big\|_{\tr} \\
&\leq  \Big\| \W_{z'}^{(n/2)} \W_{z}^{(n)}\big[ \pi_{t}, \widetilde{F}_{z} (\hat x) \big] \Big\|_{\tr} + \Big\| \W_{z'}^{(n/2)} \big[ \pi_{t}, \W_{z}^{(n)} \big]  \widetilde{F}_{z} (\hat x)\Big\|_{\tr}\;.
\end{split}
\end{equation}
The first term can be bounded by 
%Consider the first term. By the properties of $\widetilde{F}_{z}$ and of $\mathcal{W}_{z}^{(m)}$, we have 
%
\begin{equation}\label{eq:411}
\begin{split}
&\Big\| \W_{z'}^{(n/2)} \W_{z}^{(n)}\big[ \pi_{t}, \widetilde{F}_{z} (\hat x) \big] \Big\|_{\tr} \\
&\quad \leq \frac{C}{1 + |z-z'|^{2n}} \Big\| \W_{z}^{(n/2)}\big[ \pi_{t}, \widetilde{F}_{z} (\hat x) \big] \Big\|_{\tr} \\
&\quad \leq \frac{C}{1+|z-z'|^{2n}} \int dp\, | \widetilde{F}(p) |(1+|p|) \sup_{p\in \mathbb{R}^{3}}\frac{1}{1+|p|}\Big\| \W_{z}^{(n/2)}\big[ \pi_{t}, e^{ip\cdot \hat{x}} \big] \Big\|_{\tr}\\
&\quad \leq \frac{C\exp(c|t|) \varepsilon^{-2}}{1+|z-z'|^{2n}}\;,
\end{split}
\end{equation}
so that Proposition \ref{prop: semiclassical_prop-main} applies, and assuming that (\ref{eq:Fass}) holds for $|\alpha| \leq \ell$, for $\ell$ large enough (depending on $n$). Consider now the second term in (\ref{eq:dimcor}). We estimate it as:
\begin{equation}
\Big\| \W_{z'}^{(n/2)} \big[ \pi_{t}, \W_{z}^{(n)} \big]  \widetilde{F}_{z} (\hat x)\Big\|_{\tr} \leq C \Big\| \W_{z'}^{(n/2)} \big[ \pi_{t}, \W_{z}^{(n)} \big]  \mathcal{W}^{(n)}_{z}\Big\|_{\tr}\;,
\end{equation}
We have:
\begin{equation}\label{eq:412}
\begin{split}
&\Big\| \W_{z'}^{(n/2)} \big[ \pi_{t}, \W_{z}^{(m)} \big]  \mathcal{W}^{(n)}_{z}\Big\|_{\tr} \\
&\qquad \leq \int dp\, |\hat \W_{z}^{(m)}(p)| (1 + |p|) \sup_{p\in \mathbb{R}^{3}} \frac{1}{1+|p|} \Big\| \W_{z'}^{(n/2)} \big[ \pi_{t}, e^{ip\cdot \hat{x}} \big]  \mathcal{W}^{(n)}_{z}\Big\|_{\tr} \\
&\qquad \leq \frac{C\exp(c|t|) \varepsilon^{-2}}{1+|z-z'|^{2n}}\;,
\end{split}
\end{equation}
by Proposition \ref{prop: semiclassical_prop-main}; note that $\ell$ has to be taken large enough in (\ref{eq:Fass}), to make sure that $|\widetilde{F}| \leq C \W^{(n)}$. Eqs.~(\ref{eq:411}), (\ref{eq:412}) imply (\ref{eq:commF}). 
\end{proof}
\begin{remark}\label{rem:k=0} 
The bound (\ref{eq:commF}) also implies that 
\[ 
\sup_{z \in \Lambda} \Big\| \big[ \pi_{t}, F_{z} (\hat x) \big]\Big\|_{\tr}  \leq C\exp(c|t|)  \veps^{-2}\;.
\] 
This can be shown proceeding similarly as we did in Remark \ref{rem:propa}.
\end{remark} 
%
%

%Finally, we conclude the section by stating and proving an a priori bound on the density of $\omega_{t}$, evolved under the Hartree dynamics.
%
%\begin{lemma}[A priori bound of the density of $\omega_{t}$]\label{pro: a priori-propagation}
%Under the assumptions, we have
%%
%\begin{equation}\label{pro: a priori-propagation}
%\sup_{z \in \R^{3}} \tr\, \W_{z}^{(n)} \omega_{t} \leq C \exp(c|t|) \veps^{-3} \;.
%\end{equation}
%%
%\end{lemma}
%%
%\begin{proof}
%To begin, we observe that the bound holds true at $t=0$, for all $\beta \geq 1$:
%%
%\begin{equation}\label{eq:aprioriomega}
%\sup_{z \in \R^{3}} \tr\, \W_{z}^{(n)} \frac{1}{1 + e^{\beta (H-\mu)}} \leq C \veps^{-3}\;.
%\end{equation}
%%
%We postpone the proof of (\ref{eq:aprioriomega}) to Appendix \ref{sec: appen-apriori}. Then, to prove the bound at later times, we use that:
%%
%\begin{equation}
%\tr\, \W_{z}^{(n)} \omega_{t} = \tr\, U_{\text{H}}(t;0)^{*} \W_{z}^{(n)} U_{\text{H}}(t;0) \omega_{0} \leq Ce^{c|t|} \tr\, \W_{z}^{(n)} \omega_{0}\;,
%\end{equation}
%%
%where the last inequality follows from Proposition \ref{prop: evolution_loc_oper-main}. This concludes the proof of Lemma \ref{pro: a priori-propagation}.
%\end{proof}

\subsection{Fluctuation dynamics}\label{sec: proof-main-result}

The proof of Theorem \ref{thm:T>0} relies on the study of the time-evolved vector $\psi_t = e^{-i\mathcal{L} t/\veps} \mathcal{R}_\omega \Xi$. In order to compare $\psi_t$ with the quasi-free state associated with the solution $\omega_t$ of the Hartree equation (\ref{eq:hartree2}), it is useful to introduce the fluctuation dynamics 
\begin{equation}\label{eq:fluc-dyn} \mathcal{U} (t;s) =   \RR _{\omega_t}^{\ast} e^{- i \mathcal{L} (t-s) / \veps} \RR_{\omega_s} \, . \end{equation} 
It follows that \begin{equation}\label{eq:ext-v} \psi_t = \RR_{\omega_t} \mathcal{U} (t;0) \Xi = \RR_{\omega_t} \Xi_t \end{equation} 
with $\Xi_t = \mathcal{U} (t;0) \Xi$ describing the excitations, at time $t \in \bR$, of the quasi-free state with reduced density $\omega_t$. To show that $\psi_t$ is  locally close to the quasi-free state with reduced density $\omega_t$ and conclude the proof of Theorem \ref{thm:T>0}, it is enough to prove that the local density of particles in $\Xi_t$ ({\it i.e.} the local density of the excitations of the quasi-free state) is small, compared to the local density of the quasi-free state, which is of the order $\veps^{-3}$. 

\begin{proposition}[Growth of local density of excitations]\label{thm:bound-localized-fluctuation} 
Under the assumption of Theorem~\ref{thm:T>0}, there exist constants $C,c > 0$ independent of $\Lambda$ such that 
\begin{equation} \label{eq:XidGXi}  
\langle \Xi_{t}, \,  d \Gamma(\W^{(n)}_{z})  \Xi_{t}  \rangle \leq C  \exp (c \exp c|t|) \veps^{-3+\delta} 
\end{equation} 
for all $n \in \bN$ large enough, for all $t \in \bR$ and all $z \in \Lambda$ with $B_{\veps^{-\delta} |t|} (z) \subset \Lambda$. 
\end{proposition} 
\begin{remark}
Note that by the monotonicity of the functions $\W_{z}^{(n)}$ the bound \eqref{eq:XidGXi} holds indeed for constants $C,c>0$ independent of $n$.
\end{remark}

With Proposition \ref{thm:bound-localized-fluctuation}, we can conclude the proof of Theorem \ref{thm:T>0}. 
\begin{proof}[Proof of Theorem \ref{thm:T>0}] 
Using (\ref{eq:dG}), recalling that $\psi_t = \RR_{\omega_{t}}\Xi_t$ and applying Lemma \ref{lemma: Bogolubov-observables}, we have 
\begin{equation}\label{eq: reduction-localisation}
\begin{split}
 \tr \, &\mathcal{O}_{z}  (\gamma^{(1)}_{t} - \omega_{t})
 \\
 &= \big\langle \Xi_{t}, \big( d \Gamma_{l}(u_{t}\mathcal{O}_{z}u_{t}) - d \Gamma _{r} (\overline{v_{t}} \mathcal{O}_{z}  \overline{v_{t}}) 
- d \Gamma^{+}_{l r}(u_{t} \mathcal{O}_{z} v_{t}) -d \Gamma^{-}_{r l}(v_{t} \mathcal{O}_{z} u_{t}) \big) \Xi_{t} \big\rangle
\end{split}
\end{equation}
with $u_t = \sqrt{1-\omega_t}, v_t = \sqrt{\omega_t}$. 
Consider the first term in the right-hand side. Since $| (x - z)^{\alpha} \mathcal{O}_{z} | \leq \int dp |D^{\alpha}  \hat{\mathcal{O}}(p)|$, by assumption \eqref{eq: assumptions-O} on $\mathcal{O}$, it follows that $|\mathcal{O}_{z}| \leq C\mathcal{W}^{(n)}_{z}$,
which implies $d \Gamma_{l}(u_{t}\mathcal{O}_{z}u_{t})  \leq Cd \Gamma_{l}(u_{t} \mathcal{W}^{(n)}_{z}u_{t})$. Furthermore, 
\begin{equation}\label{eq:uWu-comm}  \begin{split}  u_t \W^{(n)}_z u_t  &= \big( (\W_z^{(n)})^{1/2} u_t + \big[ u_t,  (\W_z^{(n)})^{1/2} \big] \big) \big( u_t (\W_z^{(n)})^{1/2} + \big[ (\W_z^{(n)})^{1/2} , u_t \big] \big) \\ &\leq 2 (\W_z^{(n)})^{1/2} u_t^2 (\W_z^{(n)})^{1/2} + 2  \big|\big[ u_t,  (\W_z^{(n)})^{1/2} \big] \big|^2 \\ &\leq 2 \W_z^{(n)} +  2  \big|\big[ u_t,  (\W_z^{(n)})^{1/2} \big] \big|^2\;;  \end{split}  \end{equation} 
we conclude with Lemma \ref{lemma: bound-Fock-operator} and with the remark after Corollary \ref{cor:commF} (using $\| \W_z^{(n)} \|_\text{op} \leq C$, $\| u_t \|_\text{op} \leq 1$, setting $F = (\W^{(n)})^{1/2}$ and choosing $n$ large enough, so that (\ref{eq:Fass}) holds true) that:
\begin{equation}\label{eq:uWu2} d\Gamma_{l} (u_t \mathcal{O}_{z} u_t) \leq C d\Gamma_{l} (\W_z^{(n)}) + C \big\| \big[ u_t,  (\W_z^{(n)})^{1/2} \big] \big\|_\text{tr} \leq C d\Gamma_{l} (\W_z^{(n)} )+ C \veps^{-2} \exp (c |t|)\;. \end{equation} 
Similarly, we can bound the second term in (\ref{eq: reduction-localisation}) as:
\begin{equation}
- d \Gamma _{r} (\overline{v_{t}} \mathcal{O}_{z}  \overline{v_{t}}) \leq  d \Gamma _{r} (\overline{v_{t}} |\mathcal{O}_{z}|  \overline{v_{t}}) \leq C d\Gamma_{l} (\W_z^{(n)} )+ C \veps^{-2} \exp (c |t|)\;.
\end{equation}
The other terms in \eqref{eq: reduction-localisation} can be controlled by Lemma \ref{lemma: bound-Fock-operator}, Remarks \ref{rem:propa} and \ref{rem:k=0}. For example,
\begin{equation*}
\begin{split}
\big| \langle \Xi_{t}, d \Gamma^{+}_{l r}(u_{t} \mathcal{O}_{z} v_{t}) 
\Xi_{t}  \rangle \big|
& \leq C \big\| u_{t} \mathcal{O}_{z} v_{t}\big\|_{\tr}  
\\
& \leq C \big\| \big[u_{t}, \mathcal{O}_{z} \big] \big\|_{\tr} + C \big\| \mathcal{O}_{z}  u_{t} v_{t}\big\|_{\tr} \leq C \exp(c|t|) \veps^{-2} \;.
\end{split}
\end{equation*}
where in the last step we used that $\mathcal{O}_{z}$ fulfills the assumptions of Corollary \ref{cor:commF}, and that $|\mathcal{O}_{z}| \leq C\mathcal{W}_{z}$. All in all, we obtain:
\begin{equation*}
\Big| \tr \, \mathcal{O}_{z} (\gamma^{(1)}_{t} - \omega_{t}) \Big| \leq C\exp(c|t|) \veps^{-2} + C \langle \Xi_{t}, d \Gamma(\W^{(n)}_z) \Xi_{t}  \rangle \; ,
\end{equation*}
if $n \in \bN$ is large enough. Theorem \ref{thm:T>0} now follows from Proposition \ref{thm:bound-localized-fluctuation}. 
\end{proof}

\subsection{Growth of local density of excitations}

This section is devoted to the proof of Proposition \ref{thm:bound-localized-fluctuation}. To control the growth of the local density associated with the  excitation vector $\Xi_t$, we proceed in two steps. First we provide a rough a priori bound of the order $\veps^{-3}$, valid everywhere. Then, we use this a priori estimate to establish the improved bound (\ref{eq:XidGXi}), valid in a subset of $\Lambda$, where we know from (\ref{eq:assT>0}) that the initial datum has few excitations. 
\begin{lemma}[Rough Bound on Density of Excitations]\label{prop: a priori bound}
Let $0 \leq \omega \leq 1$ satisfy Assumption~\ref{ass:bd} (bounded density). Let $\Xi \in \mathcal{F} (\frak{h} \oplus \frak{h})$ satisfy the first bound in (\ref{eq:assT>0}). Let $\Xi_t$ be defined as in \eqref{eq:ext-v}. Then, there exist constants $C,c > 0$ such that 
\begin{equation*}
\langle \Xi_{t}, \,  d \Gamma_\sigma (\W^{(n)}_{z})  \Xi_{t}  \rangle 
\leq  C \exp ( c|t|) \veps^{-3}
\end{equation*}
for all $n \in \bN$ large enough, for all $z \in \bR^3$ and for $\sigma = l,r$. 
\end{lemma}
\begin{proof}
We consider the case $\sigma = l$ (the case $\sigma =r$ can be treated analogously). Also note that we can fix some $n \geq 4$ large enough, the bound for $\bar{n}>n$ following by monotonicity. First of all, recalling from (\ref{eq:ext-v}) that $\Xi_{t} = \mathcal{R}_{\omega_{t}}^{*} \psi_{t}$, we obtain, by Lemma \ref{lemma: Bogolubov-observables} (exchanging $\mathcal{R}_\omega$ and $\mathcal{R}_\omega^*$ is equivalent to replacing $v$ with $-v$),  
\begin{equation}\label{eq:apri1}
\begin{split}
\big \langle &\Xi_{t}, \dd \Gamma_{l} (\W^{(n)}_{z}) \Xi_{t} \big \rangle 
\\
& = \big \langle \psi_{t}, \big( d \Gamma_{l} (u_{t}\W^{(n)}_{z}u_{t}) - d \Gamma _{r} (\overline{v_{t}} \W^{(n)}_{z}  \overline{v_{t}}) 
+d \Gamma^{+}_{l r} (u_{t} \W^{(n)}_{z} v_{t}) +d \Gamma^{-}_{r l} (v_{t} \W^{(n)}_{z} u_{t}) \big) \psi_{t} \big \rangle \;.
\end{split}
\end{equation}
We can bound the third term by 
\begin{equation}\label{eq:apri2}
\begin{split}
\big| \big \langle \psi_{t},d \Gamma^{+}_{l r}(u_{t} \W^{(n)}_{z} v_{t})\psi_{t} \big \rangle \big|   
& = \Big| \int d y \, \W^{(n)}_{z}(y) \big \langle \psi_{t}, a^{*}_{l}(u_{t;y})a^{*}_{r}(\overline{v_{t;y}})\psi_{t} \big \rangle \Big| 
\\
&\leq \int d y \, \W^{(n)}_{z}(y) \big \| a_{l}(u_{t;y}) \psi_{t} \big\| \big\| a^{*}_{r}(\overline{v_{t;y}})\psi_{t} \big \|
\\
& \leq \Big(\int d y \, \W^{(n)}_{z}(y)  \|v_{t;y}\|_{\mathfrak{h}}^{2} \Big)^{1/2}
\Big(\int d y \, \W^{(n)}_{z}(y) \big \| a_{l}(u_{t;y}) \psi_{t} \big\|^{2} \Big)^{1/2}
\\
& \leq \big( \tr \, \W^{(n)}_{z} \omega_{t} \big)^{1/2}
 \big \langle\psi_{t}, d \Gamma_{l}(u_{t}\W^{(n)}_{z}u_{t}) \psi_{t} \big \rangle ^{1/2} ,
\end{split}
\end{equation}
where we used the Cauchy-Schwarz inequality twice and $v_{t}^{2}= \omega_{t}$. The fourth term on the r.h.s.~of \eqref{eq:apri1} is controlled similarly. The second term is negative and can be neglected for the purpose of an upper bound.
%
%As for the second term, we find by Lemma \ref{lemma: bound-Fock-operator}, {\bf [e' negativo]} 
%\[ \langle \psi_t, d \Gamma _{r} (\overline{v_{t}} \W^{(n_0)}_{z}  \overline{v_{t}})  \psi_t \rangle \leq 2 \| \overline{v_{t}} \W^{(n_0)}_{z}  \overline{v_{t}} \|_\text{tr} = 2 \, \tr \, \W^{(n_0)}_z %\omega_t \]
To control the first term on the r.h.s.~of (\ref{eq:apri1}) (and therefore also the r.h.s.~of (\ref{eq:apri2})), we observe that
\begin{equation*}
\begin{split}
u_t \W^{(n)}_z u_t &= u_t \W^{(n)}_z (u_t - 1) + (u_t -1) \W^{(n)}_z + \W^{(n)}_z 
\\
&\leq \frac{1}{2} u_t \W^{(n)}_z u_t +  (u_t - 1) \W^{(n)}_z (u_t -1) + 2 \W^{(n)}_z
\end{split}
\end{equation*}
which implies that 
\begin{equation}\label{eq:uWu3} u_t \W^{(n)}_z u_t \leq 2 (u_t - 1) \W^{(n)}_z (u_t - 1) + 4 \W^{(n)}_z \end{equation} 
and thus that 
\begin{equation}\label{eq:uWu3b}  \begin{split}  \langle \psi_t , d \Gamma_{l}(u_{t}\W^{(n)}_{z}u_{t})  \psi_t \rangle &\leq 4 \langle \psi_t, d\Gamma_l (\W^{(n)}_z) \psi_t \rangle + 2 \, \tr \, \W^{(n)}_z \omega_t \\ &\leq 4 \langle \psi_t, d\Gamma_l (\W^{(n)}_z) \psi_t \rangle + C  \veps^{-3} \exp (c|t|) 
\end{split} 
\end{equation} 
for $n \in \bN$ large enough. Here we used $(1-u_t)^2 = (1 - \sqrt{1-\omega_t})^2 \leq \omega_t$ and then we applied the bound \eqref{pro: a priori-propagation}. We are left with controlling $\big \langle\psi_{t}, d \Gamma _{l} (\W^{(n)}_{z} ) \psi_{t} \big \rangle$. We use a Gronwall argument. We have
\begin{equation*}
i \veps \partial_{t} \big \langle\psi_{t}, d \Gamma_{l} (\W^{(n)}_{z} ) \psi_{t} \big \rangle =  \big \langle\psi_{t},d \Gamma_{l} \big( \big[ \sqrt{1-\veps^{2}\Delta} \,, \W^{(n)}_{z} \big] \big) \psi_{t} \big \rangle \;,
\end{equation*}
where we used that $ \big[d \Gamma_{l}(A), d \Gamma_{l}(A) \big] = d \Gamma_{l}([A,B]) $, and that the two-body interaction commutes with $d \Gamma_{l} (\W^{(n)}_{z})$. We claim that 
\begin{equation}\label{eq:apri3}
\big\| \big(\W^{(n)}_{z}\big)^{-1/2} \big[\W^{(n)}_{z}, \sqrt{1 - \veps^{2} \Delta} \big] \big(\W^{(n)}_{z}\big)^{-1/2}\big\|_{\op} \leq C \veps\;,
\end{equation}
for some constant $C > 0$ depending on $n$ only. In fact, we can estimate 
\[ \begin{split}  \big\| &\big(\W^{(n)}_{z}\big)^{-1/2} \big[\W^{(n)}_{z}, \sqrt{1 - \veps^{2} \Delta} \big] \big(\W^{(n)}_{z}\big)^{-1/2}\big\|_{\op} \\ & \leq 2 
\big\| \big[ (\W^{(n)}_{z})^{1/2} , \sqrt{1 - \veps^{2} \Delta} \big] \big(\W^{(n)}_{z}\big)^{-1/2}\big\|_{\op} \\ & \leq 2 
\big\| \big[ (\W^{(n)}_{z})^{1/2} , \sqrt{1 - \veps^{2} \Delta} \big] |\hat{x} - z|^{2n} \big\|_{\op}  \\ & \leq 2 \big\| (\W^{(n)}_{z})^{1/2} \big[ \sqrt{1 - \veps^{2} \Delta} , |\hat{x} - z|^{2n} \big] \big\|_{\op} \leq C \sum_{1 \leq |\alpha| \leq 2n} \| \text{ad}^{\alpha}_{\hat{x}} (\sqrt{1-\veps^2 \Delta}) \|_\text{op} \leq C \veps \end{split} \]
because $\| (\W^{(n)}_{z})^{1/2} |\hat{x} - z|^{2n-|\alpha|} \|_\text{op} \leq C$ for all $1 \leq |\alpha| \leq 2n$. From (\ref{eq:apri3}), we obtain 
\begin{equation}\label{eq:comm-kin}
i \big[\W_{z}^{(n)}, \sqrt{1 - \veps^{2} \Delta} \big] \leq C \veps \W^{(n)}_{z} \;.
\end{equation}
By Gronwall's lemma, we conclude that 
\begin{equation*}
\big \langle\psi_{t}, d \Gamma_{l} (\W^{(n)}_{z} ) \psi_{t} \big \rangle \leq
\exp(c|t|) \, \big \langle\psi, d \Gamma_{l} (\W^{(n)}_{z} ) \psi \big \rangle  \;.
\end{equation*}
Recalling that $\psi = \RR_\omega \Xi$ and applying once again Lemma \ref{lemma: Bogolubov-observables}, we find 
\begin{equation*}
\begin{split}
\big \langle\psi,& d \Gamma_{l} (\W^{(n)}_{z} ) \psi \big \rangle 
\\
&= \big \langle \Xi, \big( d \Gamma_{l}(u \W^{(n)}_{z}u) - d \Gamma _{r} (\overline{v} \W^{(n)}_{z}  \overline{v}) 
-d \Gamma^{+}_{l r}(u \W^{(n)}_{z} v) - d \Gamma^{-}_{r l}(v \W^{(n)}_{z} u) \big) \Xi \big \rangle \;.
\end{split}
\end{equation*}
Proceeding as we did to handle (\ref{eq:apri1}), and using Assumption \ref{ass:bd} to bound $\tr \, \W^{(n)}_z \omega$ and the first estimate in (\ref{eq:assT>0}) to control $\big\langle \Xi, d\Gamma(\W^{(n)}_{z}) \Xi\big\rangle$, we conclude the proof of the lemma.
\end{proof}

The next step in the proof of Prop.~\ref{thm:bound-localized-fluctuation} is the computation of the generator of the fluctuation dynamics 
(\ref{eq:fluc-dyn}), defined as 
\begin{equation}\label{eq: generator-fluct-dyn}
\mathcal{G}(t) :=\big( i \varepsilon \partial_t \U (t;s)\big)  \U (t;s)^* = \big( i \veps \partial_t \,  \RR_{\omega_t}^* \big) \RR_{\omega_t} + \RR_{\omega_t}^* \mathcal{L} \, \RR_{\omega_t} 
\end{equation}
so that $\U (t;s)$ satisfies the Schr\"odinger equation \[ i \veps \partial_t \U (t;s) = \mathcal{G} (t) \U (t;s) \] 
with $\U (s;s) = 1$ for all $s \in \bR$. 
\begin{proposition}[Generator of the Fluctuation Dynamics] \label{prop:fluctuation-generator}
The generator $\mathcal{G}(t)$, defined in (\ref{eq: generator-fluct-dyn}), is a self-adjoint operator on $\F(\h \oplus \h)$ and can be written as
\begin{equation}\label{eq:generator-explicit}
\mathcal{G}(t) = d \Gamma_{l}(h_{\mathrm{H}}(t)) - d \Gamma_{r}(\overline{h_{\mathrm{H}}(t)})+ \mathcal{C}(t) + \mathcal{Q}(t)
\end{equation}
where, introducing the notation
\begin{equation}
u_{t;x}(\cdot):= u_{t}(\cdot;x) \;, \qquad v_{t;x}(\cdot):= v_{t}(\cdot;x) \;,
\end{equation}
we have (collecting in $\mathcal{C} (t)$ contributions that commute with the number of particles operator) 
\begin{equation}\label{eq:generator-C}
\begin{split}
\mathcal{C}(t) := \frac{\veps^{3}}{2} \int
d x \,&d y \, V(x-y) 
\Big(
a^{\ast} _{l}(u_{t;x}) a^{\ast} _{l}(u_{t;y}) a_{l} (u_{t;y}) a_{l} (u_{t;x})
+
2 a^{*}_{l}(u_{t;x}) a^{*}_{r}(\overline{v_{t;x}})a_{r}(\overline{v_{t;y}})a_{l}(u_{t;y})
\\
&- 2 a^{\ast}_{l} (u_{t;x}) a^{\ast}_{r}
   (\overline{v_{t;y}})  a_{r} (\overline{v_{t;y}}) a _{l}(u_{t;x})
   + a^{\ast}_{r} (\overline{v_{t;y}}) a^{\ast}_{r} (\overline{v_{t;x}}) a_{r} (\overline{v_{t;x}}) a_{r} (\overline{v_{t;y}}) 
   \\
   &
   -a^{*}_{r}(\overline{u_{t;x}}) a^{*}_{r}(\overline{u_{t;y}})  a_{r}(\overline{u_{t;y}}) a_{r}(\overline{u_{t;y}}) 
   - 2 a^{*}_{r}(\overline{u_{t;x}}) a^{*}_{l}(v_{t;x}) a_{l}(v_{t;y}) a_{r}(\overline{u_{t;y}})
   \\
   & 
  + 2 a^{\ast}_{r} (\overline{u_{t;x}}) a^{\ast}_{l} (v_{t;y}) a_{l} (v_{t;y}) a_{r} (\overline{u_{t;x}}) 
  - a^{*}_{l}(\overline{v_{t;y}}) a^{*}_{l}(\overline{v_{t;x}}) a_{l}(\overline{v_{t;x}}) a_{l}(\overline{v_{t;y}})
\\
& 
 \quad + 2 a^{\ast}_{r} (\overline{v_{t;y}}) a_{r} (\overline{v_{t;x}}) \omega_{t}(x ; y)   - 2 a^{\ast}_{l} (u_{t;x}) a_{l} (u_{t;y}) \omega_{t} (x ; y) \Big)
\end{split}
\end{equation}
and 
\begin{equation}\label{eq:generator-Q}
\begin{split}
\mathcal{Q}(t) :=   \frac{\veps^{3}}{2} \int &
d x \, d y  \, V(x-y) 
 \Big( a^{*}_{l}(u_{t;x})a^{*}_{l}(u_{t;y})a^{*}_{r}(\overline{v_{t;y}})a^{*}_{r}(\overline{v_{t;x}})
 +  2 a^{*}_{l}(u_{t;x})a^{*}_{l}(u_{t;y})a^{*}_{r}(\overline{v_{t;x}})a_{l}(u_{t;y})
\\ 
 &   -2  a^{*}_{l}(u_{t;x})a^{*}_{r}(\overline{v_{t;y}}) a^{*}_{r}(\overline{v_{t;x}}) a_{r}(\overline{v_{t;y}})
 -a^{*}_{r}(\overline{u_{t;x}})a^{*}_{r}(\overline{u_{t;y}}) a^{*}_{l}(v_{t;y})a_{l}^{*}(v_{t;x})
 \\
 & + 2 a^{*}_{r}(\overline{u_{t;x}})a^{*}_{r}(\overline{u_{t;y}})a^{*}_{l}(v_{t;x}) a_{r}(\overline{u_{t;y}}) 
 -2 a^{*}_{r}(\overline{u_{t;x}})a^{*}_{l}(v_{t;y}) a^{*}_{l}(v_{t;x})a_{l}(v_{t;y})
 \\
 & +
 2 a^{\ast}_{l} (u_{t;x}) a^{\ast}_{r} (\overline{v_{t;y}}) \omega_{t} (x ; y)  + h.c. \Big) \;.
\end{split}
\end{equation}
\end{proposition}
\begin{proof} Proposition \ref{prop:fluctuation-generator} is proved in \cite[Prop.~3.1]{BJPSS}, in the mean-field regime, and with $\omega_t$ solving the Hartree-Fock equation. Since here we let $\omega_t$ evolve according to the Hartree, rather than the Hartree-Fock dynamics (so, we neglect the exchange term appearing in the Hartree-Fock equation), the operators $\mathcal{C}(t)$ and $\mathcal{Q}(t)$ contain additional quadratic contributions. 
\end{proof}

To prove Proposition \ref{thm:bound-localized-fluctuation}, we are going to use a Gronwall argument, estimating the increments of the local density of the excitations. Therefore, we need to control the commutator of the operator $d\Gamma (\W^{(n)}_z)$ with the generator $\mathcal{G} (t)$ of the fluctuation dynamics. In the next lemma, we collect some identities which will be useful to compute this commutator. 

\begin{lemma}\label{lemma: convolution}
If $A$ is a bounded operator on $\h$, we set $A_{x}(\cdot):= A(\cdot;x)$. Let $V$ satisfy the assumptions of Theorem \ref{thm:main} and set
\begin{equation}\label{def_Vj}
V^{(1)}(x)  := \int_{\R^{3}} d  p \,\frac{e^{i p \cdot x}}{1+|p|^{6}} \;, \qquad V^{(2)}(x)  := \int_{\R^{3}} d p \, e^{i p \cdot x}\,  (1+|p|^{6})\widehat{V}(p) \;.
\end{equation}
Then, for any bounded operators $J,A,B,C,D$ on $\h$ and any $\rho, \sigma, \sigma', \bar{\sigma}, \bar{\sigma}'= r,l$ the following identities hold true:

\begin{equation}\label{eq: commutator-convolution-I}
\begin{split}
\Big[ d \Gamma_{\rho}&\big( J \big),\int d x d y V(x-y) a^{\ast}_{\sigma}(A_{x}) a^{\ast}_{\sigma '}(B_{y}) a_{\sigma'} (C_{y}) a_{\sigma} (D_{x})\Big] 
\\
&
= \delta_{\rho,\sigma }\int d z \int d y V^{(1)}_{z}(y) a^{\ast}_{\sigma '}(B_{y})
d \Gamma_{\sigma} \big( \big[J,A V^{(2)}_{z} D^{*}\big] \big) a_{\sigma'} (C_{y})
\\
& \qquad +  \delta_{\rho,\sigma' }\int d z \int d x V^{(1)}_{z}(x) a^{\ast}_{\sigma}(A_{x})
d \Gamma_{\sigma'} \big( \big[J,B V^{(2)}_{z} C^{*}\big] \big) a_{\sigma} (D_{x}) \;,
\end{split}
\end{equation}
\begin{equation}\label{eq: decomposition-convolution-II}
\begin{split}
\Big[ d & \Gamma_{\rho}\big( J\big),
\int d x d y V(x-y) a^{\ast}_{\sigma}(A_{x}) a^{\ast}_{\sigma '}(B_{x}) a^{\ast}_{\bar{\sigma}} (C_{y}) a_{\bar{\sigma}} (D_{y}) \Big] \\
& =
\int d z  \int d y V^{(2)}_{z}(y) a^{*}_{\bar{\sigma}}(C_{y})d \Gamma_{\sigma \sigma'}^{+}\big(\delta_{\rho,\sigma} J A V^{(1)}_{z} B^{T}+\delta_{\rho,\sigma'}A V^{(1)}_{z} B^{T} J^{T}\big) a_{\bar{\sigma}}(D_{y})
\\
& \qquad + \delta_{\rho,\bar{\sigma}}
\int d z d \Gamma^{+}_{\sigma \sigma'}\big( A V^{(1)}_{z} B^{T}\big) d \Gamma_{\bar{\sigma}}\big(\big[J, C V^{(2)}_{z} D^{*}\big] \big) \;.
\end{split}
\end{equation}
\begin{equation}\label{eq: decomposition-convolution-I}
\begin{split}
\Big[ d \Gamma_{\rho}\big( J\big)&,\int d x d y V(x-y)  a^{\ast}_{\sigma}(A_{x}) a^{\ast}_{\sigma '}(B_{x}) a_{\sigma'} (C_{y}) a_{\sigma} (D_{y}) \Big]
\\
& = \int d z  d \Gamma^{+}_{\sigma \sigma'}\big( \delta_{\rho,\sigma} J A V^{(1)}_{z} B^{T} + \delta_{\rho,\sigma'}A V^{(1)}_{z} B^{T}J^{T}\big) 
d \Gamma^{-}_{\sigma' \sigma}\big( \overline{C} V^{(2)}_{z} D^{*}\big) 
\\
& \qquad - \int d z d \Gamma^{+}_{\sigma \sigma'}\big( A V^{(1)}_{z} B^{T}\big) d \Gamma^{-}_{\sigma' \sigma}\big( \delta_{\rho, \sigma'}J^{T}\overline{C} V^{(2)}_{z} D^{*}+\delta_{\rho,\sigma}\overline{C} V^{(2)}_{z} D^{*} J\big) \;,
\end{split}
\end{equation}
\begin{equation}\label{eq: decomposition-convolution-III}
\begin{split}
\Big[ d \Gamma_{\rho}\big( J\big),&
\int d x d y V(x-y) a^{\ast}_{\sigma}(A_{x}) a^{\ast}_{\sigma '}(B_{x}) a^{\ast}_{\bar{\sigma}} (C_{y}) a^{\ast}_{\bar{\sigma}'} (D_{y})
\\ 
& = 
\int d z d \Gamma_{\sigma \sigma'}^{+}\big(\delta_{\rho,\sigma} J A V^{(1)}_{z} B^{T}+\delta_{\rho,\sigma'}A V^{(1)}_{z} B^{T} J^{T}\big) d \Gamma_{\bar{\sigma} \bar{\sigma}'}^{+}\big( C V^{(2)}_{z} D^{T}\big) 
\\
& \qquad + \int d z d \Gamma_{\sigma \sigma'}^{+}\big( A V^{(1)}_{z} B^{T}\big) d \Gamma_{\bar{\sigma} \bar{\sigma}'}^{+}\big( \delta_{\rho,\bar{\sigma}} J C V^{(2)}_{z} D^{T} + \delta_{\rho,\bar{\sigma}'} C V^{(2)}_{z} D^{T}J^{T}\big) \;.
\end{split}
\end{equation}
\end{lemma}
The proof of Lemma \ref{lemma: convolution} relies on the following relations, whose proof is a simple application of the canonical anticommutation relations (\ref{eq:CAR}). 
\begin{lemma}\label{lemma: commutators-dGamma}
For any bounded operators $A,B$ on $\h$ the following identities hold true, for any $\bar{\sigma},\sigma, \sigma'= r,l$:
\begin{equation}\label{eq: commutators-dGamma-lemma}
\begin{split}
\big[ d \Gamma_{\bar{\sigma}}(A),  d \Gamma _{\sigma} (B)\big] & = \delta_{\bar{\sigma},\sigma} d \Gamma_{\sigma} \big( \big[A, B \big] \big)
\\
\big[ d \Gamma _{\bar{\sigma}}(A), d \Gamma^{+}_{\sigma \sigma'}(B) \big] & = d \Gamma_{\sigma \sigma'}^{+}\big( \delta_{\bar{\sigma}, \sigma} A B +\delta_{\bar{\sigma},\sigma'} B A^{T}\big)\;,
\\ 
\big[ d \Gamma _{\bar{\sigma}}(A), d \Gamma^{-}_{\sigma \sigma'}(B) \big] & =
-d \Gamma_{\sigma \sigma'}^{-}\big( \delta_{\bar{\sigma}, \sigma} A^{T} B +\delta_{\bar{\sigma},\sigma'} B A\big) \;.
\end{split}
\end{equation}
Furthermore, for any bounded operator $A$ on $\h$ and $f \in \h$ we have
\begin{equation}\label{eq: commut-dGamma-a}
[d \Gamma_{\sigma} (A), a^{\ast}_{\sigma} (f)] = a^{\ast}_{\sigma} (A f) \;, \qquad
[d \Gamma^{-}_{\sigma}(A), a^{\ast}_{\sigma} (f)] = -a_{\sigma} (A \bar{f}) \;.
\end{equation}
\end{lemma}
%
%\begin{proof}[Proof of Lemma \ref{lemma: commutators-dGamma}] The proof is a simple consequence of the CAR. For the sake of the reader, we %sketch the proof of the second identity in \eqref{eq: commutators-dGamma-lemma}. By the Leibniz rule for commutators and by further %straightforward algebraic manipulations, we have
%%
%\begin{equation*}
%\begin{split}
%[a^{*}_{\bar{\sigma},x}a_{\bar{\sigma},y}, a^{*}_{\sigma,x'}a^{*}_{\sigma',y'}] & = a^{*}_{\bar{\sigma},x} \big\{a_{\bar{\sigma},y}, a^{*}_{\sigma,x'}\big\}%a^{*}_{\sigma',y'} - a^{*}_{\bar{\sigma},x}a^{*}_{\sigma,x'} \big\{ a_{\bar{\sigma},y}, a^{*}_{\sigma',y'}\big\}
%\\
%& = \delta_{\bar{\sigma}, \sigma}\delta(x'- y) a^{*}_{\bar{\sigma},x}   a^{*}_{\sigma',y'} 
%- \delta_{\bar{\sigma},\sigma'} \delta(y- y') a^{*}_{\bar{\sigma},x}a^{*}_{\sigma,x'} \;,
%\end{split}
%\end{equation*}
%%
%where in the last line the CAR were used, implying
%%
%\begin{equation*}
%\begin{split}
%\big[ &d \Gamma _{\bar{\sigma}}(A), d \Gamma^{+}_{\sigma \sigma'}(B) \big] 
%\\ & = \int d x d y d x'd y' A(x;y) B(x';y') \Big( \delta_{\bar{\sigma}, \sigma}\delta(x'- y) a^{*}_{\bar{\sigma},x}   a^{*}_{\sigma',y'} 
%- \delta_{\bar{\sigma},\sigma'} \delta(y- y') a^{*}_{\bar{\sigma},x}a^{*}_{\sigma,x'}\Big)
%\\
%& = \delta_{\bar{\sigma}, \sigma} d \Gamma^{+}_{\sigma \sigma' }(AB) - \delta_{\bar{\sigma},\sigma'} d \Gamma^{+}_{\sigma' \sigma}(AB^{T})
%\\
%& = d \Gamma_{\sigma \sigma'}^{+}\big( \delta_{\bar{\sigma}, \sigma} A B +\delta_{\bar{\sigma},\sigma'} B A^{T}\big)\;.
%\end{split} 
%\end{equation*}
%%
%\end{proof}
%
\begin{proof}[Proof of Lemma \ref{lemma: convolution}.] We prove \eqref{eq: commutator-convolution-I} and \eqref{eq: decomposition-convolution-I}, since \eqref{eq: decomposition-convolution-II} and \eqref{eq: decomposition-convolution-III} follow in a similar way. By the Leibniz rule for commutators and by \eqref{eq: commut-dGamma-a} we have
\begin{equation}\label{eq: commut-dGammaJ-ABCD}
\begin{split}
\Big[ &d \Gamma_{\bar{\sigma}}\big( J \big), a^{\ast}_{\sigma}(A_{x}) a^{\ast}_{\sigma '}(B_{y}) a_{\sigma'} (C_{y}) a_{\sigma} (D_{x})\Big] 
\\
& = \delta_{\bar{\sigma} ,\sigma} a^{\ast}_{\sigma}(J A_{x}) a^{\ast}_{\sigma '}(B_{y}) a_{\sigma'} (C_{y}) a_{\sigma} (D_{x})
+ a^{\ast}_{\sigma}(A_{x}) \Big[ d \Gamma_{\bar{\sigma}} \big( J \big), a^{\ast}_{\sigma '}(B_{y}) a_{\sigma'} (C_{y})\Big]  a_{\sigma} (D_{x})
\\
& \qquad  -\delta_{\bar{\sigma} ,\sigma} a^{\ast}_{\sigma}(A_{x}) a^{\ast}_{\sigma '}(B_{y}) a_{\sigma'} (C_{y})   a_{\sigma} (J D_{x})
\\
& =  a^{\ast}_{\sigma '}(B_{y})
\Big[ d \Gamma_{\bar{\sigma}} \big( J \big), a^{\ast}_{\sigma }(A_{x}) a_{\sigma} (D_{x})\Big] a_{\sigma'} (C_{y}) 
+a^{\ast}_{\sigma}(A_{x}) \Big[ d \Gamma_{\bar{\sigma}} \big( J \big), a^{\ast}_{\sigma '}(B_{y}) a_{\sigma'} (C_{y})\Big]  a_{\sigma} (D_{x}) \;.
\end{split}
\end{equation}
From \eqref{def_Vj}, we can write
\begin{equation*}
\begin{split}
&\int d x d y V(x-y) a^{\ast}_{\sigma '}(B_{y})
\Big[ d \Gamma_{\bar{\sigma}} \big( J \big), a^{\ast}_{\sigma }(A_{x}) a_{\sigma} (D_{x})\Big] a_{\sigma'} (C_{y})
\\
&  =  \int d z \int d y V^{(1)}_{z}(y) a^{\ast}_{\sigma '}(B_{y})
\Big[ d \Gamma_{\bar{\sigma}} \big( J \big), \int d x V_{z}^{(2)}(x)a^{\ast}_{\sigma }(A_{x}) a_{\sigma} (D_{x})\Big]  a_{\sigma'} (C_{y}) \;.
\end{split}
\end{equation*}
Since
\begin{equation}\label{eq: representation dGamma-V}
\begin{split}
\int d x &V^{(2)}_{z}(x) a^{\ast}_{\sigma}(A_{x}) a_{\sigma }(D_{x}) 
%\\
%& = \int d r d r a^{\ast}_{\sigma;r} a_{\sigma; s} \int d x V^{(2)}_{z}(x) A(r;x) \overline{D(s;x)} 
= d \Gamma_{\sigma}\big( A V^{(2)}_{z} D^{*}\big) \;,
\end{split}
\end{equation}
by \eqref{eq: commutators-dGamma-lemma} we conclude that
\begin{equation*}
\begin{split}
&\int d x d y V(x-y) a^{\ast}_{\sigma '}(B_{y})
\Big[ d \Gamma_{\bar{\sigma}} \big( J \big), a^{\ast}_{\sigma }(A_{x}) a_{\sigma} (D_{x})\Big] a_{\sigma'} (C_{y})
\\
& = \delta_{\bar{\sigma}, \sigma} \int d z \int d y V^{(1)}_{z}(y) a^{\ast}_{\sigma '}(B_{y})
d \Gamma_{\sigma} \big( \big[J,A V^{(2)}_{z} D^{*}\big] \big) a_{\sigma'} (C_{y}) \;.
\end{split}
\end{equation*}
The same manipulations carry over to the second term on the last line in \eqref{eq: commut-dGammaJ-ABCD}, hence the claim. To prove \eqref{eq: decomposition-convolution-I}, we write
\begin{equation*}%\label{eq: proof-convolution}
\begin{split}
\int d x d y & V(x-y) \,  a^{\ast}_{\sigma}(A_{x}) a^{\ast}_{\sigma '}(B_{x}) a_{\sigma'} (C_{y}) a_{\sigma} (D_{y}) 
\\
& = \int d z \Big( \int d x V^{(1)}_{z}(x) a^{\ast}_{\sigma}(A_{x}) a^{*}_{\sigma '}(B_{x})\Big) \Big( \int d y V^{(2)}_{z}(y)  a_{\sigma'} (C_{y}) a_{\sigma} (D_{y}) \Big)
\\
& = \int d z  d \Gamma^{+}_{\sigma \sigma'} \big( A V^{(1)}_{z}B^{T} \big)  d \Gamma^{-}_{\sigma' \sigma} \big( \overline{C} V^{(2)}_{z}D^{*} \big)  \;,
\end{split}
\end{equation*}
where in the last step one proceeds as in \eqref{eq: representation dGamma-V}. The claim then follows by the Leibniz rule for commutators and by \eqref{eq: commutators-dGamma-lemma}.
\end{proof}

We are now ready to control the growth of local fluctuations.
\begin{proof}[Proof of Proposition \ref{thm:bound-localized-fluctuation}.]
Throughout this proof, we will use the shorthand notation $C_{t}\equiv C \exp(c|t|)$ for some $C,c >0$ independent of $\Lambda$ which can possibly vary from line to line. We will not trace the dependence of the constants on $n$, since we can fix some $n \geq 4$ large enough, the bound for $\bar{n} > n_{0}$ following by monotonicity.
The proof is based on a Gronwall argument. By Proposition \ref{prop:fluctuation-generator} we have, for $\sigma = l,r$, 
\begin{equation}\label{eq: derivative XiWXi}
\begin{split}
i \veps \partial_{t} \langle \Xi_{t},  & d \Gamma_{\sigma}(\W^{(n)}_{z}) \Xi_{t}  \rangle
\\
 & = \big \langle \Xi_{t}, \big [ d \Gamma _{\sigma} (\W^{(n)}_{z}) ,d \Gamma_{l}(\sqrt{1 - \veps^{2}\Delta}) - d \Gamma_{r}(\sqrt{1-\veps^{2}\Delta}) + \mathcal{C}(t) + \mathcal{Q}(t)\big]  \Xi_{t}  \big\rangle
\end{split}
\end{equation}
where we used that $\big[  d \Gamma _{\sigma} (\W^{(n)}_{z}),d \Gamma_{\sigma '}(V \ast \varrho_{t})\big] = 0$ and where $\mathcal{C}(t)$ and $\mathcal{Q}(t)$ have been defined in \eqref{eq:generator-C} and \eqref{eq:generator-Q} respectively.
We control each term in the above commutator separately.

\paragraph{Control of $\big[d \Gamma_{\sigma}(\W^{(n)}_{z}),d \Gamma_{l}(\sqrt{1-\veps^{2}\Delta}) \big]$.} 
With the bound (\ref{eq:comm-kin}), we find 
\[ i \big[d \Gamma_{\sigma}(\W^{(n)}_{z}),d \Gamma_{l}(\sqrt{1-\veps^{2}\Delta}) \big] = \delta_{\sigma,l} \, d\Gamma_l \big( i \big[ \W^{(n)}_{z}, \sqrt{1-\veps^{2}\Delta} \big] \big) \leq C \veps \,  \delta_{\sigma,l} \, d\Gamma_l (\W^{(n)}_z)\;. \]
Therefore
\begin{equation}\label{eq: commutator-kinetic}
\Big|\langle \Xi_{t}, \Big [ d \Gamma _{\sigma} (\W^{(n)}_{z}) ,d \Gamma_{l}(\sqrt{1-\veps^{2}\Delta}) \Big]  \Xi_{t}  \rangle \Big| \leq C \veps \, \delta_{\sigma,l}  \langle \Xi_{t}, d \Gamma _{l} (\W^{(n)}_{z}) \Xi_{t}  \rangle \;.
\end{equation}
To control $\big[d \Gamma_{\sigma}(\W^{(n)}_{z}),d \Gamma_{r}(\sqrt{1-\veps^{2}\Delta}) \big]$, we proceed analogously.

\paragraph{Control of $\big[d \Gamma _{\sigma}( \W^{(n)}_{z}), \mathcal{C}(t) \big]$.} We discuss separately the various contributions to $\mathcal{C}(t)$.

\medskip

\noindent$(\mathrm{i})$ Consider the term:
\begin{equation*}
\begin{split}
\mathrm{I} :=\Big[ d \Gamma _{\sigma}(\W^{(n)}_{z})& , \int d x d y  \, V(x-y) a^{\ast} _{l}(u_{t;x}) a^{\ast} _{l}(u_{t;y}) a_{l} (u_{t;y}) a_{l} (u_{t;x})\Big] \;.
\end{split}
\end{equation*}
From Eq. \eqref{eq: commutator-convolution-I}, we have 
\begin{equation}\label{eq: I-convolution}
\mathrm{I} = 2 \delta_{\sigma, l} \int d z' \int d x \, V^{(1)}_{z'}(x) a^{\ast} _{l}(u_{t;x})
d \Gamma _{l}\big( \big[\W^{(n)}_{z}, u_{t} V_{z'}^{(2)} u_{t} \big]\big) a_{l}(u_{t;x})\;.
\end{equation}
By the Cauchy--Schwarz inequality and by Lemma \ref{lemma: bound-Fock-operator}, we find 
\begin{equation}\label{eq: I-CS-tr}
\begin{split}
\big|\big \langle \Xi_{t}, \mathrm{I}\,&\Xi_{t} \big \rangle \big| 
\\
& \leq 2 \int d z' \int d x \big|V^{(1)}_{z'}(x) \big| \,\big\|  a _{l}(u_{t;x}) \Xi_{t} \big\|  \;
\big\|d \Gamma _{l}\big( \big[\W^{(n)}_{z}, u_{t} V_{z'}^{(2)} u_{t} \big]\big) a_{l}(u_{t;x}) \Xi_{t}\big\|
\\
& \leq 2 \int d z'  \bigg( \int d x  \big|V^{(1)}_{z'}(x) \big| \,\big\|  a _{l}(u_{t;x}) \Xi_{t} \big\|^2 \bigg)   \big\|\big[\W^{(n)}_{z}, u_{t} V_{z'}^{(2)} u_{t} \big] \big\|_{\tr}  
\\
&\leq  2\int d z' \,  \big \langle \Xi_{t}, d \Gamma_{l}\big( u_{t} |V^{(1)}_{z'}| u_{t}\big)\Xi_{t}\big\rangle  \,  \big\|\big[\W^{(n)}_{z}, u_{t} V_{z'}^{(2)} u_{t} \big] \big\|_{\tr} \;.
%\\
%&=   2\int _{B_{\tilde{R}}(z)} d z' \,  \big \langle \Xi_{t}, d \Gamma_{l}\big( u_{t} |V^{(1)}_{z'} | u_{t}\big)\Xi_{t}\big\rangle  \,  \big\|\big[\W^{(n)}_{z}, u_{t} %V_{z'}^{(2)} u_{t} \big] \big\|_{\tr}
%\\
%& \qquad \qquad \qquad +  2\int _{(B_{\tilde{R}}(z))^{c}} d z'\,  \big \langle \Xi_{t}, d \Gamma_{l}\big( u_{t}V^{(1)}_{z'} u_{t}\big)\Xi_{t}\big\rangle  \,  
%\big\|\big[\W^{(n)}_{z}, u_{t} V_{z'}^{(2)} u_{t} \big] \big\|_{\tr} \;,
\end{split}
\end{equation}
From Corollary \ref{cor:commF}, using that $\big[ \W^{(n)}_{z}, V_{z'}^{(2)}\big] =0$, we have: 
\begin{equation}\label{eq: WuVu}
\begin{split}
\big\|\big[\W^{(n)}_{z}, u_{t} V_{z'}^{(2)} u_{t} \big] \big\|_{\tr}  &\leq 2  \big\|V_{z'}^{(2)}\big[\W^{(n)}_{z}, u_{t} \big] \big\|_{\tr} \leq \frac{C_{t} \veps^{-2}}{1+ |z-z'|^8}
\end{split}
\end{equation}
for all $z' \in \bR^3$ (because we assumed that $z \in \Lambda$), for all $n \in \bN$ large enough (in Corollary \ref{cor:commF}, we choose $F = \W^{(n)}$ and we use the fact that $V$ decays faster than any power to estimate $V^{(2)}$ by a factor $\W^{(n)}$, for $n \in \bN$ sufficiently large). To bound $\langle \Xi_{t}, d \Gamma_{l}\big( u_{t} |V^{(1)}_{z'}| u_{t}\big)\Xi_{t}\big\rangle$, we choose $R > 0$ so that $B_R (z) \subset \Lambda$, and we consider separately $z' \in B_R (z)$ and $z' \in B_R^c (z)$. For $z' \in B_R (z)$, we proceed as in (\ref{eq:uWu2}) to estimate 
\begin{equation}\label{eq:PV-bound}
\begin{split}
\big \langle \Xi_{t}, d \Gamma_{l}\big( u_{t} |V^{(1)}_{z'}| u_{t}\big)\Xi_{t}\big\rangle \leq 
 \big \langle \Xi_{t}, d\Gamma_{l}\big(\W^{(n)}_{z'}\big)\Xi_{t}\big\rangle + C_{t} \, \veps^{-2} \;.
\end{split}
\end{equation}
Here, we used the fact that $V^{(1)}$ is smooth in momentum space (which follows from the definition \eqref{def_Vj}), to conclude that it decays faster than any power in position space.
For $z' \in B_R^c (z)$, on the other hand, we proceed as in (\ref{eq:uWu3})-(\ref{eq:uWu3b}) and we apply Lemma \ref{prop: a priori bound} and Proposition \ref{prop: semiclassical_prop-main} to obtain  
\begin{equation}\label{eq:PV-bound-2}
\begin{split}
\big \langle \Xi_{t}, d \Gamma_{l}\big( u_{t}V^{(1)}_{z'} u_{t}\big)\Xi_{t}\big\rangle 
&\leq \big \langle \Xi_{t}, d\Gamma_{l}\big(\W^{(n)}_{z'}\big)\Xi_{t}\big\rangle + \tr  \,\W^{(n)}_{z'} \omega_{t} \leq C_{t} \veps^{-3} \;.
\end{split}
\end{equation}
Combining these estimates, and using
\[ \int_{B_R^c (z)} \frac{1}{1+|z-z'|^8} \, dz' \leq C (1+R)^{-5} , \]   
we arrive at 
\begin{equation}\label{eq:termi-fin}
\begin{split}
\big|\big \langle \Xi_{t}, \mathrm{I}\,\Xi_{t} \big \rangle \big| \leq 
 C_t \big(\veps^{-4} & + \veps^{-5} (1+R)^{-5}  \big)
+  C_t \veps^{-2}  \sup_{z' \in B_{R} (z)} \,\big \langle \Xi_{t}, d\Gamma_{l}\big(\W^{(n)}_{z'}\big)\Xi_{t} \big\rangle \;.
\end{split}
\end{equation}

\medskip

\noindent
$(\mathrm{ii})$ With \eqref{eq: decomposition-convolution-I} we write the next term as
\begin{equation}\label{eq: II-commut}
\begin{split}
\mathrm{II} &=  \Big[ d \Gamma _{\sigma}(\W^{(n)}_{z}) , \int d x d y V(x-y) a^{*}_{l}(u_{t;x}) a^{*}_{r}(\overline{v_{t;x}})a_{r}(\overline{v_{t;y}})a_{l}(u_{t;y})\Big] 
\\
& =\int d z' d \Gamma_{l r}^{+}\big( \delta_{\sigma, l} \W^{(n)}_{z} u_{t} V^{(1)}_{z'} v_{t} + \delta_{\sigma, r} u_{t} V^{(1)}_{z'} v_{t}  \W^{(n)}_{z}\big)  d \Gamma_{r l}^{-}\big(v_{t} V^{(2)}_{z'} u_{t} \big) 
\\
& \qquad - \int d z'  d \Gamma_{l r}^{+}\big(u_{t} V^{(1)}_{z'} v_{t} \big) d \Gamma_{r l}^{-}\big( \delta_{\sigma, r}\W^{(n)}_{z} v_{t} V^{(2)}_{z'} u_{t} + \delta_{\sigma, l} v_{t} V^{(2)}_{z'} u_{t} \W^{(n)}_{z}\big)\;.
\end{split}
\end{equation}
Let us consider the first term on the r.h.s. of (\ref{eq: II-commut}), the other can be treated analogously. As we did in (i), we decompose the $z'$-integral in two regions. First, we consider the contribution from $z' \in B_R (z)$. With Lemma \ref{lemma: bound-Fock-operator}, we find 
\begin{equation}\label{eq:termii0}
\begin{split}
\Big| \big\langle \Xi_{t}, \int_{B_{R}(z)} d z' d \Gamma_{l r}^{+} &\big( \W^{(n)}_{z} u_{t} V^{(1)}_{z'} v_{t} \big)  d \Gamma_{r l}^{-}\big(v_{t} V^{(2)}_{z'} u_{t}\big) \Xi_{t}\big \rangle\Big|  \\ &
\leq \int_{B_{R} (z)} d z' \big\| \W^{(n)}_{z} u_{t} V^{(1)}_{z'} v_{t}\big\|_{\tr} \, \big\| v_{t} V^{(2)}_{z'} u_{t} \big\|_{\tr} \;.
\end{split}
\end{equation}
We can estimate 
\begin{equation}\label{eq: II-WuVv}
\begin{split}
\big\| \W^{(n)}_{z} u_{t} V^{(1)}_{z'} v_{t}\big\|_{\tr} & \leq \big\| \W^{(n)}_{z}  V^{(1)}_{z'} u_{t}  v_{t}\big\|_{\tr} +
\big\| \W^{(n)}_{z}\big[ u_{t} ,V^{(1)}_{z'} \big] v_{t} \big\|_{\tr}
\\
& \leq C (1+|z-z'|^8)^{-1}  \big\| \W^{(n/2)}_{z} u_{t}  v_{t}\big\|_{\tr} +
\big\| \W^{(n)}_{z} \big[ u_{t} ,V^{(1)}_{z'} \big] \big\|_{\tr}
\end{split}
\end{equation}
if $n \in \bN$ is large enough (because then $\big\|  \W^{(n)}_{z}  V^{(1)}_{z'}(\W^{(n/2)}_{z})^{-1} \|_{\op} \leq C (1+|z-z'|^8)^{-1}$, using the decay of $V^{(1)}$). With Proposition \ref{prop: semiclassical_prop-main} and Corollary \ref{cor:commF} we obtain 
\begin{equation}\label{eq:II-WuVv2} 
\big\| \W^{(n/2)}_{z} u_{t} v_{t}\big\|_{\tr} \leq C_{t} \veps^{-2} \;, \qquad 
\big\| \W^{(n)}_{z} \big[ u_{t} ,V^{(1)}_{z'} \big] \big\|_{\tr} \leq \frac{C_{t} \veps^{-2}}{1+|z-z'|^8}  \;,
\end{equation}
if $n \in \bN$ is large enough. We conclude that
\begin{equation}\label{eq:termii1}
\big\| \W^{(n)}_{z} u_{t} V^{(1)}_{z'} v_{t}\big\|_{\tr} \leq \frac{C_t \veps^{-2}}{1+|z-z'|^8}  
\end{equation}
if $n \in \bN$ is large enough. Since we assumed $z \in \Lambda$, the last estimate holds for all $z' \in \bR^3$. For $z' \in B_R (z)$, we can also apply Proposition \ref{prop: semiclassical_prop-main} and Corollary \ref{cor:commF} to bound the second norm on the r.h.s.~of (\ref{eq:termii0}) by
\begin{equation}\label{eq: II vVu}
\begin{split}
\big\| v_{t} V^{(2)}_{z'} u_{t} \big\|_{\tr} 
& \leq 
\big\| v_{t} u_{t}  V^{(2)}_{z'} \big\|_{\tr}
+
\big\| v_{t} \big[V^{(2)}_{z'}, u_{t}\big] \big\|_{\tr}
\\
& \leq 
\big\| v_{t} u_{t}  \W^{(n)}_{z'} \big\|_{\tr}
+
\big\| \big[V^{(2)}_{z'}, u_{t}\big] \big\|_{\tr}
\leq C_{t} \veps^{-2} \; .
\end{split}
\end{equation}
Combining (\ref{eq:termii1}) with (\ref{eq: II vVu}), we obtain 
\begin{equation}\label{eq:termii3}
\int_{B_{R}(z)} d z' \big\| \W^{(n)}_{z} u_{t} V^{(1)}_{z'} v_{t}\big\|_{\tr} \, \big\| v_{t} V^{(2)}_{z'} u_{t} \big\|_{\tr}
\leq C_{t} \veps^{-4}
\end{equation}
if $n \in \bN$ is large enough. For $z' \in B^c_R (z)$ we proceed differently, expanding 
\begin{equation*}
d \Gamma_{r l}^{-}\big(v_{t} V^{(2)}_{z'} u_{t}\big) = \int d y \, V^{(2)}_{z'}(y) \, a_{r}(\overline{v_{t;y}}) a_{l}(u_{t;y}) \; .
\end{equation*}
We find 
\begin{equation}\label{eq:ii-A+B}
\begin{split}
\int_{B^c_R (z)} & d z' d \Gamma_{l r}^{+}\big( \W^{(n)}_{z} u_{t} V^{(1)}_{z'} v_{t} \big)  d \Gamma_{r l}^{-}\big(v_{t} V^{(2)}_{z'} u_{t}\big) 
\\
& = \int_{B^c_R (z)} d z' \int dy  \, V^{(2)}_{z'}(y) \, a_{r}(\overline{v_{t;y}}) d \Gamma_{l r}^{+}\big( \W^{(n)}_{z} u_{t} V^{(1)}_{z'} v_{t} \big)   a_{l}(u_{t;y})
\\
& \qquad - \int_{B^c_R (z)} d z' \, d\Gamma_{l} \big( \W^{(n)}_{z} u_{t} V^{(1)}_{z'} \omega_{t} V^{(2)}_{z'}u_{t}\big)  =: A + B\;.
\end{split}
\end{equation}
By Lemma \ref{lemma: bound-Fock-operator} and by \eqref{eq:termii1}, we obtain
\begin{equation}\label{eq:termii4}
\begin{split}
\Big| \big\langle \Xi_{t}, B \, \Xi_{t} \big \rangle \Big|
& 
\leq 
\int  d z' \Big| \big\langle \Xi_{t}, d \Gamma_{l} \big( \W^{(n)}_{z} u_{t} V^{(1)}_{z'} \omega_{t} V^{(2)}_{z'}u_{t}\big) \Xi_{t} \big \rangle \Big|
\\
& \leq \int_{B^c_{R} (z)}  d z' \big\|\W^{(n)}_{z} u_{t} V^{(1)}_{z'} \omega_{t} V^{(2)}_{z'}u_{t}  \big\|_{\tr}
\\
& \leq \int_{B^c_{R} (z)} d z' \big\|\W^{(n)}_{z} u_{t} V^{(1)}_{z'} v_{t}\big\|_{\tr} \leq C_{t} \veps^{-2}\;.
\end{split}
\end{equation}
As for the term $A$, we find, with Lemma \ref{lemma: bound-Fock-operator}, 
\begin{equation}\label{eq:termii5}
\begin{split}
\Big| \big\langle \Xi_{t}, A \Xi_{t} \big \rangle \Big|
&\leq \int_{B^c_R (z)}  d z' \int dy  \, \big| V^{(2)}_{z'}(y)\big| \,\Big| \big \langle a^{*}_{r} (\overline{v_{t;y}}) \Xi_{t},  d \Gamma_{l r}^{+}\big( \W^{(n)}_{z} u_{t} V^{(1)}_{z'} v_{t} \big)   a_{l}(u_{t;y}) \Xi_{t} \big \rangle \Big|
\\
& \leq \int_{B^c_R (z)}  d z' \big\|\W^{(n)}_{z} u_{t} V^{(1)}_{z'} v_{t} \big\|_{\tr} \int dy  \, \big| V^{(2)}_{z'}(y)\big| \,\| v_{t;y}\|_{\mathfrak{h}}
 \, \big\|a_{l}(u_{t;y}) \Xi_{t} \big\|  \\ &\leq C_t \veps^{-5} (1+R)^{-5}\;. 
\end{split}
\end{equation}
Here, we used again (\ref{eq:termii1}) and we estimated 
\begin{equation}\label{eq: A-1}
\begin{split}
\int dy & \, \big| V^{(2)}_{z'}(y)\big| \,\| v_{t;y}\|_{\mathfrak{h}}
 \, \big\|a_{l}(u_{t;y}) \Xi_{t} \big\|
 \\
 & \leq \Big(\int dy  \, \big| V^{(2)}_{z'}(y)\big| \,\| v_{t;y}\|^{2}_{\mathfrak{h}}
 \, \Big)^{1/2} \, 
 \Big(\int dy  \, \big| V^{(2)}_{z'}(y)\big| 
 \, \big\|a_{l}(u_{t;y}) \Xi_{t} \big\|^{2} \Big)^{1/2}
 \\
 & = \Big(\tr \big| V^{(2)}_{z'}\big| \omega_{t} \Big)^{1/2} \, \big \langle \Xi_{t}, d \Gamma_{l}\big(u_{t}\big| V^{(2)}_{z'}\big| u_{t} \big) \Xi_{t}\big \rangle^{1/2} \leq C_t \veps^{-3} 
\end{split}
\end{equation}
by Prop. \ref{prop: semiclassical_prop-main} and proceeding as in (\ref{eq:PV-bound-2}). Combining (\ref{eq:termii3}), (\ref{eq:termii4}) and (\ref{eq:termii5}), we arrive at 
\begin{equation*}
\big| \langle \Xi_{t},  \mathrm{II} \,\Xi_{t} \rangle \big| \leq  C_t \big(\veps^{-4}  + \veps^{-5} (1+ R) ^{-5} \big)  \;.
\end{equation*}

\medskip

\noindent
$(\mathrm{iii})$ By \eqref{eq: commutator-convolution-I} we write
\begin{equation*}
\begin{split}
\mathrm{III} & :=  \Big[ d \Gamma _{\sigma}(\W_{z}) , \int d x d y V(x-y) a^{\ast}_{l} (u_{t;x}) a^{\ast}_{r}(\overline{v_{t;y}})  a_{r} (\overline{v_{t;y}}) a _{l}(u_{t;x})\Big]
\\
& = \delta_{\sigma, l} \int d z' d y V^{(1)}_{z'}(y) a^{\ast} _{r}(\overline{v_{t;y}}) d \Gamma_{l}\big( \big[\W^{(n)}_{z}, u_{t} V^{(2)}_{z'} u_{t}\big] \big) a_{r} (\overline{v_{t;y}}) 
\\
&\qquad + \delta_{\sigma, r} \int d z' \int d x V^{(1)}_{z'}(x) a^{*}_{l}(u_{t;x}) d \Gamma_{r}\big( \big[\W^{(n)}_{z}, \overline{v_{t}} V^{(2)}_{z'} \overline{v_{t}} \big] \big) a_{l}(u_{t;x})\;.
\end{split}
\end{equation*}
Both contributions have the same structure as the term $\mathrm{I}$, just with two factors $u_t$ replaced by factors of $v_t$ and, in the second contribution, with $d\Gamma_l$ replaced by $d\Gamma_r$. Proceeding very similarly as we did in (i), we conclude that 
\begin{equation*}
\big| \langle \Xi_{t},  \mathrm{III} \,\Xi_{t} \rangle \big| \leq C_t \big(\veps^{-4}  + \veps^{-5} (1+ R)^{-5} \big)  
+ C_t  \veps^{-2} \sup_{z'\in B_R (z)} \big \langle \Xi_{t}, d \Gamma \big(\W^{(n)}_{z'}\big)\Xi_{t}\big\rangle
\end{equation*}
where we introduced the shorthand notation $d\Gamma (O) = d\Gamma_l (O) + d\Gamma_r (O)$, for second quantized operators on $\mathcal{F} (\frak{h} \oplus \frak{h})$.

\medskip

\noindent 
$(\mathrm{iv})$ With \eqref{eq: commutator-convolution-I} we write
\begin{equation*}
\begin{split}
\mathrm{IV} & :=  \Big[ d \Gamma _{\sigma}(\W^{(n)}_{z}) , \int d x d y V(x-y) a^{\ast}_{r} (\overline{v_{t;y}}) a^{\ast}_{r} (\overline{v_{t;x}}) a_{r} (\overline{v_{t;x}}) a_{r} (\overline{v_{t;y}}) \Big]
\\
& = 2 \delta_{\sigma,r}\int d x d y V(x-y) a^{\ast} _{r}(\overline{v_{t;x}}) \Big[ d \Gamma _{r}(\W^{(n)}_{z}),a^{\ast} _{r}(\overline{v_{t;y}}) a_{r} (\overline{v_{t;y}}) \Big] a_{r} (\overline{v_{t;x}}) \;.
\end{split}
\end{equation*}
Also this contribution is analogous to the term $\text{I}$, with now four factors $u_t$ replaced by factors of $v_t$. Proceeding as in (i), we find  
\begin{equation*}
\big| \langle \Xi_{t},  \mathrm{IV} \, \Xi_{t} \rangle \big| \leq C_t \big(\veps^{-4}  + \veps^{-5} (1+ R)^{-5} \big) 
+  C_t \veps^{-2} \sup_{z'\in B_R (z)} \,\big \langle \Xi_{t}, d\Gamma_{r}\big(\W^{(n)}_{z'}\big)\Xi_{t} \big\rangle  \;.
\end{equation*}

\medskip

\noindent $(\mathrm{v}) - (\mathrm{viii})$ The contribution of the next four terms appearing in the expression (\ref{eq:generator-C}) for the operator $\mathcal{C} (t)$ is completely analogous to the contribution of the first four terms and can be handled in the same way. We are left with the contribution of the quadratic terms on the last line of (\ref{eq:generator-C}). 

\medskip

\noindent 
$(\mathrm{ix})$ With Lemma \ref{lemma: commutators-dGamma} we write 
\begin{equation*}
\begin{split}
\mathrm{IX} & :=  \Big[ d \Gamma _{\sigma}(\W^{(n)}_{z}) , \int d x d y V(x-y) \big(a^{\ast}_{r} (\overline{v_{t;y}}) a_{r} (\overline{v_{t;x}}) -a^{\ast}_{l} (u_{t;x}) a_{l} (u_{t;y}) \big)\omega(x ; y) \Big] \;.
\\
& = \int d z' \Big( \delta_{\sigma,r} d \Gamma_{r}\big(\big[ \W^{(n)}_{z}, \overline{v_{t}} V^{(2)}_{z'} \omega_{t} V^{(1)}_{z'} \overline{v_{t}}\big] \big) 
-\delta_{\sigma,l} d \Gamma_{l}\big(\big[ \W^{(n)}_{z}, u_{t} V^{(2)}_{z'} \omega_{t} V^{(1)}_{z'} u_{t}\big] \big) \Big)
\\
& := \mathrm{IX}_{1} + \mathrm{IX}_{2} \;.
\end{split}
\end{equation*}
By Lemma \ref{lemma: bound-Fock-operator}, we have
\begin{equation*}
\big| \big \langle \Xi_{t}, \mathrm{IX}_{1} \, \Xi_{t}\big \rangle \big| \leq \int d z' \big\|\big[ \W^{(n)}_{z}, \overline{v_{t}} V^{(2)}_{z'} \omega_{t} V^{(1)}_{z'} \overline{v_{t}}\big]\big\|_{\tr} \;.
\end{equation*}
With $\|\omega_{t}\|_{\op} \leq 1$, $\| u_{t} \|_{\op} \leq 1$, $\| v_{t}\|_{\op}\leq 1$ and since $V^{(1)}, V^{(2)}$ are bounded functions, we obtain 
\begin{equation*}
\begin{split}
\big\|\big[ \W^{(n)}_{z},& \overline{v_{t}} V^{(2)}_{z'} \omega_{t} V^{(1)}_{z'} \overline{v_{t}}\big]\big\|_{\tr}
\\
& \leq C \big\|\big[ \W^{(n)}_{z}, \overline{v_{t}}\big] V^{(2)}_{z'} \big\|_{\tr} 
+C \big\|\big[ \W^{(n)}_{z}, \omega_{t}\big] V^{(1)}_{z'}\big\|_{\tr}
+ C \big\| V^{(1)}_{z'} \big[ \W^{(n)}_{z}, \overline{v_{t}}\big]\big\|_{\tr} \;.
\end{split}
\end{equation*}
From Corollary \ref{cor:commF} and from the invariance of the trace norm under complex conjugation, we find 
\begin{equation*}
\big\|\big[ \W^{(n)}_{z}, \overline{v_{t}} V^{(2)}_{z'} \omega_{t} V^{(1)}_{z'} \overline{v_{t}}\big]\big\|_{\tr} \leq C_t \veps^{-2} \,  (1+|z-z'|^8)^{-1}
\end{equation*}
if $n \in \bN$ is large enough. A similar analysis holds for the term $\mathrm{IX}_{2}$. We conclude that
\begin{equation*}
\big| \big \langle \Xi_{t}, \mathrm{IX} \, \Xi_{t}\big \rangle \big| \leq C_t \veps^{-2} \;.
\end{equation*}

\medskip

Combining the bounds in $(\mathrm{i}) - (\mathrm{ix})$, we arrive at 
\begin{equation}\label{eq: control WCt}
\begin{split} 
\big| \langle \Xi_{t},  \big [ &d \Gamma _{\sigma} (\W^{(n)}_{z}) ,  \mathcal{C}(t) \big] \Xi_{t} \rangle \big| \\ &\leq C_t \big(\veps^{-1}  + \veps^{-2} (1+ R)^{-5} \big) +  C_t \veps \sup_{z'\in B_R (z)} \big \langle \Xi_{t}, d \Gamma \big(\W^{(n)}_{z'}\big)\Xi_{t}\big\rangle\; ,
\end{split} 
\end{equation}
for $\sigma = l,r$, for $n \in \bN$ large enough and for all $t \in \bR$. 

\paragraph{Control of $\big[d \Gamma_{\sigma}( \W^{(n)}_{z}), \mathcal{Q}(t) \big]$.} We discuss separately the various terms appearing in $\mathcal{Q}(t)$:

\bigskip

\noindent
$(\mathrm{i})$ From \eqref{eq: decomposition-convolution-III}, we have 
\begin{equation}\label{eq: Itilde}
\begin{split}
\widetilde{\mathrm{I}} & =  \Big[ d \Gamma _{\sigma}(\W^{(n)}_{z}) , \int d x d y V(x-y) a^{*}_{l}(u_{t;x})a^{*}_{l}(u_{t;y})a^{*}_{r}(\overline{v_{t;y}})a^{*}_{r}(\overline{v_{t;x}}) \Big] 
\\
& =
\int d z' d \Gamma_{l r}^{+}\big( \delta _{\sigma, l}\W^{(n)}_{z} u_{t} V^{(1)}_{z'} v_{t} + \delta_{\sigma,r} u_{t} V^{(1)}_{z'} v_{t} \W^{(n)}_{z'} \big) d \Gamma^{+}_{lr} \big( u_{t} V^{(2)}_{z'}v_{t}\big)
\\
&
\qquad +
\int d z' d \Gamma_{l r}^{+}\big(  u_{t} V^{(1)}_{z'} v_{t} \big) d \Gamma^{+}_{lr} \big(\delta _{\sigma, l}\W^{(n)}_{z}  u_{t} V^{(2)}_{z'}v_{t}+\delta_{\rho,r}  u_{t} V^{(2)}_{z'}v_{t}\W^{(n)}_{z}\big) \;.
\end{split}
\end{equation}
We consider the first term, the second can be handled similarly. We proceed similarly as we did for the term $\text{II}$. With Lemma \ref{lemma: bound-Fock-operator}, the contribution associated with $z' \in B_R (z)$ is bounded by 
\begin{equation*}
\begin{split}
\Big| \big \langle \Xi_{t},&  \int_{B_R (z)} d z' d \Gamma_{l r}^{+}\big( \delta _{\sigma, l}\W^{(n)}_{z} u_{t} V^{(1)}_{z'} v_{t} \big) d \Gamma^{+}_{lr} \big( u_{t} V^{(2)}_{z'}v_{t}\big)\Xi_{t}\big\rangle \Big|
\\
& \leq \int_{B_R (z)} d z' \big\|\W^{(n)}_{z} u_{t} V^{(1)}_{z'} v_{t} \big\|_{\tr} \, \big\| u_{t} V^{(2)}_{z'} v_{t} \big\|_{\tr} \leq C_t \veps^{-4}  \;,
\end{split}
\end{equation*}
where in the last step we used \eqref{eq:termii1} and \eqref{eq: II vVu}. For $z' \in B^c_R (z)$, on the other hand, we expand $d\Gamma_{lr}^+ (u_t V_z^{(2)} v_t)$ and we estimate, using again Lemma \ref{lemma: bound-Fock-operator}, \eqref{eq:termii1} and \eqref{eq: A-1} (in contrast to (\ref{eq:ii-A+B}), there is here no $\text{B}$ term, because creation operators anticommute):  
\begin{equation}\label{eq: partialItilde}
\begin{split}
\Big| & \big \langle \Xi_{t},
\int_{B_R^c (z)} d z' d \Gamma_{l r}^{+}\big( \delta _{\sigma, l}\W^{(n)}_{z} u_{t} V^{(1)}_{z'} v_{t} \big) d \Gamma^{+}_{lr} \big( u_{t} V^{(2)}_{z'}v_{t}\big) \Xi_{t} \big\rangle \Big|
\\
%& \leq 
%\int_{B_R^c (z)}  d z' \int dy \big|V^{(2)}_{z'}(y) \big| \big \langle a_{l}(u_{t;y}) \Xi_{t}, d \Gamma_{l r}^{+}\big( \delta _{\sigma, l}\W^{(n)}_{z} u_{t} %V^{(1)}_{z'} v_{t} \big) a^{*}_{r}(\overline{v_{t;y}}) \Xi_{t} \rangle \big|
%\\
& \leq 
\int_{B_R^c (z)}  d z' \big\|\W^{(n)}_{z} u_{t} V^{(1)}_{z'} v_{t} \big\|_{\tr} \int dy \big|V^{(2)}_{z'}(y)\big| \, \| \bar{v}_{t,y} \|_\frak{h} \big \| a_{l}(u_{t;y}) \Xi_{t}\big\|  \leq C_t \veps^{-5} (1+R)^{-5}\;.
\end{split}
\end{equation}
Thus, 
\begin{equation*}
\Big| \big \langle \Xi_{t}, \widetilde{\mathrm{I}} \,\Xi_{t} \big\rangle \Big| \leq C_t \big(\veps^{-4} + \veps^{-5}(1+R)^{-5} \big) \;.
\end{equation*}
for $n \in \bN$ large enough.

\medskip

\noindent
$(\mathrm{ii})$ With \eqref{eq: decomposition-convolution-II}, we find 
\begin{equation*}
\begin{split}
\widetilde{\mathrm{II}} & :=  \Big[ d \Gamma _{\sigma}(\W^{(n)}_{z}) , \int d x d y V(x-y) a^{*}_{l}(u_{t;x})a^{*}_{l}(u_{t;y})a^{*}_{r}(\overline{v_{t;x}})a_{l}(u_{t;y}) \Big] 
\\
& =
 - \int d z' \int d y V^{(2)}_{z'}(y) a^{*}_{l}(u_{t;y})d \Gamma_{l r}^{+}\big( \delta _{\sigma, l}\W^{(n)}_{z} u_{t} V^{(1)}_{z'} v_{t} + \delta_{\sigma,r} u_{t} V^{(1)}_{z'} v_{t} \W^{(n)}_{z'} \big) 
 a_{l}(u_{t;y})
\\
&
\qquad - \delta_{\sigma,l}
\int d z' d \Gamma_{l r}^{+}\big(  u_{t} V^{(1)}_{z'} v_{t} \big) d \Gamma_{l} \big(\big[\W^{(n)}_{z},  u_{t} V^{(2)}_{z'}u_{t}\big]\big) 
\\
& =: \widetilde{\mathrm{II}}_{1} + \widetilde{\mathrm{II}}_{2} \;.
\end{split}
\end{equation*}
The contribution $\widetilde{\mathrm{II}}_{1}$ is similar to the term $\mathrm{I}$; the main difference is that the commutator 
$[\W_z^{(n)}, u_t V_{z'}^{(2)} u_t]$ is now replaced by the product $W^{(n)}_{z} u_{t} V^{(1)}_{z'} v_{t}$ or by the product $u_{t} V^{(1)}_{z'} v_{t} \W^{(n)}_{z'} $, whose trace norm can however be controlled with (\ref{eq:termii1}). Similarly to \eqref{eq:termi-fin}, we obtain
\begin{equation*}
\big| \big \langle  \Xi_{t},  \widetilde{\mathrm{II}}_{1} \, \Xi_{t} \big \rangle \big|  \leq  C_t \big(\veps^{-4}  + \veps^{-5} (1+ R)^{-5} \big) + C_t  \veps^{-2}\sup_{z'\in B_R (z)} \,\big \langle \Xi_{t}, d\Gamma_{l}\big(\W^{(n)}_{z'}\big)\Xi_{t} \big\rangle \;.
\end{equation*}
As for $\widetilde{\mathrm{II}}_{2}$, it has the same form as the term $\widetilde{\mathrm{I}}$, with  
$\W_z^{(n)} u_t V^{(1)}_{z'} v_t$ replaced by $[\W_z^{(n)}, u_t V^{(2)}_{z'} u_t ]$, whose trace norm can be estimated with Corollary 
\ref{cor:commF}. We conclude that 
\begin{equation*}
\big| \big \langle  \Xi_{t},  \widetilde{\mathrm{II}}_{2} \, \Xi_{t} \big \rangle \big|  \leq  C_t \big(\veps^{-4}  + \veps^{-5} (1+ R)^{-5} \big)\;. \end{equation*}
Therefore
\begin{equation*}
\big| \big \langle  \Xi_{t},  \widetilde{\mathrm{II}} \, \Xi_{t} \big \rangle \big|  \leq  C_t \big(\veps^{-4}  + \veps^{-5} (1+ R)^{-5} \big) + C_t  \veps^{-2}\sup_{z'\in B_R (z)} \,\big \langle \Xi_{t}, d\Gamma_{l}\big(\W^{(n)}_{z'}\big)\Xi_{t} \big\rangle \;.
\end{equation*}

\noindent
$(\mathrm{iii})$ By \eqref{eq: decomposition-convolution-II}, we write the next term as
\begin{equation*}
\begin{split}
\widetilde{\mathrm{III}} & :=  \Big[ d \Gamma _{\sigma}(\W^{(n)}_{z}) , \int d x d y V(x-y) a^{*}_{l}(u_{t;x})a^{*}_{r}(\overline{v_{t;y}}) a^{*}_{r}(\overline{v_{t;x}}) a_{r}(\overline{v_{t;y}}) \Big] 
\\
& =
 - \int d z' \int d y V^{(2)}_{z'}(y) a^{*}_{r}(\overline{v_{t;y}})d \Gamma_{l r}^{+}\big( \delta _{\sigma, l}\W^{(n)}_{z} u_{t} V^{(1)}_{z'} v_{t} + \delta_{\sigma,r} u_{t} V^{(1)}_{z'} v_{t} \W^{(n)}_{z'} \big) 
 a_{r}(\overline{v_{t;y}})
\\
&
\qquad - \delta_{\sigma,r}
\int d z' d \Gamma_{l r}^{+}\big(  u_{t} V^{(1)}_{z'} v_{t} \big) d \Gamma_{r} \big(\big[\W^{(n)}_{z},  \overline{v_{t}} V^{(2)}_{z'}  \overline{v_{t}}\big]\big)\;.
\end{split}
\end{equation*}
This term has exactly the same structure as the term $\widetilde{\mathrm{II}}$, with $u_t$ and $v_t$ interchanged. We proceed analogously as in $(\mathrm{ii})$, to bound
\begin{equation*}
\big| \big \langle  \Xi_{t},  \widetilde{\mathrm{III}} \, \Xi_{t} \big \rangle \big| \leq C_t \big(\veps^{-4}  + \veps^{-5} (1+ R) ^{-5} \big) 
+  C_t \veps^{-2}  \sup_{z'\in B_R (z)} \,\big \langle \Xi_{t}, d\Gamma_{r}\big(\W^{(n)}_{z'}\big)\Xi_{t} \big\rangle \;.
\end{equation*}

\noindent
$(\mathrm{iv})-(\mathrm{vi})$ The contribution of the next three terms appearing in the expression (\ref{eq:generator-Q}) for the operator $\mathcal{Q} (t)$ is completely analogous to the contribution of the first three terms and can be handled in the same way. We are left with the contribution of the quadratic terms on the last line of (\ref{eq:generator-Q}). 

\medskip

\noindent
$(\mathrm{vii})$ By using Lemma \ref{lemma: commutators-dGamma} we write the commutator with the quadratic term as
\begin{equation*}
\begin{split}
\widetilde{\mathrm{VII}} & := \Big[ d \Gamma _{\sigma}(\W^{(n)}_{z}) , \int d x d y V(x-y)  a^{\ast}_{l} (u_{t;x}) a^{\ast}_{r} (\overline{v_{t;y}}) \omega (x ; y) \Big] 
\\
 & = \int d z' d \Gamma^{+} _{l r} \big( \delta_{\sigma,l}\W^{(n)}_{z} u_{t} V^{(1)}_{z'} \omega_{t} V^{(2)}_{z'} v_{t} + \delta_{\sigma,r} u_{t} V^{(1)}_{z'} \omega_{t} V^{(2)}_{z'} v_{t} \W^{(n)}_{z} \big)
\\
& =: \widetilde{\mathrm{VII}}_{1} + \widetilde{\mathrm{VII}}_{2} \;.
\end{split}
\end{equation*}
With Lemma \ref{lemma: bound-Fock-operator}, we have
\begin{equation*}
\big| \big \langle \Xi_{t}, \widetilde{\mathrm{VII}}_{1} \, \Xi_{t}\big \rangle \big| \leq \int d z' \big\|  \W^{(n)}_{z} u_{t} V^{(1)}_{z'} \omega_{t} V^{(2)}_{z'} v_{t}  \big\|_{\tr} \;.
\end{equation*}
Since $\omega_{t} = v_{t}^{2}$, $\| v_{t}\|_{\op}\leq 1$ and since $V^{(2)}$ is a bounded functions, we find, with (\ref{eq:termii1}),  
\begin{equation*}
\begin{split}
\big\|  \W^{(n)}_{z} u_{t} V^{(1)}_{z'} \omega_{t} V^{(2)}_{z'} v_{t}  \big\|_{\tr}
 & \leq C \big\|  \W^{(n)}_{z} u_{t} V^{(1)}_{z'} v_{t}\big\|_{\tr} \leq \frac{C_t \veps^{-2}}{1+|z-z'|^8}  \, ,
\end{split}
\end{equation*}
if $n \in \bN$ is large enough. This implies that $|\langle \Xi_{t}, \widetilde{\mathrm{VII}}_{1} \, \Xi_{t} \rangle | \leq C_t \veps^{-2}$. The term 
$\widetilde{\mathrm{VII}}_{2}$ can be handled similarly. Hence 
\begin{equation*}
\big| \big \langle \Xi_{t}, \widetilde{\mathrm{VII}} \, \Xi_{t}\big \rangle \big| \leq C_t \veps^{-2} \;.
\end{equation*}
Combining the estimates in $(\mathrm{i})-(\mathrm{vii})$, we conclude that  
\begin{equation}\label{eq: control WQt}
\begin{split} 
\big| \langle \Xi_{t}, \big[ &d\Gamma_\sigma (\W_z^{(n)}) ,  \mathcal{Q}(t) \big]  \Xi_{t} \rangle \big| \\ & \leq C_t \big(\veps^{-1}  + \veps^{-2} (1+ R)^{-5} \big) 
+  C_t \veps  \sup_{z'\in B_{\tilde{R}}(z)} \,\big \langle \Xi_{t}, d\Gamma \big(\W^{(n)}_{z'}\big)\Xi_{t} \big\rangle 
\end{split} 
\end{equation}
for $\sigma = l,r$, for $n \in \bN$ large enough and for all $t \in \bR$. 
\medskip

\noindent{\bf Conclusion.} We plug \eqref{eq: commutator-kinetic}, \eqref{eq: control WCt} and \eqref{eq: control WQt} into \eqref{eq: derivative XiWXi}. We obtain
\begin{equation}\label{eq:last-step-gronw}
\big| \partial_{t} \langle \Xi_{t},  d \Gamma_{\sigma}(\W^{(n)}_{z}) \Xi_{t}  \rangle\big|
\leq  C_t \big(\veps^{-2}  + \veps^{-3} (1+ R) ^{-5} \big)  
+  C_{t}  \sup_{z'\in B_R (z)} \,\big \langle \Xi_{t}, d\Gamma \big(\W^{(n)}_{z'}\big)\Xi_{t} \big\rangle
\end{equation}
for $\sigma = l,r$, for $n \in \bN$ large enough and for all $t \in \bR$. It will be convenient to choose the parameter $R$ as a function of time; we replace $R$ with $R_t$ in (\ref{eq:last-step-gronw}). Rewriting this bound in integral form (and recalling the notation $d\Gamma = d\Gamma_l + d\Gamma_r$), we find 
\begin{equation*}
\begin{split}
\langle \Xi_{t},&  d \Gamma (\W^{(n)}_{z}) \Xi_{t}  \rangle  \leq  \langle \Xi_{0},  d \Gamma (\W^{(n)}_{z}) \Xi_{0} \rangle   \\
&  + C_t \veps^{-2} 
 + C_t \veps^{-3} \int_{0}^{t}\frac{ds}{(1+ R_s)^{5}}  + C_{t} \int_{0}^{t} d s  \sup_{z'\in B_{R_s}(z)} \,\big \langle \Xi_{s}, d\Gamma \big(\W^{(n)}_{z'}\big)\Xi_{s} \big\rangle 
\;.
\end{split}
\end{equation*}
Choosing $R_s = r (t-s)$ for a parameter $r > 0$ to be fixed later on, we obtain 
\begin{equation}\label{eq: partial-Gronwall}
\begin{split}
&\langle \Xi_{t},  d \Gamma (\W^{(n)}_{z}) \Xi_{t}  \rangle  \\
& \leq  \langle \Xi_{0},  d \Gamma (\W^{(n)}_{z}) \Xi_{0} \rangle  +C_t  \Big(\veps^{-2} + \frac{\veps^{-3}}{r}\Big) 
 + C_{t} \int_{0}^{t} d s \sup_{z'\in B_{r (t-s)}(z)} \,\big \langle \Xi_{s}, d\Gamma \big(\W^{(n)}_{z'}\big)\Xi_{s} \big\rangle
 \end{split}
\end{equation}
for all $t > 0$. Next, for fixed $T > 0$ and $z_0 \in \Lambda$, we define 
\begin{equation*}
F (t) := \sup_{z \in B_{r (T-t)}(z_0)} \, \big \langle \Xi_{t}, d \Gamma \big(\W^{(n)}_{z}\big)\Xi_{t}\big\rangle\; 
\end{equation*}
for any $t \in [0;T]$. Observing that 
\begin{equation*}
\begin{split}
\sup_{z \in B_{r (T-t)}(z_0)} \langle \Xi_{0},  d \Gamma (\W^{(n)}_{z}) \Xi_{0} \rangle & \leq F (0) ,  
\\
\sup_{z \in B_{r (T-t)}(z_0)}  \sup_{z' \in B_{r (t-s)}(z)} \,\big \langle \Xi_{s}, d\Gamma \big(\W^{(n)}_{z'}\big)\Xi_{s} \big\rangle & \leq F (s), 
\end{split}
\end{equation*}
it follows from \eqref{eq: partial-Gronwall} that 
\begin{equation*}
\begin{split}
F (t) \leq F (0) + C_T \Big(\veps^{-2} + \frac{\veps^{-3}}{r}\Big) 
 + C_T \int_{0}^{t} ds \, F (s) 
\end{split}
\end{equation*}
for all $t \in [0;T]$. From Gronwall's lemma, we conclude that
\[ F (t) \leq e^{C_T} \left(F (0) +  C_T \veps^{-2} + C_T \veps^{-3} /r \right) \, . \]
for all $t \in [0;T]$. Choosing $t=T$ and $r = \veps^{-\delta}$, we arrive at 
\[ \begin{split} \langle \Xi_T , d\Gamma (\W^{(n)}_{z_0}) \Xi_T \rangle &\leq C \exp (c \exp (c|T|)) \, \left[ \sup_{z \in B_{\veps^{-\delta} T} (z_0)} \langle \Xi, d\Gamma (\W^{(n)}_z) \Xi \rangle + \veps^{-3+\delta} \right] \\ &\leq C \exp (c \exp (c|T|)) \veps^{-3+\delta} \end{split} \]
for all $z_0 \in \Lambda$, with $B_{\veps^{-\delta} T} (z_0) \subset \Lambda$, by the assumption (\ref{eq:assT>0}). This concludes the proof of Proposition \ref{thm:bound-localized-fluctuation}. 
\end{proof}

\section{Proof of Theorem \ref{thm:main}}\label{sec:proofmain}

We are now ready to prove our main result, Theorem \ref{thm:main}. Without loss of generality, we consider here $t > 0$; the case $t < 0$ can be handled analogously. We also introduce the notation $C_t = C \exp (ct)$, for constants $C,c > 0$ independent of $\Lambda$, possibly varying from line to line. 

The proof is by approximation of the Fermi projection $\omega_\mu = \chi (H \leq \mu)$ via a positive-temperature state
\begin{equation*}
\omega_{\mu,\beta}  = \frac{1}{1 + e^{\beta (H - \mu)}}\;,
\end{equation*}
with inverse temperature $\beta = O(\veps^{-1})$. We will denote by $\omega_t$ the solution of the nonlinear Hartree equation \ref{eq:hartree-main} with initial datum $\omega_{t=0} = \omega_\mu$ and by $\widetilde{\omega}_t$ the solution with initial data $\widetilde{\omega}_{t=0} = \omega_{\mu,\beta}$. 

Let $z\in \Lambda$, with $\Lambda$ as in Assumption \ref{ass:Weyl}. We start by decomposing 
\begin{equation}\label{eq:A+B}
\begin{split}
\tr\, \mathcal{O}_{z} ( \gamma^{(1)}_{t} - \omega_{t}) &=   \tr\, \mathcal{O}_{z} (\widetilde \omega_{t} - \omega_{t}) + \tr\, \mathcal{O}_{z} ( \gamma^{(1)}_{t} - \widetilde \omega_{t})\equiv \text{A} + \text{B}\;.
\end{split}
\end{equation}
\paragraph{Bound for the term $\text{A}$.} We claim that, for $z \in \Lambda$ such that $B_{\varepsilon^{-\delta}t}(z) \subset \Lambda$ and for $0<\delta \leq 1$:
\begin{equation}\label{eq:Best}
\Big| \tr\, \mathcal{O}_{z} (\widetilde \omega_{t} - \omega_{t}) \Big| \leq C \exp (c \exp (ct)) \, \varepsilon^{-3+\delta}\;,
\end{equation}
for constants $C,c>0$ independent of $\Lambda$. To prove this bound we start by estimating:
\begin{equation}\label{eq:OtoW}
\begin{split}
\Big| \tr\, \mathcal{O}_{z} (\widetilde \omega_{t} - \omega_{t}) \Big| &\leq \Big\| \mathcal{O}_{z} (\widetilde \omega_{t} - \omega_{t}) \Big\|_{\text{tr}} \\ 
&\leq C\Big\| \mathcal{W}_{z}^{(n)} (\widetilde \omega_{t} - \omega_{t}) \Big\|_{\text{tr}}\;,
\end{split}
\end{equation}
where we used that $|\mathcal{O}_{z}| \leq C\mathcal{W}^{(n)}_{z}$, where $n \in \bN$ will be chosen below, large enough. Furthermore,
\begin{equation}\label{eq:I+II}
\begin{split}
\| \mathcal{W}^{(n)}_{z} (\widetilde \omega_{t} - \omega_{t}) \|_\text{tr}  
%&\leq \Big\| \mathcal{W}^{(n)}_{z} (\widetilde \omega_{t} - \omega_{t}) \Big\|_{\text{tr}}
%\\
&\leq  \Big\| \mathcal{W}^{(n)}_{z} \Big( \widetilde{U}(t;0) \omega_{\mu,\beta} \widetilde{U}(t;0)^{*} - U(t;0) \omega_{\mu,\beta} U(t;0)^{*}\Big) \Big\|_{\text{tr}} \\ & \quad + \Big\| \mathcal{W}^{(n)}_{z} \Big( U(t;0) \omega_{\mu,\beta} U(t;0)^{*} - U(t;0) \omega_{\mu} U(t;0)^{*} \Big)\Big\|_{\text{tr}} \\ &= \text{I} + \text{II}\;,
\end{split}
\end{equation}
with $U(t;s)$, $\widetilde{U}(t;s)$ being the unitary dynamics generated by the Hartree Hamiltonians associated with the density $\varrho_{t}(x) = \omega_{t}(x;x)$ and $\widetilde{\varrho}_{t}(x) = \widetilde{\omega}_{t}(x;x)$, respectively. 
%Note that it suffices to control \eqref{eq:I+II} for some $n\geq 4$ large enough, the bound for $\bar{n} > n$ following by %monotonicity. With (\ref{eq:Up-W}) (with $p=0$), we obtain  
%
\begin{equation*}
\begin{split}
\text{II} &=  \Big\| \mathcal{W}^{(n)}_{z} U(t;0) ( \omega_{\mu,\beta}  - \omega_\mu  ) \Big\|_{\text{tr}} \leq C_t \Big\| \mathcal{W}^{(n)}_{z} (  \omega_{\mu,\beta}  - \omega_\mu ) \Big\|_{\text{tr}}\;.
\end{split}
\end{equation*}
To bound the r.h.s.~of the last equation, we observe that  
\begin{equation}\label{eq:oo}
\begin{split}
 \big\| \mathcal{W}^{(n)}_{z} (\omega_{\mu} &- \omega_{\mu,\beta}) \big\|_{\text{tr}} \\ &\leq  \big\| \mathcal{W}^{(n)}_{z} \omega_{\mu} (1 - \omega_{\mu,\beta}) \big\|_{\text{tr}} + \big\| \mathcal{W}^{(n)}_{z} (1-\omega_{\mu}) \omega_{\mu,\beta} \big\|_{\text{tr}} \\
 &\leq 4 \big\| \mathcal{W}^{(n)}_{z} \omega_{\mu,\beta} (1 - \omega_{\mu,\beta}) \big\|_{\text{tr}}\;,
 \end{split}
\end{equation}
where we used that $\chi(x\leq \mu) \leq 2 f_{\mu,\beta}(x)$ and that $1 - \chi(x\leq \mu) \leq 2(1-f_{\mu,\beta}(x))$ with $f_{\mu,\beta}(x)$ the Fermi-Dirac function. Noticing that for any $m\in \mathbb{N}$ there exists $C_m>0$ such that
\begin{equation}\label{eq:oocbd}
\omega_{\mu,\beta} (1 - \omega_{\mu,\beta}) = \frac{e^{\beta (H-\mu)}}{(1 + e^{\beta (H-\mu)})^{2}} \leq \frac{C_{m}}{(\beta (H-\mu))^{2m} + 1} ,
\end{equation}
from Assumption \ref{ass:Weyl}, we conclude that 
\begin{equation}\label{eq:IIest}
\begin{split} 
 \text{II} &\leq C_{t}\big\| \mathcal{W}^{(n)}_{z} (\omega_{\mu} - \omega_{\mu,\beta} ) \big\|_{\text{tr}} \leq C_{t} \big\| \mathcal{W}^{(n)}_{z} \omega_{\mu,\beta} (1 - \omega_{\mu,\beta}) \big\|_{\text{tr}} \\ & \leq C_{t} \Big\| \mathcal{W}^{(n)}_{z} \frac{1}{(\beta (H-\mu))^{2m} + 1} \Big\|_{\text{tr}} \leq C_{t}\varepsilon^{-2}\;,
\end{split} \end{equation} 
where we also used the invariance of the trace norm under hermitian conjugation. Consider now the first term in Eq. (\ref{eq:I+II}). Using again (\ref{eq:Up-W}),  we have 
\begin{equation}\label{eq:Iest}
\begin{split}
\text{I} &= \Big\| \mathcal{W}^{(n)}_{z} U(t;0) \Big( U(t;0)^{*} \widetilde{U}(t;0) \widetilde \omega_{0} \widetilde{U}(t;0)^{*} U(t;0) - \widetilde \omega_{0}\Big) \Big\|_{\text{tr}} \\
&\leq C_t \Big\|  \mathcal{W}^{(n)}_{z}  \Big( U(t;0)^{*} \widetilde{U}(t;0) \widetilde{\omega}_{0} \widetilde{U}(t;0)^{*} U(t;0) - \widetilde\omega_{0}\Big) \Big\|_{\text{tr}}\;.
\end{split}
\end{equation}
Using the Duhamel formula
\begin{equation*}
\begin{split}
&U(t;0)^{*} \widetilde{U}(t;0) \widetilde \omega_{0} \widetilde{U}(t;0)^{*} U(t;0) - \widetilde \omega_{0} \\
&\qquad = -i \varepsilon^{2} \int_{0}^{t} ds\, U(s;0)^{*} \Big[ (V * (\widetilde \varrho_{s} - \varrho_{s}))\,,  \widetilde{U}(s;0) \widetilde \omega_{0} \widetilde{U}(s;0)^{*} \Big] U(s;0)\;,
\end{split}
\end{equation*}
we find 
\begin{equation}\label{eq:UtildeU}
\begin{split}
&\Big\|  \mathcal{W}^{(n)}_{z}  \Big( U(t;0)^{*} \widetilde{U}(t;0) \widetilde\omega_{0} \widetilde{U}(t;0)^{*} U(t;0) - \widetilde \omega_{0}\Big) \Big\|_{\text{tr}} \\
&\qquad \leq \varepsilon^{2} \int_{0}^{t} ds\, \Big\|  \mathcal{W}^{(n)}_{z} U(s;0)^{*} \Big[ (V * (\widetilde\varrho_{s} - \varrho_{s}))\,,  \widetilde \omega_{s} \Big] U(s;0)   \Big\|_{\text{tr}} \\
&\qquad \leq C_t \varepsilon^{2} \int_{0}^{t} ds\, \Big\|  \mathcal{W}^{(n)}_{z} \Big[ (V * (\widetilde \varrho_{s} - \varrho_{s}))\,,  \widetilde \omega_{s} \Big] \Big\|_{\text{tr}}\;.
\end{split}
\end{equation}
Next, we write 
\begin{equation*}
(V * (\varrho_{s} - \widetilde \varrho_{s}))(\hat x) = \int dy\, V_{y}(\hat x) ( \varrho_{s}(y) - \widetilde \varrho_{s}(y) )
\end{equation*}
which gives:
\begin{equation}\label{eq:rhotilderho}
\Big\|  \mathcal{W}^{(n)}_{z} \Big[ (V * (\varrho_{s} - \widetilde \varrho_{s}))\,,  \widetilde \omega_{s} \Big] \Big\|_{\text{tr}}  \leq \int dy\, \big|\varrho_{s}(y) - \widetilde \varrho_{s}(y) \big| \Big\|  \mathcal{W}^{(n)}_{z} [ V_{y}(\hat x)\;,  \widetilde \omega_{s} ] \Big\|_{\text{tr}}\;.
\end{equation}
In order to set up a Gronwall argument, we proceed as follows. We start by writing:
\begin{equation}\label{eq:511}
\begin{split}
&\int dy\, \big|\varrho_{s}(y) - \widetilde \varrho_{s}(y) \big| \Big\|  \mathcal{W}^{(n)}_{z} [ V_{y}(\hat x)\;,  \widetilde \omega_{s} ] \Big\|_{\text{tr}} \\
&\qquad = \int_{|y - z| \leq \varepsilon^{-\delta} s} dy\, \big|\varrho_{s}(y) - \widetilde \varrho_{s}(y) \big| \Big\|  \mathcal{W}^{(n)}_{z} [ V_{y}(\hat x)\;,  \widetilde \omega_{s} ] \Big\|_{\text{tr}}\\&\quad\qquad  + \int_{|y - z| > \varepsilon^{-\delta} s} dy\, \big|\varrho_{s}(y) - \widetilde \varrho_{s}(y) \big| \Big\|  \mathcal{W}^{(n)}_{z} [ V_{y}(\hat x)\;,  \widetilde \omega_{s} ] \Big\|_{\text{tr}} = \text{I}^{(a)} + \text{I}^{(b)}\;.
\end{split}
\end{equation}
Consider the second term. From Corollary \ref{cor:commF}, using the fast decay of the potential $V$, we obtain 
\begin{equation}\label{eq:rhotilderho2}
\Big\|  \mathcal{W}^{(n)}_{z} [ V_{y}(\hat x)\;,  \widetilde \omega_{s} ] \Big\|_{\text{tr}} \leq \frac{C_s \veps^{-2}}{1 + |z - y|^{8}}
\end{equation}
for all $n \in \bN$ large enough. Using this estimate we have:
\begin{equation}\label{eq:IIestaa}
\begin{split}
\text{I}^{(b)} &\leq C_{s} \varepsilon^{-2} \int_{|y - z| > \varepsilon^{-\delta} s} dy\, (\varrho_{s}(y) + \widetilde \varrho_{s}(y))\, \frac{1}{1+|z-y|^{8}} \\
&\leq C_{s} \varepsilon^{-2} \int dy\, (\varrho_{s}(y) + \widetilde \varrho_{s}(y))\, \frac{(1+|y-z|)^{4}}{(1 + \varepsilon^{-\delta} s)^{4}}\frac{1}{1+|z-y|^{8}} \\
&\leq \frac{C_{s} \varepsilon^{-2}}{1 + (\varepsilon^{-\delta} s)^{4}} \int dy\, (\varrho_{s}(y) + \widetilde \varrho_{s}(y)) \frac{1}{1+|z-y|^{4}}\\
&= \frac{C_{s} \varepsilon^{-2}}{1 + (\varepsilon^{-\delta} s)^{4}} \tr\, (\omega_{s} + \tilde \omega_{s}) \mathcal{W}_{z}\;,
\end{split}
\end{equation}
which we can further estimate as, using the propagation of locality of Eq. (\ref{eq:Up-W}):
\begin{equation}
\begin{split}
\text{I}^{(b)} &\leq \frac{C_{s} \varepsilon^{-2}}{1 + (\varepsilon^{-\delta} s)^{4}} \tr\,( \omega_{\mu} + \omega_{\mu,\beta}) \mathcal{W}_{z} \\
&\leq  \frac{2 C_{s} \varepsilon^{-2}}{1 + (\varepsilon^{-\delta} s)^{4}} \tr\,\omega_{\mu,\beta} \mathcal{W}_{z} \\
&\leq \frac{C_{s} \varepsilon^{-5}}{1 + (\varepsilon^{-\delta} s)^{4}}\;,
\end{split}
\end{equation}
where the second inequality follows from the trivial bound $\omega_{\mu} \leq 2\omega_{\mu,\beta}$, and the last from the boundedness of the density of the Fermi-Dirac state, Proposition \ref{prop:localsc}.

Consider now the term $\text{I}^{(a)}$. Observe that, since $s\leq t$, the integration variable $y$ is in $\Lambda$, since the ball $B_{\varepsilon^{-\delta} s}(z)$ is contained in the region $\Lambda$. We have, writing $V_{y}(\hat x) = \mathcal{W}_{y}^{(2n)}(\hat x) \widetilde{V}_{y}(\hat x)$:
\begin{equation}\label{eq:I1I2aa}
\begin{split}
\text{I}^{(a)} &\leq \int_{|y - z| \leq \varepsilon^{-\delta} s} dy\, \big|\varrho_{s}(y) - \widetilde \varrho_{s}(y) \big| \Big\|  \mathcal{W}^{(n)}_{z} [ V_{y}(\hat x)\;,  \widetilde \omega_{s} ] \Big\|_{\text{tr}} \\
&\leq \int_{|y - z| \leq \varepsilon^{-\delta} s} dy\, \big|\varrho_{s}(y) - \widetilde \varrho_{s}(y) \big| \Big(\Big\|  \mathcal{W}^{(n)}_{z} \mathcal{W}_{y}^{(2n)} [ \widetilde{V}_{y}(\hat x)\;,  \widetilde \omega_{s} ] \Big\|_{\text{tr}} + \Big\|  \mathcal{W}^{(n)}_{z}  [ \mathcal{W}_{y}^{(2n)} \;,  \widetilde \omega_{s} ] \widetilde{V}_{y}(\hat x) \Big\|_{\text{tr}}\Big) \\
&= \text{I}^{(a)}_{1} + \text{I}^{(a)}_{2}\;.
\end{split}
\end{equation}
We can estimate the term $\text{I}^{(a)}_{1}$ as:
\begin{equation}
\begin{split}
\text{I}^{(a)}_{1} &\leq C\int_{|y - z| \leq \varepsilon^{-\delta} s} dy\, \big|\varrho_{s}(y) - \widetilde \varrho_{s}(y) \big| \frac{1}{1+|z-y|^{4n}} \Big\| \mathcal{W}_{y}^{(n)} [ \widetilde{V}_{y}(\hat x)\;,  \widetilde \omega_{s} ] \Big\|_{\text{tr}} \\
&\leq C \Big(\sup_{p\in \mathbb{R}^{3}} \sup_{y \in \Lambda} \frac{1}{1+|p|}\Big\| \mathcal{W}_{y}^{(n)} [ e^{ip\cdot \hat x},  \widetilde \omega_{s} ] \Big\|_{\text{tr}}\Big) \int dy\, \big|\varrho_{s}(y) - \widetilde \varrho_{s}(y) \big| \mathcal{W}^{(n)}_{z}(y) \\
&\leq C_{s} \varepsilon^{-2} \int dy\, \big|\varrho_{s}(y) - \widetilde \varrho_{s}(y) \big| \mathcal{W}^{(n)}_{z}(y)\;,
\end{split}
\end{equation}
by the propagation of the local semiclassical structure, Proposition \ref{prop: semiclassical_prop-main}, and Remark \ref{rem:propa}. Next, using that:
\begin{equation*}
 \int dy \,\big|\varrho_{s}(y) - \widetilde \varrho_{s}(y) \big| \mathcal{W}^{(n)}_{z}(y) = \tr\, J_s (\omega_{s} - \widetilde \omega_{s})\;,
\end{equation*}
where $J_s$ is the operator of multiplication by the function $J_s (x) = \mathcal{W}^{(n)}_{z}(x) \text{sign}( \varrho_{s}(x) - \widetilde \varrho_{s}(x))$, we can further estimate:
\begin{equation*}
\begin{split}
\tr J (\omega_{s} - \widetilde \omega_{s}) &\leq \Big\| J_s (\omega_{s} - \widetilde \omega_{s}) \Big\|_{\text{tr}} \leq \Big\| \mathcal{W}^{(n)}_{z} (\omega_{s} - \widetilde \omega_{s}) \Big\|_{\text{tr}}\;. 
\end{split}
\end{equation*}
Therefore, we obtained:
\begin{equation}\label{eq:I1est00}
\text{I}^{(a)}_{1} \leq C_{s} \varepsilon^{-2} \Big\| \mathcal{W}^{(n)}_{z} (\omega_{s} - \widetilde \omega_{s}) \Big\|_{\text{tr}}\;.
\end{equation}
Consider now the term $\text{I}^{(a)}_{2}$ in (\ref{eq:I1I2aa}). Here we shall use that, by the assumptions on the potential,
\begin{equation}\label{eq:517est}
\Big\|  \mathcal{W}^{(n)}_{z}  [ \mathcal{W}_{y}^{(2n)} \;,  \widetilde \omega_{s} ] \widetilde{V}_{y}(\hat x) \Big\|_{\text{tr}} \leq C\sup_{p\in \mathbb{R}^{3}} \frac{1}{1+|p|} \Big\|  \mathcal{W}^{(n)}_{z}  [ e^{ip\cdot x},  \widetilde \omega_{s} ] \mathcal{W}^{(2n)}_{y} \Big\|_{\text{tr}}\;.
\end{equation}
We estimate the right-hand side using Proposition \ref{prop: semiclassical_prop-main}. The key remark is that $y\in \Lambda$, and that the localizer associated with $y$ has a larger power,  $2n$, than the initial localizer associated with $z$, which has power $n$. Therefore, from the first  bound in (\ref{eq:eipx-u-v}), we obtain:
\begin{equation}\label{eq:518}
\sup_{p\in \mathbb{R}^{3}} \frac{1}{1+|p|} \Big\|  \mathcal{W}^{(n)}_{z}  [ e^{ip\cdot x},  \widetilde \omega_{s} ] \mathcal{W}^{(2n)}_{y} \Big\|_{\text{tr}} \leq \frac{C_s \varepsilon^{-2}}{1 + |z-y|^{4n}}\;,
\end{equation}
and by (\ref{eq:517est}), (\ref{eq:518}), we have:
\begin{equation}
\Big\|  \mathcal{W}^{(n)}_{z}  [ \mathcal{W}_{y}^{(2n)} \;,  \widetilde \omega_{s} ] \widetilde{V}_{y}(\hat x) \Big\|_{\text{tr}} \leq \frac{C_{s}\varepsilon^{-2}}{1 + |y-z|^{4n}} = C_{s}\varepsilon^{-2} \mathcal{W}^{(n)}_{z}(y)\;.
\end{equation}
Plugging this estimate in $\text{I}^{(a)}_{2}$, we have:
\begin{equation}\label{eq:I2est00}
\text{I}^{(a)}_{2} \leq C_{s} \varepsilon^{-2}\int dy\, \big|\varrho_{s}(y) - \widetilde \varrho_{s}(y) \big| \mathcal{W}^{(n)}_{z}(y) \leq C_{s} \varepsilon^{-2} \big\| \mathcal{W}^{(n)}_{z} (\omega_{s} - \widetilde \omega_{s}) \big\|_{\text{tr}}\;.
\end{equation}
Putting together (\ref{eq:I1I2aa}), (\ref{eq:I1est00}), (\ref{eq:I2est00}) we obtain:
\begin{equation}
\text{I}^{(a)} \leq C_{s} \varepsilon^{-2} \big\| \mathcal{W}^{(n)}_{z} (\omega_{s} - \widetilde \omega_{s}) \big\|_{\text{tr}}\;;
\end{equation}
this bound together with (\ref{eq:UtildeU}), (\ref{eq:rhotilderho}), (\ref{eq:511}), (\ref{eq:IIestaa}) implies:
\begin{equation}\label{eq:diff}
\begin{split}
&\Big\|  \mathcal{W}^{(n)}_{z}  \Big( U(t;0)^{*} \widetilde{U}(t;0) \widetilde\omega_{0} \widetilde{U}(t;0)^{*} U(t;0) - \widetilde\omega_{0}\Big) \Big\|_{\text{tr}} \\
&\qquad \qquad \leq C_{t} \int_{0}^{t} ds\, \Big( \big\| \mathcal{W}^{(n)}_{z} (\omega_{s} - \widetilde \omega_{s}) \big\|_{\text{tr}} + \frac{\varepsilon^{-3}}{1 + (\varepsilon^{-\delta} s)^{4}}\Big) \\
&\qquad \qquad \leq C_{t} \varepsilon^{-3 + \delta} + C_{t} \int_{0}^{t} ds\, \big\| \mathcal{W}^{(n)}_{z} (\omega_{s} - \widetilde \omega_{s}) \big\|_{\text{tr}}\;.
\end{split}
\end{equation}
Combining this estimate with (\ref{eq:I+II}), (\ref{eq:IIest}), we finally get 
\begin{equation*}
\Big\| \mathcal{W}^{(n)}_{z} (\omega_{t} - \widetilde{\omega}_{t}) \Big\|_{\text{tr}} \leq C_{t} \varepsilon^{-3+\delta} + C_{t}  \int_{0}^{t} ds\, \Big\| \mathcal{W}^{(n)}_{z} ( \omega_{s} - \widetilde \omega_{s}  ) \Big\|_{\text{tr}}\;.
\end{equation*}
Thus, the bound (\ref{eq:Best}) follows from Gronwall's lemma and from (\ref{eq:OtoW}).
\paragraph{Bound for the term $\text{B}$.} We claim that:
\begin{equation}\label{eq:Aest}
\Big| \tr\, \mathcal{O}_{z} ( \gamma^{(1)}_{t} - \widetilde\omega_{t}) \Big| \leq C\exp (c \exp (ct)) \varepsilon^{-3+\delta} \;.
\end{equation}
Recall that $\gamma^{(1)}_{t}$ is the reduced one-particle density matrix of the state $\psi_t = e^{-i \mathcal{H} t / \varepsilon} R_{\omega_\mu} \xi \in \mathcal{F}(\frak{h})$. Hence
\[
\tr\, \mathcal{O}_{z} \gamma^{(1)}_{t} = \langle e^{-i \mathcal{H} t / \varepsilon} R_{\omega_\mu} \xi , d\Gamma (\mathcal{O}_z) e^{-i \mathcal{H} t / \varepsilon} R_{\omega_\mu} \xi \rangle\;. \]
Switching to the doubled Fock space (see Section \ref{sec: mixed-states}), we can also write 
\[ 
\tr\, \mathcal{O}_{z} \gamma^{(1)}_{t} = 
\big\langle U  (R_{\omega_\mu} \otimes R_{\bar{\omega}_\mu}) (\xi \otimes \bar{\xi}), e^{i\mathcal{L} t/\veps} d\Gamma_l (\mathcal{O}_{z} ) e^{-i \mathcal{L} t / \varepsilon} U (R_{\omega} \otimes R_{\bar{\omega}_\mu}) (\xi \otimes \bar{\xi}) \big\rangle\;. \]
Denoting by $\mathcal{R}_{\omega_{\mu,\beta}}$ the Bogoliubov transformation on $\mathcal{F} (\frak{h} \oplus \frak{h})$, generating the mixed quasi-free state with reduced density $\omega_{\mu,\beta}$ (see Section \ref{sec:bogpure}), we define 
\[
\Theta := \mathcal{R}_{\omega_{\mu,\beta}}^*  U (R_{\omega_\mu} \otimes R_{\bar{\omega}_\mu}) (\xi \otimes \bar{\xi}) \in \mathcal{F} (\frak{h} \oplus \frak{h})\;,
\]
so that 
\[ \tr\, \mathcal{O}_{z} \gamma^{(1)}_{t} =  \big\langle \mathcal{R}_{\omega_{\mu,\beta}} \Theta , e^{i\mathcal{L} t/\veps} d\Gamma_l (\mathcal{O}_{z}) e^{-i \mathcal{L} t / \varepsilon} \mathcal{R}_{\omega_{\mu,\beta}} \Theta \big\rangle\;. \]
With Proposition \ref{prop:localsc}, it follows that we can apply Theorem \ref{thm:T>0} to compute the r.h.s.~of the last equation and to show (\ref{eq:Aest}), if we can prove that 
\begin{equation}\label{eq:Theta-bds} \sup_{z \in \mathbb{R}^{3}} \big\langle \Theta, d\Gamma_\sigma (\W^{(n)}_z) \Theta \big\rangle \leq C \veps^{-3} \;,
\qquad \sup_{z \in\Lambda} 
\big\langle \Theta, d\Gamma_\sigma (\W^{(n)}_z) \Theta \big\rangle \leq C \veps^{-3+\delta} \;,
\end{equation} 
for $\sigma = l,r$ and for $n \in \bN$ sufficiently large. Let us show (\ref{eq:Theta-bds}) for $\sigma = l$, the case $\sigma = r$ can be handled similarly. From Lemma \ref{lemma: Bogolubov-observables}, we have
\begin{equation}\label{eq:act-Rbeta} \begin{split}  \mathcal{R}_{\omega_{\mu,\beta}} d\Gamma_l (\W_z^{(n)}) \mathcal{R}^*_{\omega_{\mu,\beta}}  = \; &\tr \, \W^{(n)}_z \omega_{\mu,\beta} + d\Gamma_l (u_{\mu,\beta} \W_z^{(n)} u_{\mu,\beta}) - d\Gamma_r (v_{\mu,\beta} \W_z^{(n)} v_{\mu,\beta}) \\ &+ d\Gamma_{lr}^+ (u_{\mu,\beta} \W_z^{(n)} v_{\mu,\beta}) + d\Gamma^-_{rl} (v_{\mu,\beta} \W^{(n)}_z u_{\mu,\beta}) \end{split} 
\end{equation} 
where we set $u_{\mu,\beta} = \sqrt{1-\omega_{\mu,\beta}}$, $v_{\mu,\beta} = \sqrt{\omega_{\mu,\beta}}$; since the eigenfunctions of $H$ can be chosen real-valued, we can take here $\bar{v}_{\mu,\beta} = v_{\mu,\beta}$. Next, we compute the action of $(R_{\omega_\mu} \otimes R_{\bar{\omega}_\mu})$ on the terms on the r.h.s.~of (\ref{eq:act-Rbeta}). To this end, we use the identity (\ref{eq:Romega2}). Again, since the eigenfunctions of $H$ can be chosen real-valued, in (\ref{eq:Romega2}) we can take $v_{\mu} = \omega_\mu$. Setting $u_\mu = 1-\omega_\mu$, we find, considering for example the action of $(R_{\omega_\mu} \otimes R_{\bar{\omega}_\mu})$ on the term $d\Gamma_l (u_{\mu,\beta} \W_z^{(n)} u_{\mu,\beta})$ on the r.h.s. of (\ref{eq:act-Rbeta}), 
\begin{equation}\label{eq:act-Rw1} \begin{split}  U(R^*_{\omega_\mu} \otimes R^*_{\bar{\omega}_\mu})  &U^* d\Gamma_l (u_{\mu,\beta} \W_z^{(n)} u_{\mu,\beta}) U (R_{\omega_\mu} \otimes R_{\bar{\omega}_\mu})U^{*} \\ =\; &\tr \, \omega_\mu u_{\mu,\beta} \W^{(n)}_z u_{\mu,\beta}\,  \omega_\mu + d\Gamma_l (u_\mu u_{\mu,\beta} \W_z^{(n)} u_{\mu,\beta} \, u_\mu) - d\Gamma_l (\omega_\mu u_{\mu,\beta} \W_z u_{\mu,\beta} \, \omega_\mu) \\ &+d\Gamma_l^+ (u_\mu u_{\mu,\beta} \W_z^{(n)}  u_{\mu,\beta} \, \omega_\mu) + d\Gamma_{l}^- (\omega_\mu u_{\mu,\beta} \W_z^{(n)} u_{\mu,\beta} \, u_\mu)\;. \end{split} \end{equation}  
Consider the first term on the right-hand side of (\ref{eq:act-Rw1}). If $z\in \mathbb{R}^{3}$ we have, using that $\omega_\mu$ and $\omega_{\mu,\beta}$ commute, that $\omega_{\mu} \leq 2\omega_{\mu,\beta}$ and the fact that $\omega_{\mu,\beta}$ has bounded density, Proposition \ref{prop:localsc}:
\[
\tr \, \omega_\mu u_{\mu,\beta} \W^{(n)}_z u_{\mu,\beta}\,  \omega_\mu \leq 2\tr\, \W_z^{(n)} \omega_{\mu,\beta} \leq C\varepsilon^{-3}\;.
\]
Instead, if $z\in \Lambda$, we obtain
\[ \tr \, \omega_\mu u_{\mu,\beta} \W^{(n)}_z u_{\mu,\beta}\,  \omega_\mu = \tr\, \W_z^{(n)} (1-\omega_{\mu,\beta}) \omega_\mu \leq C \veps^{-2}\;, \]
arguing as we did to prove (\ref{eq:IIest}). This proves that the contribution arising from the first term on the r.h.s.~of (\ref{eq:act-Rw1}) fulfills (\ref{eq:Theta-bds}). 

To bound the expectation of the second term on the r.h.s.~of (\ref{eq:act-Rw1}) for $z\in \mathbb{R}^{3}$, we can use twice the estimate \eqref{eq:uWu3}. We obtain:
\begin{equation}
\begin{split}
&u_\mu u_{\mu,\beta} \W_z^{(n)} u_{\mu,\beta} \, u_\mu \leq 4 (u_{\mu} - 1) (u_{\mu,\beta} - 1) \W_z^{(n)} (u_{\mu} - 1) (u_{\mu,\beta} - 1)\\&\qquad\qquad  + 8 (u_{\mu} - 1) \W_z^{(n)} (u_{\mu} - 1) + 8 (u_{\mu,\beta} - 1) \W_z^{(n)} (u_{\mu,\beta} - 1) + 16 \W_z^{(n)}\;,
\end{split}
\end{equation}
which gives, proceeding as in (\ref{eq:uWu3b}):
\begin{equation}
\begin{split}
\langle U (\xi \otimes \bar{\xi}), d\Gamma_l (u_\mu u_{\mu,\beta} \W_z^{(n)} u_{\mu,\beta} \, u_\mu) U (\xi \otimes \bar{\xi})  \rangle &\leq 16 \langle U (\xi \otimes \bar{\xi}), d\Gamma_l(  \W_z^{(n)} ) U (\xi \otimes \bar{\xi})  \rangle \\
&\quad + C \tr\, \W_z^{(n)} \omega_{\mu,\beta}\;,
\end{split}
\end{equation}
where we used that $u_{\mu} - 1 = \omega_{\mu}$ and that $(u_{\mu,\beta} - 1)^{2} \leq \omega_{\mu,\beta}$. With (\ref{eq:assxi}), this implies that
\begin{equation}\label{eq:uuWuu} \langle U (\xi \otimes \bar{\xi}), d\Gamma_l (u_\mu u_{\mu,\beta} \W_z^{(n)} u_{\mu,\beta} \, u_\mu) U (\xi \otimes \bar{\xi}) \rangle \leq C \veps^{-3} \end{equation} for all $z \in \bR^3$. Instead, for $z\in \Lambda$, we can first use the estimate (\ref{eq:IIest}) to replace $u_\mu$ with $u_{\mu,\beta}$, showing that 
\begin{equation}
\begin{split}
&\langle U (\xi \otimes \bar{\xi}), d\Gamma_l (u_\mu u_{\mu,\beta} \W_z^{(n)} u_{\mu,\beta} \, u_\mu) U (\xi \otimes \bar{\xi}) \rangle \\
 &\qquad \leq C \langle U (\xi \otimes \bar{\xi}), d\Gamma_l ((1-\omega_{\mu,\beta}) \W_z^{(n)} (1-\omega_{\mu,\beta})) U (\xi \otimes \bar{\xi}) \rangle + C \veps^{-2}\;.
 \end{split} 
\end{equation}
Then we can follow \eqref{eq:uWu-comm} to conclude that 
\[ \begin{split}  &\langle U (\xi \otimes \bar{\xi}), d\Gamma_l (u_\mu u_{\mu,\beta} \W_z^{(n)} u_{\mu,\beta} \, u_\mu) U (\xi \otimes \bar{\xi}) \rangle \\ &\qquad \leq C \langle \xi, d\Gamma (\W_z^{(n)}) \xi \rangle + C \langle \xi,  d\Gamma ( | [\W_z^{(n)} , \omega_{\mu,\beta}] |^2) \xi \rangle \\&\qquad \leq C \veps^{-3+ \delta} + C \veps^{-2} \leq C \veps^{-3+\delta} \end{split} \]
for all $z \in \Lambda$. This proves that the contribution arising from the second term on the r.h.s.~of (\ref{eq:act-Rw1}) fulfills (\ref{eq:Theta-bds}). 

To prove that the third term on the r.h.s.~of (\ref{eq:act-Rw1}) satisfies (\ref{eq:Theta-bds}), it is enough, by Lemma \ref{lemma: bound-Fock-operator}, to observe that
\begin{equation*}
\| \omega_\mu u_{\mu,\beta} \W_z^{(n)} u_{\mu,\beta} \omega_\mu \|_\text{tr} \leq \tr \, \omega_\mu (1-\omega_{\mu,\beta}) \W_z^{(n)} \leq \left\{\begin{array}{cc} C \veps^{-3} & \text{for $z\in \mathbb{R}^{3}$} \\ C\varepsilon^{-2} & \text{for $z\in \Lambda$} \end{array}\right.
\end{equation*}
as we argued for the first term in the right-hand side of (\ref{eq:act-Rw1}). 

As for the last two terms on the r.h.s of (\ref{eq:act-Rw1}), to show the first bound in (\ref{eq:Theta-bds}) we can proceed as in (\ref{eq:apri2}), replacing $u_{t}$ by $u_\mu u_{\mu,\beta}$ and $\overline{v_{t}}$ by $u_{\mu,\beta} \omega_\mu$ (again, we are using that the eigenfunctions of $H$ can be chosen real). We obtain:
\begin{equation}
\begin{split}
&\Big|\langle U (\xi \otimes \bar{\xi}), d\Gamma_l^+ (u_\mu u_{\mu,\beta} \W_z^{(n)}  u_{\mu,\beta} \, \omega_\mu) U (\xi \otimes \bar{\xi}) \rangle\Big| \\
&\quad \leq \Big(\int dy\, \W_z^{(n)}(y) \| (u_{\mu,\beta} \omega_{\mu})_{y} \|_{\frak{h}}^{2}\Big)^{1/2} \Big( \int dy\, \W_z^{(n)}(y) \big\| a(( u_{\mu} u_{\mu,\beta} )_{y}) U (\xi \otimes \bar{\xi}) \big\|^{2}  \Big)^{1/2} \\
&\quad \leq C \tr \W_z^{(n)}(y) \omega_{\mu,\beta} + C\langle U (\xi \otimes \bar{\xi}), d\Gamma_l (u_\mu u_{\mu,\beta} \W_z^{(n)} u_{\mu,\beta} \, u_\mu) U (\xi \otimes \bar{\xi}) \rangle\\
&\quad \leq C\varepsilon^{-3}\;,
\end{split}
\end{equation}
where we used the boundedness of the density to estimate the first term, and (\ref{eq:uuWuu}) for the second term. To prove the second bound in (\ref{eq:Theta-bds}), we can observe that, for $z \in \Lambda$, 
\[ \begin{split}  \| u_\mu u_{\mu,\beta} \W^{(n)}_z &u_{\mu,\beta} \omega_\mu \|_\text{tr} \\ &\leq \| \W_z^{(n)} (\omega_\mu - \omega_{\mu,\beta}) \|_\text{tr} + \| \W_z^{(n)} (1-\omega_{\mu,\beta}) \omega_{\mu,\beta} \|_\text{tr} + \| [ \W_z^{(n)} , u_{\mu,\beta} ] \|_\text{tr} \leq C \veps^{-2} \end{split} \]
as it follows from Corollary \ref{cor:commF} and from (\ref{eq:oo})-(\ref{eq:IIest}). 

The action of $(R_{\omega_\mu} \otimes R_{\bar{\omega}_\mu})$ on the other terms on the r.h.s.~of (\ref{eq:act-Rbeta}) can be handled similarly. We skip here the details. We observe, however, that conjugation of the term $-d\Gamma_r (v_{\mu,\beta} \W_z^{(n)} v_{\mu,\beta})$ produces the contracted contribution 
\[ - \tr \, \omega_\mu v_{\mu,\beta} \W_z^{(n)} v_{\mu,\beta} \omega_\mu = - \tr \W_z^{(n)} \omega_{\mu,\beta} \omega_\mu \;;\]
combined with the first contribution on the r.h.s.~of (\ref{eq:act-Rbeta}) we get:
\begin{equation}
\tr \, \W^{(n)}_z \omega_{\mu,\beta} (1 - \omega_\mu)  \leq \left\{\begin{array}{cc} C \veps^{-3} & \text{for $z\in \mathbb{R}^{3}$} \\ C\varepsilon^{-2} & \text{for $z\in \Lambda$} \end{array}\right.
\end{equation}
where the first bound follows from the boundedness of the density and the second arguing as we did to prove (\ref{eq:IIest}). This concludes the proof of (\ref{eq:Theta-bds}) and of Theorem~\ref{thm:main}. 

\qed

\appendix

\section{Bound on the density of the Fermi-Dirac distribution}
\label{sec: appen-apriori}

In this section we will show that, for all $\beta \geq 1$, 
\begin{equation}\label{eq:omegapriori}
\tr\, \W_{z} \omega_{\mu,\beta} \leq C \veps^{-3}
\end{equation}
for some universal constant $C >0$, where for brevity $\W_{z} \equiv \W_{z}^{(1)}$, where $\omega_{\mu,\beta}$ is the Fermi-Dirac distribution (\ref{eq:FD}) associated with the Hamiltonian $H = \sqrt{1 - \varepsilon^{2}\Delta}+V_\text{ext}$ satisfying Assumption \ref{ass:Weyl} (in particular, with $0 \leq V_\text{ext}\leq C x^2$). The bound (\ref{eq:omegapriori}) proves the statement in Proposition \ref{prop:localsc} that $\omega_{\mu,\beta}$ satisfies Assumption \ref{ass:bd}. We will actually show that 
\begin{equation}\label{eq: apriori-exp}
\tr\, \W_{z}  e^{-(H-\mu)} \leq C \veps^{-3} \;,
\end{equation}
from which the claim on $\tr\, \W_{z} \omega_{\mu,\beta}$ follows by the operator inequality 
\begin{equation}\label{eq:omegaexp}
\omega_{\mu,\beta} \leq  e^{-(H-\mu)}, \qquad \text{for any $\beta \in [1,\infty]$.}
\end{equation}
To check (\ref{eq:omegaexp}), it suffices to note that $\omega_{\mu,\beta} e^{H-\mu} = f_{\beta}(e^{H-\mu})$ with $f_{\beta}(x) = (x^{-1} + x^{\beta - 1})^{-1}$ and that $f_{\beta}(x) \leq 1$ for $\beta \in [1,\infty]$ and $x \geq 0$.
%consider the Fermi-Dirac function $\omega^{\beta}(\lambda) = 1 / (1 + e^{\beta (\lambda - \mu)})$. We observe that $\omega^{\beta}(\mu)  = 1/2$ for all $\beta > 0$. Also, $\omega^{\beta}(\lambda) > \omega^{\beta}(\mu)$ for all $\lambda < \mu$. Thus, for $\lambda \leq \mu$ we have $\omega^{\beta}(\lambda)  \leq 2 \omega^{1}(\lambda)$ simply because the right-hand side is bigger than $1$. Instead, for $\lambda > \mu$ we use that the function $\omega^{\beta}(\lambda)$ is decreasing as $\beta$ increases; hence, $\omega^{\beta}(\lambda) \leq \omega^{1}(\lambda)$ for $\lambda > \mu$ and for all $\beta \geq 1$.

Next, we note that the bound \eqref{eq: apriori-exp} is straightforward if $V_\text{ext} =0$. In fact, we have:
\begin{equation}\label{eq: free-apriori-density}
\begin{split}
\tr\, \W_{z}  e^{-(\sqrt{1-\veps^{2}\Delta} - \mu)}  & 
= e^{\mu}\int dx\, \W_{z} (x) \int dk\, e^{-\sqrt{1 + \veps^{2}k^{2}}} = C \veps^{-3} \;.
\end{split}
\end{equation}
To prove \eqref{eq: apriori-exp} in the case $V_\text{ext} \neq 0$, we shall use a Feynman-Kac formula to get rid of $V_\text{ext}$ and thus conclude by \eqref{eq: free-apriori-density}. Feynman-Kac formulas for the pseudo-relativistic Schrödinger operators were first obtained in \cite{CS} in terms of a suitable Lévy process in place of the usual Brownian motion. The Lévy process relevant for our purposes can be written in terms of Brownian motion $(B_{t})_{t\geq0}$ and a one-dimensional subordinator $(T_{t}) _{t\geq 0}$, see the pedagogical exposition in \cite{Carmona} and \cite[Sections 2.4 and 3.6]{Hiroshima}.
\begin{proposition}[Feynman-Kac formula] \label{prop: Feynamn-Kac formula}
Let $(B^{x}_{t})_{t\geq 0}$ denote Brownian motion starting at $x \in \R^{3}$ and let $\mathbb{E}^{x}$ be the associated expectation. Let $(T_{t})_{t \geq 0}$ be an independent $\R $-valued Lévy process such that, denoting the associated expectation by $\mathbb{E}$, we have for $u \geq 0$
\begin{equation*}
\mathbb{E} e^{-u T_{t}} = e^{-t\big(\sqrt{2\veps^{2}u + 1}-1\big)} \;.
\end{equation*}
Let $f,g \in \mathfrak{h}$ and let $V_\text{ext} \in L^{\infty}(\R^{3})$ and continuous. Then,
\begin{equation}\label{eq: rel-FK-formula}
\langle f, e^{-t (\sqrt{1-\veps^{2}\Delta}-1 + V_\text{ext}) } g \rangle = \int d x \, \mathbb{E} \,\mathbb{E}^{x} \Big(  \overline{f(B^{x}_{T_{0}})} g(B^{x}_{T_{t}}) e^{-\int_{0}^{t} V_\text{ext} (B^{x}_{T_{s}}) ds}  \Big) \;.
\end{equation}
\end{proposition}

As an application of this formula, we obtain the following bound.
\begin{corollary}\label{cor:FK}
Let $H= \sqrt{1-\veps^{2}\Delta} + V_\text{ext}$ with $V_\text{ext} \in L^{\infty}(\R^{3})$, continuous and $V_\text{ext} \geq 0$. Then, the following bound holds true, for any $t\geq 0$, $f \in L^{2}(\R^{3})$:
\begin{equation*}
\langle f, e^{-t H} f \rangle \leq \langle f, e^{-t \sqrt{1-\veps^{2}\Delta}} f \rangle \;.
\end{equation*}
\end{corollary}
\begin{proof}
First of all, we note that we can swap the order of integration in \eqref{eq: rel-FK-formula} by Fubini-Tonelli, since $\exp\big(- \int_{0}^{t} V_\text{ext} (B^{x}_{T_{s}}) d s \big)$ is positive and $|f|$ and $|g|$ are in $\mathfrak{h}$. Accordingly, denoting $d\mathbb{P}^{x,y}_{[0,t]}$ the Wiener measure of Brownian paths from $x$ to $y$ over the interval $[0,t]$, we can write
\begin{equation*}
\begin{split}
\langle f, e^{-t (H-1) } f \rangle
& = \mathbb{E} \int dx dy \overline{f(x)} f(y)\int d \mathbb{P}^{x,y}_{[T_{0},T_{t}]}(\nu) e^{-\int_{0}^{t} V_\text{ext} (\nu_{T_{s}}) ds} 
\\
& = \int d \mathbb{P}^{f}_{[T_{0},T_{t}]}(\nu) e^{-\int_{0}^{t} V_\text{ext} (\nu_{T_{s}}) ds} \;,
\end{split}
\end{equation*}
where we used that, for almost every $\nu$ $d\mathbb{P}^{x,y}_{[0,t]}(\nu)$ is positive definite (with respect to integration in $x,y$) and defined the measure $d \mathbb{P}^{f}_{[0,t]}(\nu) := \int dx d y\overline{f(x)} d\mathbb{P}^{x,y}_{[0,t]}(\nu) f(y)$.
The claim then follows from the positivity of $V_\text{ext}$. 
\end{proof}
Corollary \ref{cor:FK} implies that, for $V_\text{ext} \geq 0$ bounded:
\begin{equation}\label{eq:corFK}
\tr\, \W_{z}  e^{-(H-\mu)} \leq \tr\, \W_{z}  e^{-(H_{0}-\mu)}\;.
\end{equation}
In our case, however, the potential $V_\text{ext}$ is not bounded, instead, it is confining. We will overcome this by an approximation argument. This is the content of the next lemma.
\begin{lemma}\label{lemma: strong convergence}
Let $H =\sqrt{1-\veps^{2}\Delta} + V_\text{ext}$ with $V_\text{ext} \geq 0$ and $V_\text{ext} (x)\leq C|x|^{2}$. For any $L >0$, let:
\begin{equation*}
V^{L}:= V_\text{ext} \chi(x/L) \;, \qquad H^{L} := \sqrt{1-\veps^{2}\Delta}  + V^{L} \;.
\end{equation*}
where $\chi(x)$ is a smooth cut-off function, such that $\chi(x) = 1$ for $| x |  \leq 1$ and $\chi(x) = 0$ for $|x| >2$. Then, 
\begin{equation}\label{eq:Lapprox}
\lim_{L\to \infty} \tr\, \mathcal{W}_{z} \frac{1}{1 + e^{\beta (H^{L} - \mu)}} = \tr\, \mathcal{W}_{z} \frac{1}{1 + e^{\beta (H - \mu)}}\;.
\end{equation}
\end{lemma}
Thanks to this lemma, we have:
\begin{equation}
\begin{split}
\tr\, \mathcal{W}_{z} \frac{1}{1 + e^{\beta (H - \mu)}} &= \lim_{L\to \infty} \tr\, \mathcal{W}_{z} \frac{1}{1 + e^{\beta (H^{L} - \mu)}} \\
&\leq \limsup_{L\to \infty}\, \tr\, \mathcal{W}_{z} e^{-(H^{L}-\mu)} \\
&\leq \tr\, \mathcal{W}_{z} e^{-(H_{0}-\mu)}\;,
\end{split}
\end{equation}
where the first inequality follows from (\ref{eq:omegaexp}) while the second follows from Eq.~(\ref{cor:FK}). This bound, combined with (\ref{eq: free-apriori-density}), proves the claim (\ref{eq:omegapriori}). Let us now prove Lemma \ref{lemma: strong convergence}.
\begin{proof}[Proof of Lemma \ref{lemma: strong convergence}.]
To begin, observe that, by the Golden-Thompson inequality, $\omega_{\mu,\beta}$ is trace class. In fact:
\begin{equation*}
\tr\, \omega_{\mu,\beta} \leq \tr\, e^{-\beta (H - \mu)} \leq \tr\, e^{-\beta (H_{0} - \mu)} e^{-\beta V_\text{ext}} = e^{\beta \mu}\int d k \, e^{-\beta \sqrt{1+\veps^{2}k^{2}}}  \int dx \, e^{-\beta V_\text{ext}(x)} <\infty\;.
\end{equation*}
Thus, to show (\ref{eq:Lapprox}) it is enough to prove, for any function $f\in L^{2}(\mathbb{R}^{3})$:
\begin{equation}\label{eq:B7}
\lim_{L \to \infty} \langle f,  \mathcal{W}_{z} \frac{1}{1 + e^{\beta (H^{L} - \mu)}} f \rangle = \langle f,  \mathcal{W}_{z} \frac{1}{1 + e^{\beta (H - \mu)}} f \rangle\;.
\end{equation}
By the boundedness of the operators in the trace and by density in $L^{2}(\mathbb{R}^{3})$, it is enough to consider $f\in C^{\infty}_{\text{c}}(\mathbb{R}^{3})$. To prove (\ref{eq:B7}) we use the integral representation 
\begin{equation}\label{eq:omegaint0}
\frac{1}{1+e^{\beta (H-\mu)}}  = \int_{\mathcal{C}} \frac{dz}{2\pi i} \frac{1}{1 + e^{\beta (z - \mu)}} \frac{1}{H - z}\;,
\end{equation}
where $\mathcal{C}$ is an unbounded clockwise path that encloses the spectrum of $H$. More precisely:
\begin{equation}\label{eq:Cdef0}
\mathcal{C} = \mathcal{C}_{1} \cup \mathcal{C}_{2} \cup \mathcal{C}_{3}\;,
\end{equation}
where, for $K > 0$:
\begin{equation}\label{eq:Cdef}
\begin{split}
\mathcal{C}_{1} &= \big\{ z\in \mathbb{C} \, \mid \, \text{Re}\, z = -K\;,\; \text{Im}\, z \in (-\pi/2\beta; \pi/2\beta) \big\} \\
\mathcal{C}_{2} &= \big\{ z\in \mathbb{C} \, \mid \, \text{Re}\, z \in (-K, \infty)\;,\; \text{Im}\, z = \pi / 2\beta \big\} \\
\mathcal{C}_{3} &= \big\{ z\in \mathbb{C} \, \mid \, \text{Re}\, z \in (-K, \infty)\;,\; \text{Im}\, z = -\pi / 2\beta \big\}\;.
\end{split}
\end{equation}
Observe that the path $\mathcal{C}$ does not enclose the poles of the Fermi-Dirac function $ z \mapsto 1/(1 + e^{\beta (z - \mu)})$, which are:
\begin{equation*}
z = \mu + i \frac{2\pi}{\beta} \Big(n + \frac{1}{2}\Big)\;,\qquad n \in \mathbb{Z}\;.
\end{equation*}
Thus, we obtain 
\begin{equation*}
\frac{1}{1 + e^{\beta (H^{L} - \mu)}} - \frac{1}{1 + e^{\beta (H - \mu)}} = \int_{\mathcal{C}} \frac{dz}{2\pi i}\, \frac{1}{1 + e^{\beta (z-\mu)}} \Big( \frac{1}{H^{L} - z} - \frac{1}{H - z} \Big)\;.
\end{equation*}
Let $z = x + i y$ with $y = \pi/2\beta$. Then, we have:
\begin{equation*}%\label{eq:normL0}
 \frac{1}{i(H^{L} - z)} - \frac{1}{i(H - z)} = \int_{-\infty}^{0} dt\, e^{t y} \Big( e^{i (H^{L} - x) t} - e^{i (H - x) t} \Big)\;.
\end{equation*}
A similar representation holds true for $y = -\pi/2\beta$. Hence, to prove the claim (\ref{eq:B7}) it is sufficient to show:
\begin{equation}\label{eq:limnorm}
\lim_{L \to \infty} \Big\| \Big( e^{i (H - x) t} - e^{i (H^{L} - x) t} \Big) f \Big\| = 0\;,
\end{equation}
where we abridged $\| \cdot \|\equiv \| \cdot \|_{L^{2}(\R^{3})}$. To this end, let us estimate:
\begin{equation}\label{eq:normL}
\begin{split}
\Big\| \Big( e^{i (H - x) t} - e^{i (H^{L} - x) t} \Big) f \Big\| &\leq \int_{0}^{t}ds\, \Big\| e^{i (H - x) (t-s)} (V_\text{ext} - V^{L}) e^{i (H^{L} - x) s} f \Big\| \\
&\leq  \int_{0}^{t}ds\, \Big\|  V_\text{ext} \mathbf{1}_{|\hat x| > L} e^{i (H^{L} - x) s} f \Big\| \\
&\leq \int_{0}^{t}ds\, \frac{C}{L} \Big\|  \langle \hat x \rangle^{3} e^{i (H^{L} - x) s} f \Big\|\;,
\end{split}
\end{equation}
where $\langle \hat x \rangle^{2} = 1 + |\hat x|^{2}$ and we used that, by assumption, $V_\text{ext} (x) \leq C|x|^{2}$, together with the estimate $\mathbf{1}_{|\hat x| > L} \leq | \hat x | / L$. Next, let us estimate the moment in the right-hand side. We have:
\begin{equation*}
\begin{split}
\partial_{s} \Big\|  \langle \hat x \rangle^{3} e^{i (H^{L} - x) s} f \Big\|^{2} &= \langle e^{i (H^{L} - x) s} f , [H_{0}, \langle \hat x \rangle^{6}] e^{i (H^{L} - x) s} f  \rangle \\
&\leq C\Big\|  \langle \hat x \rangle^{5/2} e^{i (H^{L} - x) s} f \Big\|^{2}\;.
\end{split}
\end{equation*}
Iterating the inequality five more times, we get:
\begin{equation*}
\begin{split}
\Big\|  \langle \hat x \rangle^{3} e^{i (H^{L} - x) s} f \Big\|^{2} &\leq C (1 + |s|^{6}) \Big\|  \langle \hat x \rangle^{3} f \Big\|^{2} \\
&=: K_{s,f} < \infty\;,
\end{split}
\end{equation*}
where we used that $f\in C^{\infty}_{\text{c}}$. Plugging this bound in (\ref{eq:normL}), we get:
\begin{equation*}
\Big\| \Big( e^{i (H - x) t} - e^{i (H^{L} - x) t} \Big) f \Big\| \leq \frac{C_{t,f}}{L}
\end{equation*}
for a suitable constant $C_{t,f}>0$. This implies (\ref{eq:limnorm}), and concludes the proof of Lemma~\ref{lemma: strong convergence}.
\end{proof}
\section{Proof of Proposition \ref{prop:localsc}}\label{app:gibbs}

In this section we will conclude the proof of Proposition \ref{prop:localsc}, showing that the Fermi-Dirac distribution (\ref{eq:FD}) satisfies Assumption \ref{ass:sc} in a domain $\Lambda \subset \bR^3$, if the Hamiltonian $H$ and the chemical potential $\mu$ fulfill Assumption \ref{ass:Weyl}.  In the following, $z_{1}$ will denote a generic point in the set $\Lambda$ (as specified by Assumption \ref{ass:Weyl}), for which the bound
\begin{equation*}
\Big\| \frac{1}{(\beta (H-\mu))^{2m} + 1} \mathcal{W}_{z_{1}}^{(n)} \Big\|_{\text{tr}} \leq C\varepsilon^{-2}
\end{equation*}
holds true for all $n,m \in \bN$ large enough. Also, we shall denote by $z_{2}$ a generic point in $\mathbb{R}^{3}$. In this section, to shorten the notation we will set $\omega \equiv \omega_{\mu,\beta} = (1 + e^{\beta (H - \mu)})^{-1}$. 
\subsection{Proof of Eq.~(\ref{eq:sc1})} 
We consider $\big\| \mathcal{W}_{z_{2}}^{(n/2)}  \omega  \mathcal{W}_{z_{1}}^{(n)} \big\|_{\text{tr}}$. We observe that, for $\theta =1,2,3$, 
\begin{equation}\label{eq:deccomm}
\begin{split}
\Big\| (z_{1,\theta} - z_{2,\theta}) \mathcal{W}_{z_{2}}^{(n/2)}  \omega \mathcal{W}_{z_{1}}^{(n)} \Big\|_{\text{tr}} &= \Big\| (z_{1,\theta} - \hat x_{\theta} + \hat x_{\theta} - z_{2,\theta}) \mathcal{W}_{z_{2}}^{(n/2)}  \omega \mathcal{W}_{z_{1}}^{(n)} \Big\|_{\text{tr}} \\
&\leq \Big\| (z_{1,\theta} - \hat x_{\theta}) \mathcal{W}_{z_{2}}^{(n/2)}  \omega \mathcal{W}_{z_{1}}^{(n)} \Big\|_{\text{tr}} + \Big\| \mathcal{W}_{z_{2}}^{(n/2)}  \omega (z_{2,\theta} - \hat x_{\theta}) \mathcal{W}_{z_{1}}^{(n)} \Big\|_{\text{tr}} \\
&\quad + \Big\|  \mathcal{W}_{z_{2}}^{(n/2)}  [\omega, \hat x_{\theta}] \mathcal{W}_{z_{1}}^{(n)} \Big\|_{\text{tr}}\;.
\end{split}
\end{equation}
For an integer ${p} \leq 4n$, we have
\begin{equation*}
\Big| \prod_{j=1}^{{p}} ( x_{\theta_{j}} - z_{\theta_j}) \mathcal{W}^{(n)}_{z}(x)\Big| \leq C \mathcal{W}^{(n-{p}/4)}_{z}(x)\;.
\end{equation*}
Let $p=2n-4$. Iterating (\ref{eq:deccomm}), we find:
\begin{equation}\label{eq:ad1} 
\Big\| \prod_{i=1}^{{2n-4}} (z_{1,\theta_{i}} - z_{2,\theta_{i}} ) \mathcal{W}_{z_{2}}^{(n/2)}  \omega \mathcal{W}_{z_{1}}^{(n)} \Big\|_{\text{tr}} \leq C_{n} \sum_{j=0}^{{2n}} \sum_{\alpha} \Big\| \mathcal{W}_{z_{2}}^{(1)} \text{ad}^{j;\alpha}_{\hat x}(\omega) \mathcal{W}_{z_{1}}^{(n/2)} \Big\|_{\text{tr}}\;.
\end{equation}
Here $\alpha = (\alpha_{1}, \ldots, \alpha_{j}) \in \{1,2,3 \}^j$ is a multi-index, and $\text{ad}^{j;\alpha}_{\hat x}(O)$ is the multi-commutator defined by 
\begin{equation*}
\text{ad}^{j;\alpha}_{\hat x}(\omega) = \big[ \cdots \big[\big[ \omega, \hat x_{\alpha_{1}}\big], \hat x_{\alpha_{2}} \big] \ldots, \hat x_{\alpha_{j}} \big]\;.
\end{equation*}
Therefore, 
\begin{equation}\label{eq:WoW}
\begin{split}
\Big\| \mathcal{W}_{z_{2}}^{(n/2)} \omega \mathcal{W}_{z_{1}}^{(n)} \Big\|_{\text{tr}} &\leq \frac{C_{n}}{1 + |z_{1} - z_{2}|^{{2n-4}}} \sum_{j=0}^{{2n}} \sum_{\alpha} \Big\| \mathcal{W}_{z_{2}}^{(1)} \text{ad}^{j;\alpha}_{\hat x}(\omega) \mathcal{W}_{z_{1}}^{(n/2)} \Big\|_{\text{tr}} \\
&\leq \frac{C_{n}}{1 + |z_{1} - z_{2}|^{{2n-4}}} \Big(\| \mathcal{W}_{z_{2}}^{(1)} \omega \mathcal{W}_{z_{1}}^{(n/2)} \|_\text{tr}  + \sum_{j=1}^{{2n}} \sum_{\alpha} \Big\| \text{ad}^{j;\alpha}_{\hat x}(\omega) \mathcal{W}_{z_{1}}^{(n/2)} \Big\|_{\text{tr}}\Big)\;.
\end{split}
\end{equation}
We can bound the first term in the parenthesis by Cauchy-Schwarz inequality as:
\[ \| \mathcal{W}_{z_{2}}^{(1)} \omega \mathcal{W}_{z_{1}}^{(n/2)} \|_\text{tr}  \leq \tr \,  \mathcal{W}_{z_{1}}^{(n)}  \omega  + \tr \,   \mathcal{W}_{z_{2}}^{(2)}  \omega  \leq C \veps^{-3}\;; \]
the last inequality follows from the fact that $\omega = \omega_{\mu,\beta}$ satisfies Assumption \ref{ass:bd}, as we proved in Appendix \ref{sec: appen-apriori}.

We are now left with bounding the contributions with $j>0$ on the r.h.s.~of (\ref{eq:WoW}). It is convenient to use the representation of the Fermi-Dirac distribution: 
\begin{equation}\label{eq:omegaint}
\omega = \int_{\mathcal{C}} \frac{dz}{2\pi i} \frac{1}{1 + e^{\beta (z - \mu)}} \frac{1}{H - z}\;,
\end{equation}
where $\mathcal{C}$ is an unbounded clockwise path that encloses the spectrum of $H$, see (\ref{eq:omegaint0}) for more details. 
In order to estimate contributions with $j > 0$ in (\ref{eq:WoW}), we have to consider multi-commutators of the resolvent of $H$ with $\hat x_{\alpha}$. Let $R_{H}(z) = (H-z)^{-1}$ be the resolvent of $H$ at $z$. To begin, we compute:
\begin{equation*}
\begin{split}
\Big[ R_{H}(z) , \hat x_{\alpha} \Big] &= -R_{H}(z) [ H, \hat x_{\alpha} ] R_{H}(z)\;, \\
\Big[ \Big[ R_{H}(z) , \hat x_{\alpha_{1}} \Big], \hat x_{\alpha_{2}} \Big] &= - \Big[ R_{H}(z) [ H, \hat x_{\alpha_{1}} ] R_{H}(z), \hat x_{\alpha_{2}}  \Big] \\
&= - \Big[ R_{H}(z), \hat x_{\alpha_{2}}\Big] [ H, \hat x_{\alpha_{1}} ] R_{H}(z) - R_{H}(z) [ [ H, \hat x_{\alpha_{1}} ], \hat x_{\alpha_{2}}] R_{H}(z)\\& \quad - R_{H}(z) [ H, \hat x_{\alpha_{1}} ] \Big[ R_{H}(z), \hat x_{\alpha_{2}}\Big]\;;
\end{split}
\end{equation*}
plugging the first relation in the second, we get:
\begin{equation*}
\begin{split}
\Big[ \Big[ R_{H}(z) , \hat x_{\alpha_{1}} \Big], \hat x_{\alpha_{2}} \Big]  &= R_{H}(z) [ H, \hat x_{\alpha_{2}} ] R_{H}(z) [ H, \hat x_{\alpha_{1}} ] R_{H}(z) \\
&\quad - R_{H}(z) [ [ H, \hat x_{\alpha_{1}} ], \hat x_{\alpha_{2}}] R_{H}(z) \\ 
&\quad + R_{H}(z) [ H, \hat x_{\alpha_{1}} ] R_{H}(z) [ H, \hat x_{\alpha_{2}} ] R_{H}(z)\;.
\end{split}
\end{equation*}
More generally, the $j$-fold multi-commutator $\text{ad}_{\hat x}^{j;\alpha}(R_{H}(z))$ is given by a linear combination of expressions having the form 
\begin{equation}\label{eq:RRR}
R_{H}(z) \text{ad}_{\hat x}^{\ell_{1}}(H) R_{H}(z) \text{ad}_{\hat x}^{\ell_{2}}(H) R_{H}(z) \cdots R_{H}(z) \text{ad}_{\hat x}^{\ell_{r}}(H) R_{H}(z) =: E_{r;\ell}\;,
\end{equation}
with $r\leq j$, $\text{ad}_{\hat x}^{\ell_{i}}(H) \equiv \text{ad}_{\hat x}^{\ell_{i}; \alpha(\ell_{i})}(H)$ with $\alpha(\ell_{i})$ a multi-index with $\ell_{i}$ components, and
\begin{equation}\label{eq:ellcond}
1 \leq \ell_{i} \leq j\;,\qquad \ell_{1} + \ell_{2} + \cdots + \ell_{r} = j\;.
\end{equation}
It is convenient to move resolvents to the left of the string. To this end, we rewrite:
\begin{equation*}
\begin{split}
E_{r;\ell} &= R_{H}(z)^{2} \text{ad}_{\hat x}^{\ell_{1}}(H) \text{ad}_{\hat x}^{\ell_{2}}(H) R_{H}(z) \cdots R_{H}(z) \text{ad}_{\hat x}^{\ell_{r}}(H) R_{H}(z) \\
&+ R_{H}(z)\Big[ R_{H}(z), \text{ad}_{\hat x}^{\ell_{1}}(H)\Big] \text{ad}_{\hat x}^{\ell_{2}}(H) R_{H}(z) \cdots R_{H}(z) \text{ad}_{\hat x}^{\ell_{r}}(H) R_{H}(z) \\
&= R_{H}(z)^{2} \text{ad}_{\hat x}^{\ell_{1}}(H) \text{ad}_{\hat x}^{\ell_{2}}(H) R_{H}(z) \cdots R_{H}(z) \text{ad}_{\hat x}^{\ell_{r}}(H) R_{H}(z) \\
&\quad + R_{H}(z)^{2} \text{ad}_{H,\hat x}^{1, \ell_{1}}(H) R_{H}(z) \text{ad}_{\hat x}^{\ell_{2}}(H) R_{H}(z) \cdots R_{H}(z) \text{ad}_{\hat x}^{\ell_{r}}(H) R_{H}(z)\;,
\end{split}
\end{equation*}
where:
\begin{equation*}
\text{ad}_{H,\hat x}^{a,\ell}(H) := \text{ad}_{H}^{a} \Big( \text{ad}_{\hat x}^{\ell}(H) \Big)\quad \text{for $a\in \mathbb{N}$.}
\end{equation*}
This procedure can be iterated; we can bring resolvents to the left using repeatedly the identity
\begin{equation}\label{eq:commid}
O R_{H}(z) = R_{H}(z) O + R_{H}(z) [ H, O ] R_{H}(z)\;.
\end{equation}
In doing so, we produce new terms involving higher commutators with $H$. We ultimately get that $E_{r;\ell}$ can be written as:
\begin{equation}\label{eq:EaEb}
E_{r;\ell} = E^{(\text{a})}_{r;\ell} + E^{\text{(b)}}_{r;\ell}
\end{equation}
where $E^{(\text{a})}_{r;\ell}$ is given by the sum of terms having the form
\begin{equation}\label{eq:Ea}
R_{H}(z)^{q+1} \text{ad}_{H,\hat x}^{a_{1}, \ell_{1}}(H) \cdots \text{ad}_{H,\hat x}^{a_{r}, \ell_{r}}(H)\;,
\end{equation}
with $q\geq r$. Observe that, in (\ref{eq:RRR}), the number of resolvents is equal to $r+1$. The number of new resolvents, produced by the identity (\ref{eq:commid}), is equal to the number of commutators with $H$; that is,
\begin{equation}\label{eq:qra}
q - r = \sum_{i=1}^{r} a_{i}\;.
\end{equation}
On the other hand, the term $E^{\text{(b)}}_{r;\ell}$ collects contributions with at least one resolvent not in the leftmost place. It consists of terms having the form
\begin{equation}\label{eq:Eb}
R_{H}(z)^{q_{0}+1} \text{ad}_{H,\hat x}^{a_1,\ell_1}(H) R_H (z)^{q_1} \text{ad}_{H,\hat x}^{a_2,\ell_2}(H) R_H (z)^{q_2} \dots \dots \text{ad}_{H,\hat x}^{a_r,\ell_r}(H) R_H (z)^{q_r} 
\end{equation}
where $q_i \in \{0,1\}$ for all $i = 1, \dots , r$. 
%, $\sum_{i=1}^{s} q_{i} \geq 1$, while $A_{i}$ has the form
%%
%\begin{equation}
%A_{i} = \text{ad}_{H,\hat x}^{a_{i_{1}}, \ell_{i_{1}}}(H) \cdots \text{ad}_{H,\hat x}^{a_{i_{n_{i}}}, \ell_{i_{n_{i}}}}(H)\;.
%\end{equation}
%%
As in (\ref{eq:Ea}), the number of new resolvents produced with respect to (\ref{eq:RRR}) equals the number of commutators with respect to $H$:
\begin{equation}\label{eq:qcond}
\sum_{i=0}^{r} q_{i} - r = \sum_{i=1}^{r} a_{i}\;.
\end{equation}
At this point, a few remarks are in order.
\begin{remark}\label{rem:qja}
\begin{itemize}
\item[(i)] Observe that, in (\ref{eq:Eb}), we can assume that $q_{0}$ is as large as we wish, up to increasing the number of commutators with respect to $H$, according to Eq.~(\ref{eq:qcond}). We will choose $q_0$ so large that we can apply Assumption \ref{ass:Weyl} to control the trace norm of $\mathcal{W}_{z_1}^{(n/2)} R_H (z)^{q_0+1}$. All the other resolvents in (\ref{eq:Eb}) will be estimated using the nonzero imaginary part of $z$. 
\item[(ii)] In (\ref{eq:Ea}), on the other hand, we can only assume that $0< q < q_0$ (if $q = q_0$, this term can be included in $E_{r,\ell}^{(b)}$). For these contributions, the $z$ integral will be performed explicitly, and the result will have good decay in energy. From (\ref{eq:qra}) and (\ref{eq:qcond}),  we conclude that, for all contributions, $\sum_{i=1}^r a_i \leq q_0$.
%\item[(iii)] From (\ref{eq:ellcond}), (\ref{eq:qcond}), the following relation holds between the total number of commutators with $H, \hat x$, and the total number of resolvents (recall that $r\leq j$):
%
%\begin{equation}\label{eq:qja}
%\sum_{i = 1}^{r} (\ell_{i} + a_{i}) - \sum_{i=0}^{s} q_{i} = j-r \geq 0\;.
%\end{equation}
%
%Since every commutator, both with respect to $H$ and to $\hat x$, introduces a factor $\varepsilon$, and every resolvent introduces a factor $\beta = O(1/\varepsilon)$, this relation will be important to derive an estimate for the strings of operators contributing to $E_{r;\ell}$.
\end{itemize}
\end{remark}
Let us denote by $\mathcal{A}$ the general contribution to $E^{(\text{a})}_{r;\ell}$, which has the form (\ref{eq:Ea}), and by $\mathcal{B}$ the general contribution to $E^{(\text{b})}_{r;\ell}$, which has the form (\ref{eq:Eb}). We shall discuss them separately. 
\paragraph{Terms of type $\mathcal{A}$.} Plugging $\mathcal{A}$ in the integral defining $\omega$, Eq.~(\ref{eq:omegaint}), we are left with estimating:
\begin{equation}\label{eq:A28}
\Big\| \Big( \int_{\mathcal{C}} \frac{dz}{2\pi i} \frac{1}{1 + e^{\beta (z - \mu)}} \mathcal{A}\Big) \mathcal{W}_{z_{1}}^{(n/2)}  \Big\|_{\text{tr}}\;.
\end{equation}
To bound these terms, we use that the $z$-integral can be performed explicitly. In fact:
\begin{equation}\label{eq:A29}
\begin{split}
 \int_{\mathcal{C}} \frac{dz}{2\pi i} \frac{1}{1 + e^{\beta (z - \mu)}} \mathcal{A} &=  \int_{\mathcal{C}} \frac{dz}{2\pi i} \frac{1}{1 + e^{\beta (z - \mu)}} \frac{1}{(H-z)^{q+1}} \text{ad}_{H,\hat x}^{a_{1}, \ell_{1}}(H) \cdots \text{ad}_{H,\hat x}^{a_{r}, \ell_{r}}(H) \\
 &= \frac{1}{q!} \big( \partial^{q}_{\mu} \omega\big) \text{ad}_{H,\hat x}^{a_{1}, \ell_{1}}(H) \cdots \text{ad}_{H,\hat x}^{a_{r}, \ell_{r}}(H)\;.
 \end{split}
\end{equation}
Plugging (\ref{eq:A29}) into (\ref{eq:A28}), we have:
\begin{equation}\label{eq:trop}
\begin{split}
&\Big\|  \big( \partial^{q}_{\mu} \omega\big) \text{ad}_{H,\hat x}^{a_{1}, \ell_{1}}(H) \cdots \text{ad}_{H,\hat x}^{a_{r}, \ell_{r}}(H) \mathcal{W}_{z_{1}}^{(n/2)}   \Big\|_{\text{tr}} \\
&\quad \leq \Big\| \big( \partial^{q}_{\mu} \omega\big)  \mathcal{W}_{z_{1}}^{(n/4)} \Big\|_{\text{tr}} \Big\|  \big(\mathcal{W}_{z_{1}}^{(n/4)}\big)^{-1}\text{ad}_{H,\hat x}^{a_{1}, \ell_{1}}(H) \cdots \text{ad}_{H,\hat x}^{a_{r}, \ell_{r}}(H) \mathcal{W}_{z_{1}}^{(n/2)}   \Big\|_{\text{op}}\;.
\end{split}
\end{equation}
We further estimate the operator norm as:
\begin{equation}\label{eq:trop2}
\begin{split}
&\Big\| \big(\mathcal{W}_{z_{1}}^{(n/4)}\big)^{-1} \text{ad}_{H,\hat x}^{a_{1}, \ell_{1}}(H) \cdots \text{ad}_{H,\hat x}^{a_{r}, \ell_{r}}(H) \mathcal{W}_{z_{1}}^{(n/2)}   \Big\|_{\text{op}} \\
&\quad \leq C\Big\|\big(\mathcal{W}_{z_{1}}^{(n/4)}\big)^{-1} \text{ad}_{H,\hat x}^{a_{1}, \ell_{1}}(H) \mathcal{W}_{z_{1}}^{(n/4 + a_{1})}   \Big\|_{\text{op}} \Big\| \big(\mathcal{W}_{z_{1}}^{(n/4 + a_{1})}\big)^{-1} \text{ad}_{H,\hat x}^{a_{1}, \ell_{1}}(H) \mathcal{W}_{z_{1}}^{(n/4 + a_{1} + a_{2})}   \Big\|_{\text{op}} \cdot \\
&\qquad \cdots \cdot \Big\| \big(\mathcal{W}_{z_{1}}^{(n/4 + \sum_{i=1}^{r-1} a_{i})}\big)^{-1} \text{ad}_{H,\hat x}^{a_{r}, \ell_{r}}(H) \mathcal{W}_{z_{1}}^{(n/4 + \sum_{i=1}^{r} a_{i})}   \Big\|_{\text{op}}\;,
\end{split}
\end{equation}
where we choose $n/4 \geq q_{0}$, which guarantees that $n/4 \geq \sum_{i=1}^{r} a_{i}$ and implies that $\mathcal{W}_{z_{1}}^{(n/2)}\leq C\mathcal{W}_{z_{1}}^{(n/4 + \sum_{i=1}^{r} a_{i})}$. Proceeding similarly to (\ref{eq:ad1}), it is enough to control the quantities
\begin{equation*}
\| \text{ad}^\gamma_{\hat{x}} \big( \text{ad}_{H,\hat{x}}^{a, \ell} (H) \big) \mathcal{W}_{z_1}^{(a)} \|_\text{op}\;,
\end{equation*}
for all $\gamma \leq n + 4q_{0}$. We claim that 
\begin{equation}\label{eq:normad} 
\| \text{ad}^\gamma_{\hat{x}} \big( \text{ad}_{H,\hat{x}}^{a, \ell} (H) \big) \mathcal{W}_{z_1}^{(a)} \|_\text{op} \leq C \varepsilon^{\ell+a}
\end{equation} 
for a constant $C>0$ depending only on $n$ and for all $\ell, a \leq 2n$ (these conditions hold true, since we assumed $\ell \leq j \leq m \leq 2n$ and $a \leq q_0 \leq 2n$). 

To prove (\ref{eq:normad}), we first observe that 
\begin{equation*}
\text{ad}_{\hat x}^{\ell}(H) = \text{ad}_{\hat x}^{\ell}(H_{0})
\end{equation*}
with $H_{0} = \sqrt{1 - \varepsilon^{2}\Delta}$. Since the commutators with $\hat x_{\alpha}$ act as derivatives with respect to $\nabla_{\alpha}$, we have:
\begin{equation}\label{eq:multiH0}
\text{ad}_{\hat x}^{\ell}(H_{0}) = \varepsilon^{\ell} \sum_{k,\alpha} C_{k,\alpha} \frac{1}{(1 - \varepsilon^{2} \Delta)^{\frac{k}{2}}} (\varepsilon D)^{\alpha}\;,
\end{equation}
for suitable coefficients $C_{k,\alpha}$, and where the sum is restricted to $|\alpha| \leq k$, and to odd $k$ such that $1\leq k \leq \ell$. In \eqref{eq:multiH0}, we used the convenient notation $D^{\alpha}\equiv \partial_{x_{1}}^{\alpha_{1}}\partial_{x_{2}}^{\alpha_{2}}\partial_{x_{3}}^{\alpha_{3}}$. It is convenient to use the integral representation:
\begin{equation*}
\frac{1}{(1 - \varepsilon^{2} \Delta)^{\frac{k}{2}}} = \frac{1}{\pi} \int_{0}^{\infty} \frac{d\lambda}{\sqrt{\lambda}} \frac{1}{\lambda + (1 - \varepsilon^{2}\Delta)^{k}}\;,
\end{equation*}
to write:
\begin{equation}\label{eq:res}
\text{ad}_{\hat x}^{\ell}(H_{0}) = \varepsilon^{\ell} \sum_{k,\alpha} \int_{0}^{\infty} \frac{d\lambda}{ \sqrt{\lambda}} \frac{C_{k,\alpha}}{\pi} P_{\alpha,k}(\lambda)\;,
\end{equation}
where we defined, for any $\alpha$ with $|\beta| \leq 2k$,
\begin{equation}\label{eq:Pk}
P_{\alpha,k}(\lambda) := \frac{(\varepsilon D)^{\alpha}}{\lambda + (1 - \varepsilon^{2}\Delta)^{k}}\;;
\end{equation}
Let us now consider $\text{ad}^{a,\ell}_{H,x} (H) = \text{ad}^a_H (\text{ad}_x^\ell (H))$. The first commutator with $H = H_0 +V_\text{ext}$ only involves the potential $V_\text{ext}$ (because $H_0$ commutes with (\ref{eq:res})).  To compute it, we observe that 
\begin{equation*}
\begin{split}
[ V_\text{ext} ,  P_{\alpha,k}(\lambda)] = \frac{1}{\lambda + (1 - \varepsilon^{2}\Delta)^{k}} [ V_\text{ext} , (\varepsilon D)^{\alpha}] + \Big[ V_\text{ext} ,  \frac{1}{\lambda + (1 - \varepsilon^{2}\Delta)^{k}}\Big] (\varepsilon D)^{\alpha}\;.
\end{split}
\end{equation*}
Hence, $[ V_\text{ext} ,  P_{\alpha,k}(\lambda)]$ can be expressed as a linear combination of terms having either the form 
 \begin{equation}\label{eq:1com}
\frac{(\varepsilon D)^{\alpha_{1}}}{\lambda + (1 - \varepsilon^{2}\Delta)^{k}} (( \varepsilon D )^{\alpha_{2}} V_\text{ext}) = P_{\alpha_1,k} (\lambda) (( \varepsilon D )^{\alpha_{2}} V_\text{ext}) 
\end{equation}
with $|\alpha_1| + |\alpha_2|  = |\alpha| \leq k$ and $|\alpha_2| \geq 1$, or the form  
\begin{equation}\label{eq:com2}
\frac{(\varepsilon D)^{\alpha_{1}}}{\lambda + (1 - \varepsilon^{2}\Delta)^{k}} (( \varepsilon D )^{\alpha_{2}} V_\text{ext})\frac{(\varepsilon D)^{\alpha}}{\lambda + (1 - \varepsilon^{2}\Delta)^{k}} = P_{\alpha_1, k} (\lambda) (( \varepsilon D )^{\alpha_{2}} V_\text{ext}) P_{\alpha,k} (\lambda) 
\end{equation}
with $|\alpha_1| + |\alpha_2| = 2k$ and $|\alpha_2| \geq 1$. Further commutators with $V_\text{ext}$ are given again by alternating products of derivatives $(( \varepsilon D )^{\beta} V_\text{ext})$, with $|\beta| \geq 1$, and of factors $P_{\delta,k}(\lambda)$. All these factors have an index $|\delta| \leq 2k-1$ (which implies that $\| P_{\delta,k} (\lambda) \|_\text{op} \lesssim (1+\lambda)^{-1/2k})$; one of them has $|\delta| \leq k$ (which implies that $\| P_{\delta,k} (\lambda) \|_\text{op} \lesssim (1+\lambda)^{-1/2}$). Commutators with $H_0$, on the other hand, only affect the multiplication operators $(\varepsilon D)^\beta V_\text{ext}$. 
Under the assumption that all derivatives of $V_\text{ext}$ of order two or higher are bounded, it is easy to check that the commutators $[H_0, (( \varepsilon D )^{\beta} V_\text{ext})]$ are bounded, and have bounded multiple commutators with $V_\text{ext}$, $H_0$ and $\hat{x}$. 

We conclude that $\text{ad}^\gamma_{\hat{x}} (\text{ad}^{a,\ell}_{H,\hat{x}} (H))$ is given by a linear combination of terms having the form 
\[ \varepsilon^{\kappa} \int_0^\infty \frac{d\lambda}{\sqrt{\lambda}} \, P_{\alpha_1, k} (\lambda) \widetilde{K}_1 P_{\alpha_2, k} (\lambda) \widetilde{K}_2 \dots  P_{\alpha_{\tilde{s}},k} (\lambda) \widetilde{K}_{\tilde{s}} \]
with some $\tilde{s} \in \bN$, an exponent $\widetilde{\kappa} \geq \ell+ a + \gamma$ (each commutator with $H_0$, $V_\text{ext}$ or $\hat{x}$ produces at least one new factor $\varepsilon$; this corresponds for example to the restriction $|\alpha_2| \geq 1$, in (\ref{eq:1com}), (\ref{eq:com2})) and where each $\widetilde{K}_j$ is either a bounded operator, with norm independent of $\varepsilon$, or a component of $\nabla V_\text{ext}$.  To control the unbounded contributions, we move all factors $\nabla V_\text{ext}$ to the right (so that we can bound them with the function $\mathcal{W}_{z_1}^{(a)}$ in (\ref{eq:normad})). Since every further commutator produces only bounded terms, we conclude that $\text{ad}^\gamma_{\hat{x}} (\text{ad}^{a,\ell}_{H,\hat{x}} (H)$ can be expressed as a linear combination of terms like 
\[ \varepsilon^{\kappa} \int_0^\infty \frac{d\lambda}{\sqrt{\lambda}} \, P_{\alpha_1, k} (\lambda) K_1 P_{\alpha_2, k} (\lambda) K_2 \dots  P_{\alpha_s,k} (\lambda) K_s \prod_{i=1}^t  (\nabla V_\text{ext})_{\delta_i}   \]
where $s \in \bN$, $\kappa \geq \ell+a + \gamma$, $t \leq a$ (there are $a$ commutators with $H$ and thus at most $a$ factors $\nabla V_\text{ext}$), where each operator $K_i$ is bounded, with norm independent of $\varepsilon$, and where $\prod_{i=1}^s \| P_{\alpha_i , k} (\lambda) \|_\text{op} \leq (1+\lambda)^{-1/2-s/2k}$ (by the remarks below (\ref{eq:com2})). 

Using 
\begin{equation*}
|\nabla V_\text{ext} (\hat x)| = |\nabla V_\text{ext} (\hat x - z_{1} + z_{1})| \leq |\nabla V_\text{ext} (z_{1})| + C(|\hat x - z_{1}|)
\end{equation*}
and recalling that $z_1 \in \Lambda$ and the assumption that $\nabla V_\text{ext}$ is bounded on $\Lambda$, we obtain 
\begin{equation*}
 |\nabla V_\text{ext} (\hat x)|  \mathcal{W}_{z_{1}}^{(1)}(\hat{x}) \leq C\;.
\end{equation*}
Thus, after integrating over $\lambda$ (using $\prod_{i=1}^s \| P_{\alpha_i , k} (\lambda) \|_\text{op} \leq (1+\lambda)^{-1/2-s/2k}$), we obtain  
\begin{equation*}
\| \text{ad}^\gamma_{\hat{x}} \big( \text{ad}_{H,\hat{x}}^{a, \ell} (H) \big) \mathcal{W}_{z_2}^{(a)} \|_\text{op} \leq C \varepsilon^{\ell + a}
\end{equation*}
which completes the proof of (\ref{eq:normad}). 

From (\ref{eq:trop}), (\ref{eq:trop2}), (\ref{eq:normad}) we conclude that
\begin{equation}\label{eq:Afin} \begin{split} 
\Big\| \big( \partial^{q}_{\mu} \omega\big) \text{ad}_{H,\hat x}^{a_{1}, \ell_{1}}(H) \cdots \text{ad}_{H,\hat x}^{a_{r}, \ell_{r}}(H) \mathcal{W}_{z_{1}}^{(n/2)}   \Big\|_{\text{tr}} \leq C \Big\| \big( \partial^{q}_{\mu} \omega\big) \mathcal{W}_{z_{1}}^{(n/4)} \Big\|_{\text{tr}}  \varepsilon^{\sum_{i=1}^r (\ell_i + a_i)}\;.
\end{split}
\end{equation} 
Let us now estimate the trace norm in the right-hand side of (\ref{eq:Afin}). Observe that:
\begin{equation*}
\partial_{\mu} \omega = \frac{\beta e^{\beta (H - \mu)}}{(1 + e^{\beta (H-\mu)})^{2}} = \beta \omega (1 - \omega)\;;
\end{equation*}
iterating and using $0\leq \omega \leq 1$, we find that:
\begin{equation}\label{eq:derest}
\pm \partial^{q}_{\mu} \omega \leq C_{q} \beta^{q} \omega (1 - \omega)\;,
\end{equation}
for an appropriate constant $C_{q} > 0$. Recalling that $\beta = O(\varepsilon^{-1})$ and that $\sum_{i=1}^r \ell_i = j \geq r$, we get, with (\ref{eq:qra}):
\begin{equation}\label{eq:Afin2}
\begin{split}
&\Big\| \big( \partial^{q}_{\mu} \omega\big) \text{ad}_{H,\hat x}^{a_{1}, \ell_{1}}(H) \cdots \text{ad}_{H,\hat x}^{a_{r}, \ell_{r}}(H) \mathcal{W}_{z_{1}}^{(n/2)}   \Big\|_{\text{tr}} \\
&\quad \leq C \varepsilon^{j + \sum_{i} a_{i} - q}  \Big\| \omega (1 - \omega) \mathcal{W}_{z_{1}}^{(n/4)} \Big\|_{\text{tr}} \\
&\quad \leq C \Big\| \omega (1 - \omega)\mathcal{W}_{z_{1}}^{(n/4)} \Big\|_{\text{tr}}\;.
\end{split}
\end{equation}
To estimate the trace norm, we will use Assumption \ref{ass:Weyl}. To do this, we shall use \eqref{eq:oocbd}, so that,
%%
%\begin{equation}\label{eq:oocbd}
%\begin{split}
%\omega (1 - \omega) &= \frac{e^{\beta (H-\mu)}}{(1 + e^{\beta (H-\mu)})^{2}}  \\
%&\leq \frac{C_{m}}{(\beta (H-\mu))^{2m} +1}\;,\qquad \text{for all $m\in \mathbb{N}$;}
%\end{split}
%\end{equation}
%%
%therefore, 
from Assumption \ref{ass:Weyl}:
\begin{equation}\label{eq:wwc}
\Big\| \omega (1 - \omega)\mathcal{W}_{z_{1}}^{(n/4)} \Big\|_{\text{tr}} \leq C\varepsilon^{-2}\;.
\end{equation}
All in all, from (\ref{eq:trop}), (\ref{eq:Afin}), (\ref{eq:Afin2}), (\ref{eq:wwc}), we get:
\begin{equation}\label{eq:Aestfin}
\Big\| \Big( \int_{\mathcal{C}} \frac{dz}{2\pi i} \frac{1}{1 + e^{\beta (z - \mu)}} \mathcal{A}\Big) \mathcal{W}_{z_{2}}^{(n/2)}   \Big\|_{\text{tr}} \leq C\varepsilon^{-2}\;.
\end{equation}
This concludes the analysis of all terms contributing to $E^{(\text{a})}_{r;\ell}$ in (\ref{eq:EaEb}).

\paragraph{Terms of type $\mathcal{B}$.} Let us now discuss the contribution due to the terms of the form (\ref{eq:Eb}). Here, we cannot compute explicitly the $z$ integral. In order to estimate these terms, we choose $q_{0} \geq 2 m$ in (\ref{eq:Eb}), with $m \in \bN$ as defined in Assumption \ref{ass:Weyl}. We start by writing:
\begin{equation}\label{eq:BestC}
\begin{split}
&\Big\|  \Big(\int_{\mathcal{C}} \frac{dz}{2\pi i}\, \frac{1}{1 + e^{\beta (z - \mu)}} \mathcal{B}\Big) \mathcal{W}_{z_{1}}^{(n/2)} \Big\|_{\text{tr}} \\
&\quad \leq \Big\| \Big(\int_{\mathcal{C}_{1}} \frac{dz}{2\pi i}\, \frac{1}{1 + e^{\beta (z - \mu)}} \mathcal{B}\Big) \mathcal{W}_{z_{1}}^{(n/2)} \Big\|_{\text{tr}} + \Big\|  \Big(\int_{\mathcal{C}_{2} \cup \mathcal{C}_{3}} \frac{dz}{2\pi i}\, \frac{1}{1 + e^{\beta (z - \mu)}} \mathcal{B}\Big) \mathcal{W}_{z_{1}}^{(n/2)} \Big\|_{\text{tr}}
\end{split}
\end{equation}
recall the definition of the path $\mathcal{C}$ in Eqs.~(\ref{eq:Cdef0}), (\ref{eq:Cdef}). Using the arbitrariness of $K$ in the definition of $\mathcal{C}$, recall (\ref{eq:Cdef}), the first term on the right-hand side of (\ref{eq:BestC}) can be made arbitrarily small by taking $K$ large enough. Let us now consider the contribution of the terms associated with $\mathcal{C}_{2}$ and $\mathcal{C}_{3}$. Up to errors that can be made arbitrarily small in $K$, we can estimate the second term on the right-hand side of (\ref{eq:BestC}) by:
\begin{equation}\label{eq:BR}
\begin{split}
&\Big\| \Big(\int_{\mathbb{R}} dy\,  \frac{1}{1 + i e^{\beta (y - \mu)}} \mathcal{B}(y + i \pi/2\beta)\Big) \mathcal{W}_{z_{1}}^{(n/2)} \Big\|_{\text{tr}}\\
&\quad + \Big\| \Big(\int_{\mathbb{R}} dy\,  \frac{1}{1 - i e^{\beta (y - \mu)}} \mathcal{B}(y - i \pi/2\beta)\Big) \mathcal{W}_{z_{1}}^{(n/2)} \Big\|_{\text{tr}}\;,
\end{split}
\end{equation}
where we made explicit the $z$-dependence $\mathcal{B} \equiv \mathcal{B}(z)$ due to the resolvents $R_{H}(z)$, recall (\ref{eq:Eb}). Let us consider the first term in \eqref{eq:BR}; the second can be estimated in the same way. Recalling that $\mathcal{B}$ is of the form (\ref{eq:Eb}), it is convenient to write:
\begin{equation*}
\mathcal{B}(z) = R_{H}(z)^{q_{0} + 1} \widetilde{\mathcal{B}}(z)\;,
\end{equation*}
that is, we make explicit the dependence on the leftmost resolvents. We write:
\begin{equation}\label{eq:ABC}
\begin{split}
&\int_{\mathbb{R}} dy\, \frac{1}{1 + i e^{\beta (y - \mu)}} R_{H}(y + i \pi/2\beta)^{q_{0} + 1} \widetilde{\mathcal{B}}(y + i \pi/2\beta) \\
&\quad = \int_{-\infty}^{\mu} dy\, \frac{1}{1 + i e^{\beta (y - \mu)}} R_{H}(y + i \pi/2\beta)^{q_{0} + 1} \widetilde{\mathcal{B}}(y + i \pi/2\beta)\\
&\qquad + \int_{\mu}^{\infty} dy\, \frac{1}{1 + i e^{\beta (y - \mu)}} R_{H}(y + i \pi/2\beta)^{q_{0} + 1} \widetilde{\mathcal{B}}(y + i \pi/2\beta)\\
& \quad = \chi(H\leq \mu)  \int_{-\infty}^{\mu} dy\, \frac{1}{1 + i e^{\beta (y - \mu)}} R_{H}(y + i \pi/2\beta)^{q_{0} + 1} \widetilde{\mathcal{B}}(y + i \pi/2\beta) \\
&\qquad + \chi(H > \mu)  \int_{-\infty}^{\mu} dy\, \frac{1}{1 + i e^{\beta (y - \mu)}} R_{H}(y + i \pi/2\beta)^{q_{0} + 1} \widetilde{\mathcal{B}}(y + i \pi/2\beta) \\
&\qquad + \int_{\mu}^{\infty} dy\, \frac{1}{1 + i e^{\beta (y - \mu)}} R_{H}(y + i \pi/2\beta)^{q_{0} + 1} \widetilde{\mathcal{B}}(y + i \pi/2\beta) =: \text{A} + \text{B} + \text{C}\;.
\end{split}
\end{equation}
Consider the term $\text{A}$. We have:
\begin{equation*}
\begin{split}
\text{A} &= \chi(H \leq \mu)  \int_{-\infty}^{\mu} dy\, \frac{1}{1 + i e^{\beta (y - \mu)}} R_{H}(y + i \pi/2\beta)^{q_{0} + 1} \widetilde{\mathcal{B}}(y + i \pi/2\beta) \\
&= \chi(H \leq \mu)  \int_{-\infty}^{\mu} dy\, R_{H}(y + i \pi/2\beta)^{q_{0} + 1} \widetilde{\mathcal{B}}(y + i \pi/2\beta) \\
&\quad - \chi(H \leq \mu)  \int_{-\infty}^{\mu} dy\, \frac{i e^{\beta (y - \mu)}}{1 + i e^{\beta (y - \mu)}} R_{H}(y + i \pi/2\beta)^{q_{0} + 1} \widetilde{\mathcal{B}}(y + i \pi/2\beta)\;,
\end{split}
\end{equation*}
which we further rewrite as:
\begin{equation}\label{eq:A1A2A3}
\begin{split}
\text{A} &= \chi(H \leq \mu)  \int_{-\infty}^{\infty} dy\, R_{H}(y + i \pi/2\beta)^{q_{0} + 1} \widetilde{\mathcal{B}}(y + i \pi/2\beta) \\
&\quad - \chi(H \leq \mu)  \int_{-\infty}^{\mu} dy\, \frac{i e^{\beta (y - \mu)}}{1 + i e^{\beta (y - \mu)}} R_{H}(y + i \pi/2\beta)^{q_{0} + 1} \widetilde{\mathcal{B}}(y + i \pi/2\beta) \\
&\quad - \chi(H \leq \mu)  \int_{\mu}^{\infty} dy\, R_{H}(y + i \pi/2\beta)^{q_{0} + 1} \widetilde{\mathcal{B}}(y + i \pi/2\beta) =: \text{A}_{1} + \text{A}_{2} + \text{A}_{3}\;.
\end{split}
\end{equation}
Consider $\text{A}_{1}$. Observe that the integrand is analytic in $y$ in the upper half complex plane, it is absolutely integrable in $y$, and it vanishes as $|y|\to \infty$. Thus, by Cauchy formula, one immediately gets that $\text{A}_{1} = 0$. Consider now $\text{A}_{2}$. We are interested in estimating:
\begin{equation}\label{eq:factors}
\begin{split}
&\Big\| \Big( \chi(H \leq \mu) R_{H}(y + i \pi/2\beta)^{q_{0} + 1} \widetilde{\mathcal{B}}(y + i \pi/2\beta)\Big) \mathcal{W}_{z_{1}}^{(n/2)} \Big\|_{\text{tr}} \\
&\quad \leq \Big\| \chi(H \leq \mu) R_{H}(y + i \pi/2\beta)^{q_{0} + 1}  \mathcal{W}_{z_{1}}^{(n/4)} \Big\|_{\text{tr}} \Big\|   \mathcal{W}_{z_{1}}^{(n/4)-1} \widetilde{\mathcal{B}}(y + i \pi/2\beta) \mathcal{W}_{z_{1}}^{(n/2)} \Big\|_{\text{op}}\;.
\end{split}
\end{equation}
The first factor in the right hand side can be estimated as:
\begin{equation}\label{eq:fact1}
\begin{split}
&\Big\|  \chi(H \leq \mu) \frac{1}{(H - y - i \pi/2\beta)^{q_{0} + 1}}  \mathcal{W}_{z_{1}}^{(n/4)} \Big\|_{\text{tr}} \\
&\quad \leq C\beta^{q_{0} + 1} \Big\| \chi(H \leq \mu) \frac{1}{(\beta(H - y))^{q_{0} + 1} + 1}  \mathcal{W}_{z_{1}}^{(n/4)}\Big\|_{\text{tr}}
\end{split}
\end{equation}
by choosing $q_{0} + 1$ even (to avoid introducing absolute values in the denominator). Let us now consider the second factor in (\ref{eq:factors}). Here, we proceed similarly as we did to bound (\ref{eq:trop2}). The only difference is that, according to (\ref{eq:Eb}), the operator $\widetilde{\mathcal{B}}(y + i \pi / 2\beta)$ also contains resolvents $R_H (y+i\pi / 2\beta)$. To handle these factors, however, it is enough to observe that $\| R_H ( y+i\pi / 2\beta) \|_\text{op} \leq C \beta$ and that, for any $v \in \bN$, $v \leq n/4$,   
\[ \| \mathcal{W}^{(v)-1}_{z_2}  R_H (y+i \pi/2\beta) \mathcal{W}^{(v)}_{z_2} \|_\text{op} \leq C \sum_{j=0}^v \| \text{ad}^j_{\hat{x}} (R_H (y+i \pi/2\beta)) \|_\text{op} \leq C \sum_{j=0}^v \beta^{j+1} \varepsilon^j \leq C \varepsilon^{-1} \]
because every commutator with $\hat{x}$ produces a new resolvent (whose norm is proportional to $\beta$) and a new commutator (providing an additional $\varepsilon$). Thus, we get:
\begin{equation}\label{eq:fact2}
\Big\| \mathcal{W}_{z_{1}}^{(n/4)-1} \widetilde{\mathcal{B}}(y + i \pi/2\beta) \mathcal{W}_{z_{1}}^{(n/2)} \Big\|_{\text{op}} \leq C \varepsilon^{j + \sum_{i=1}^{r} a_{i} - \sum_{i=1}^{s} q_{i}}\;,
\end{equation}
where the $-q_{i}$ factors at the exponent come from the resolvents $R_{H}(z)^{q_{i}}$. Plugging the bounds (\ref{eq:fact1}), (\ref{eq:fact2}) in (\ref{eq:factors}), we find:
\begin{equation*}
\begin{split}
&\Big\| \Big( \chi(H \leq \mu) R_{H}(y + i \pi/2\beta)^{q_{0} + 1} \widetilde{\mathcal{B}}(y + i \pi/2\beta)\Big) \mathcal{W}_{z_{1}}^{(n/2)} \Big\|_{\text{tr}} \\
&\leq C \varepsilon^{- 1 + j + \sum_{i=1}^{r} a_{i} - \sum_{i=0}^{s} q_{i}} \Big\|  \chi(H \leq \mu) \frac{1}{(\beta(H - y))^{q_{0} + 1} + 1}  \mathcal{W}_{z_{1}}^{(n/4)}\Big\|_{\text{tr}}\;;
\end{split}
\end{equation*}
using that $j + \sum_{i = 1}^{r} a_{i} -\sum_{i=0}^{s} q_{i} \geq 0$ (recall (\ref{eq:qcond}) and the condition $j \geq r$) and that $\sum_{i=1}^r a_i \leq q_0$, we get
\begin{equation*}
\begin{split}
&\Big\|  \Big( \chi(H \leq \mu) R_{H}(y + i \pi/2\beta)^{q_{0} + 1} \widetilde{\mathcal{B}}(y + i \pi/2\beta)\Big) \mathcal{W}_{z_{1}}^{(n/2)} \Big\|_{\text{tr}} \\
&\quad \leq C\beta \Big\| \chi(H \leq \mu)\frac{1}{(\beta(H - y))^{q_{0} + 1} + 1} \mathcal{W}_{z_{1}}^{(n/4)} \Big\|_{\text{tr}}\;.
\end{split}
\end{equation*}
We are now ready to estimate $\text{A}_{2}$. We have:
\begin{equation}\label{eq:A2est}
\begin{split}
&\Big\|\text{A}_{2} \mathcal{W}_{z_{1}}^{(n/2)}  \Big\|_{\text{tr}} \\
&\quad \leq \int_{-\infty}^{\mu} dy\, e^{\beta (y - \mu)} \Big\| \Big( \chi(H \leq \mu) R_{H}(y + i \pi/2\beta)^{q_{0} + 1} \widetilde{\mathcal{B}}(y + i \pi/2\beta)\Big) \mathcal{W}_{z_{1}}^{(n/2)} \Big\|_{\text{tr}} \\
&\quad \leq C \int_{-\infty}^{\mu} dy\, \beta e^{\beta (y - \mu)} \Big\| \frac{1}{(\beta(H - y))^{q_{0} + 1} + 1}  \mathcal{W}_{z_{1}}^{(n/4)}\Big\|_{\text{tr}} \\
&\quad \leq  C \int_{-\infty}^{\mu} dy\, \beta e^{\beta (y - \mu)} \Big\| \frac{(\beta(H - \mu))^{q_{0} + 1} + 1}{(\beta(H - y))^{q_{0} + 1} + 1}  \Big\|_{\text{op}} \Big\| \frac{1}{(\beta(H - \mu))^{q_{0} + 1} + 1}  \mathcal{W}_{z_{1}}^{(n/4)}\Big\|_{\text{tr}}\;; \\
%&\quad \leq K \varepsilon^{-2}\;,
\end{split}
\end{equation}
using that, for all $\lambda \in \mathbb{R}$:
\begin{equation*}
\begin{split}
\Big| \frac{(\beta(\lambda - \mu))^{q_{0} + 1} + 1}{(\beta(\lambda - y))^{q_{0} + 1} + 1}  \Big| &\leq C_{q_{0}} \Big| \frac{(\beta(\lambda - y))^{q_{0} + 1} + (\beta (y - \mu))^{q_{0}} + 1}{(\beta(\lambda - y))^{q_{0} + 1} + 1} \Big| \\
&\leq C_{q_{0}} (\beta (y - \mu))^{q_{0}}\;,
\end{split}
\end{equation*}
we have:
\begin{equation*}
\begin{split}
\Big\|\text{A}_{2} \mathcal{W}_{z_{1}}^{(n/2)}  \Big\|_{\text{tr}} &\leq K_{q_{0}} \int_{-\infty}^{\mu} dy\, \beta e^{\beta (y - \mu)} (\beta (y - \mu))^{q_{0}}  \Big\| \frac{1}{(\beta(H - \mu))^{q_{0} + 1} + 1}  \mathcal{W}_{z_{1}}^{(n/4)}\Big\|_{\text{tr}} \\
&\leq C \varepsilon^{-2}\;,
\end{split}
\end{equation*}
where in the last bound we used Assumption \ref{ass:Weyl}. Next, consider the term $\text{A}_{3}$ in Eq. (\ref{eq:A1A2A3}).  We have:
\begin{equation*}
\begin{split}
&\Big\| \text{A}_{3} \mathcal{W}_{z_{1}}^{(n/2)}  \Big\|_{\text{tr}} \\
&\quad \leq \int_{\mu}^{\infty}dy\, \Big\|  \Big( \chi(H \leq \mu) R_{H}(y + i \pi/2\beta)^{q_{0} + 1} \widetilde{\mathcal{B}}(y + i \pi/2\beta)\Big) \mathcal{W}_{z_{1}}^{(n/2)} \Big\|_{\text{tr}} \\
&\quad \leq C\beta  \int_{\mu}^{\infty}dy\, \Big\| \chi(H \leq \mu)\frac{1}{(\beta(H - y))^{q_{0} + 1} + 1}  \mathcal{W}_{z_{1}}^{(n/4)}\Big\|_{\text{tr}} \\
&\quad \leq K \int_{\mu}^{\infty} dy\, \frac{\beta}{(\beta (y - \mu))^{2} + 1} \Big\| \frac{1}{(\beta (H-\mu))^{q_{0}-1} + 1}\mathcal{W}_{z_{1}}^{(n/4)} \Big\|_{\text{tr}} \\
&\quad \leq C\varepsilon^{-2}\;.
\end{split}
\end{equation*}
again by Assumption \ref{ass:Weyl}, if we choose $q_0 \geq 2m+1$. This concludes the discussion of the term $\text{A}$ in Eq. (\ref{eq:ABC}). Consider now the term $\text{B}$. We have, similarly to the previous case:
\begin{equation*}
\begin{split}
&\Big\|  \text{B} \mathcal{W}_{z_{1}}^{(n/2)}  \Big\|_{\text{tr}} \\
&\leq \int_{-\infty}^{\mu} dy\, \Big\| \chi(H > \mu)  R_{H}(y + i \pi/2\beta)^{q_{0} + 1} \widetilde{\mathcal{B}}(y + i \pi/2\beta) \mathcal{W}_{z_{1}}^{(n/2)} \Big\|_{\text{tr}} \\
&\leq C \beta\int_{-\infty}^{\mu} dy\, \Big\|  \chi(H > \mu) \frac{1}{(\beta(H - y))^{q_{0} + 1} + 1} \mathcal{W}_{z_{1}}^{(n/4)}\Big\|_{\text{tr}} \\
&\leq K \int_{-\infty}^{\mu} dy\, \frac{\beta}{(\beta(y - \mu))^{2} + 1} \Big\| \frac{1}{(\beta(H - \mu))^{q_{0}-1} + 1} \mathcal{W}_{z_{1}}^{(n/4)} \Big\|_{\text{tr}}\\
&\leq K \varepsilon^{-2}\;.
\end{split}
\end{equation*}
Finally, let us consider the term $\text{C}$ in Eq. (\ref{eq:ABC}). We have:
\begin{equation*}
\begin{split}
&\Big\|  \text{C} \mathcal{W}_{z_{1}}^{(n/2)}  \Big\|_{\text{tr}} \\
&\quad \leq \int_{\mu}^{\infty} dy\, \Big| \frac{1}{1 + i e^{\beta (y - \mu)}}\Big| \Big\| R_{H}(y + i \pi/2\beta)^{q_{0} + 1} \widetilde{\mathcal{B}}(y + i \pi/2\beta) \mathcal{W}_{z_{1}}^{(n/2)} \Big\|_{\text{tr}} \\
&\quad \leq  C \int_{\mu}^{\infty} dy\, \beta e^{-\beta (y - \mu)} \Big\|  \frac{(\beta(H - \mu))^{q_{0} + 1} + 1}{(\beta(H - y))^{q_{0} + 1} + 1} \Big\|_{\text{op}} \Big\| \frac{1}{(\beta(H - \mu))^{q_{0} + 1} + 1} \mathcal{W}_{z_{1}}^{(n/4)}\Big\|_{\text{tr}} \\
&\quad \leq K\varepsilon^{-2}\;.
\end{split}
\end{equation*}
where we proceeded as after (\ref{eq:A2est}). This concludes the discussion of the terms in (\ref{eq:BR}). All in all, we proved that:
\begin{equation}\label{eq:Bestfin}
\Big\| \Big(\int_{\mathcal{C}} \frac{dz}{2\pi i}\, \frac{1}{1 + e^{\beta (z - \mu)}} \mathcal{B}\Big) \mathcal{W}_{z_{1}}^{(n/2)} \Big\|_{\text{tr}} \leq C \varepsilon^{-2}\;.
\end{equation}
\paragraph{Conclusion.} The estimates (\ref{eq:Aestfin}), (\ref{eq:Bestfin}) allow us to control all terms on the right-hand side of (\ref{eq:WoW}).  We conclude that, for $n$ large enough, $p = 2n - 4$, $q_{0} - 1\geq 2m$ (with $m$ as defined in Assumption \ref{ass:Weyl}) there exists $C > 0$ independent of $\Lambda$ such that 
\begin{equation*}
\Big\| \mathcal{W}_{z_{2}}^{(n/2)}  \omega \mathcal{W}_{z_{1}}^{(n)} \Big\|_{\text{tr}} \leq \frac{C\varepsilon^{-3}}{1 + |z_{1} - z_{2}|^{2n-4}}\;.
\end{equation*}
which proves (\ref{eq:sc1}). 
\subsection{Proof of Eq. (\ref{eq:sc1b})} The proof is very similar to the proof of (\ref{eq:sc1}). Here, we shall choose $p=2n$, compare with \eqref{eq:ad1}. The $j=0$ contribution in (\ref{eq:WoW}) is replaced by 
\[ \begin{split}  \frac{C}{1+|z_1 - z_2|^{2n}} \|\sqrt{\omega} \sqrt{1-\omega} \mathcal{W}_{z_1}^{(n/2)} \|_\text{tr} &\leq \frac{C}{1+|z_1 - z_2|^{2n}} \|  \sqrt{\omega} \sqrt{1-\omega} \mathcal{W}_{z_1}^{(n/2)} \|_\text{tr} \\ &\leq \frac{C\veps^{-2}}{1+|z_1 - z_2|^{2n}}\;, \end{split} \]
where we used (\ref{eq:oocbd}) and Assumption \ref{ass:Weyl}. Contributions with $j \geq 1$ can be handled as in the proof of (\ref{eq:sc1}), just replacing the integral representation (\ref{eq:omegaint}) by 
\[ \sqrt{\omega} \sqrt{1-\omega} = \int_{\mathcal{C}} \frac{dz}{2\pi i} \frac{e^{\beta (z-\mu)/2}}{1+ e^{\beta (z-\mu)}} \frac{1}{H-z}  \, .\]

\subsection{Proof of Eq.~(\ref{eq:sc2})}\label{sec:commeipx} We shall only discuss the case $\pi = \omega$, the other cases are analogous. In fact, for the cases $\pi = \sqrt{\omega}$ and $\pi = \sqrt{1 - \omega} $ we use the same integral representation (\ref{eq:omegaint}), where now the Fermi-Dirac function is replaced by:
\begin{equation*}
\sqrt{\frac{1}{1 + e^{\beta (z-\mu)}}} \qquad \text{and, respectively,}\qquad \sqrt{\frac{e^{\beta (z-\mu)}}{1 + e^{\beta (z-\mu)}}} \;.
\end{equation*}
If $|\text{Im} z| \leq \pi/2\beta$, these functions have the same qualitative properties of the Fermi-Dirac function: that is, they decay exponentially fast as $\text{Re}\, z \to \infty$, which is the only property we use. 

To begin, we estimate:
\begin{equation*}
\big \| \W_{z_{2}}^{(n/2)} \big[ e^{i p \cdot \hat{x}}, \omega \big] \W_{z_{1}}^{(n)}\big\|_{\tr} \leq \sum_{\alpha=1,2,3} |p_{\alpha}| \big \| \W_{z_{2}}^{(n/2)} \big[ \hat x_{\alpha}, \omega \big] \W_{z_{1}}^{(n)}\big\|_{\tr}\;.
\end{equation*}
Similarly to (\ref{eq:WoW}), recall also \eqref{eq:ad1}, we have, choosing $p=2n$:
\begin{equation}\label{eq:commsqrt}
\big \| \W_{z_{2}}^{(n/2)} \big[ \hat x_{\alpha}, \omega \big] \W_{z_{1}}^{(n)}\big\|_{\tr} \leq  \frac{K_{n}}{1 + |z_{1} - z_{2}|^{2n}} \sum_{j=1}^{2n+1} \sum_{\alpha} \Big\| \text{ad}^{j;\alpha}_{\hat x}(\omega) \mathcal{W}_{z_{1}}^{(n/2)}\Big\|_{\text{tr}}\;.
\end{equation}
%
%Then, recall the following representation for $\sqrt{\omega}$:
%%
%\begin{equation}\label{eq:sqrtomega}
%\sqrt{\omega} = \int_{\mathcal{C}} \frac{dz}{2\pi i} \frac{1}{\sqrt{1 + e^{\beta (z - \mu)}}} \frac{1}{H - z}\;,
%\end{equation}
%%
%with $\mathcal{C}$ as in (\ref{eq:Cdef0}), (\ref{eq:Cdef}). 
Using the integral representation (\ref{eq:omegaint}) for $\omega$, we can proceed exactly as we did for the $j>0$ terms on the r.h.s.~of (\ref{eq:WoW}). We ultimately get, for $n$ large enough, 
\begin{equation*}
\Big\| \text{ad}^{j;\alpha}_{\hat x}(\omega) \mathcal{W}_{z_{1}}^{(n/2)}\Big\|_{\text{tr}} \leq C \varepsilon^{-2}
\end{equation*}
Combined with (\ref{eq:commsqrt}), this proves (\ref{eq:sc2}).

%\subsection{Proof of Eqs. (\ref{eq:sc2}), (\ref{eq:sc3})} We shall only discuss (\ref{eq:sc3}), since (\ref{eq:sc2}) is completely %analogous. Observe that $\text{ad}^{j;\alpha}_{\hat x}(\sqrt{1-\omega}) = \text{ad}^{j;\alpha}_{\hat x}(\sqrt{1-\omega} - 1)$, and:
%%
%\begin{equation}\label{eq:1-o}
%[ \sqrt{1-\omega} - 1, \hat x_{\alpha} ] = \int_{\mathcal{C}} \frac{dz}{2\pi i} \Big(\sqrt{\frac{e^{\beta (z - \mu)}}{1 + e^{\beta (z - 
%\mu)}}} -1\Big) \Big[ \frac{1}{H - z}, \hat x_{\alpha}\Big]\;.
%\end{equation}
%%
%Next, observe that the function:
%%
%\begin{equation}\label{eq:gz}
%g(z) = \Big(\sqrt{\frac{e^{\beta (z - \mu)}}{1 + e^{\beta (z - \mu)}}} -1\Big)\;,
%\end{equation}
%%
%has the same qualitative features as the Fermi-Dirac function $1/(1 + e^{\beta (z - \mu)})$ on $\mathcal{C}$. In particular, for %every $m \in \bN$, there is $C_m > 0$ such that 
%%
%\begin{equation}\label{eq:bdg}
%\begin{split}
%| g(y \pm i \pi/2\beta) - 1 | &\leq \frac{C_{m}}{(\beta|y - \mu|)^{m} + 1}\qquad \text{for $y\leq \mu$} \\
%| g(y \pm i \pi/2\beta) | &\leq \frac{C_{m}}{(\beta|y - \mu|)^{m} + 1}\qquad \text{for $y> \mu$}\;.
%\end{split}
%\end{equation}
%%
%Thus, the proof of (\ref{eq:sc3}) follows the same lines of the proof of (\ref{eq:sc2}).

\subsection{Proof of Eq.~(\ref{eq:sc4})} \label{sec:commgrad}
Again, we consider only the case $\pi = \omega$, the other cases being analogous. Proceeding as in (\ref{eq:WoW}), we find, for $p=2n$:
\[ \big\|  \mathcal{W}_{z_2}^{(n/2)} \big[ \veps \nabla, \omega \big] \mathcal{W}_{z_1}^{(n)} \|_\text{tr} \leq \frac{C}{1+|z_1 - z_2|^{2n}} \sum_{j=0}^{2n} \sum_\alpha \big\| \big[ \veps \nabla , \text{ad}^{j;\alpha}_{\hat{x}} (\omega) \big] \mathcal{W}_{z_1}^{(n/2)} \big\|_\text{tr}\;. \] 
With the integral representation (\ref{eq:omegaint}), we need to estimate $\|  [\veps \nabla , \text{ad}^{j;\alpha}_{\hat{x}} (R_H (z))] \mathcal{W}_{z_1}^{(n/2)}\|$, with $R_H (z) = 1/(H-z)$. As in (\ref{eq:RRR}), we can express $\text{ad}^{j;\alpha}_{\hat{x}} (R_H (z))$ as a 
linear combination of terms like 
\[ R_H (z) \text{ad}_{\hat x}^{\ell_{1}}(H) R_{H}(z) \text{ad}_{\hat x}^{\ell_{2}}(H) R_{H}(z) \cdots R_{H}(z) \text{ad}_{\hat x}^{\ell_{r}}(H) R_{H}(z)\;. \]
Since the operators $\text{ad}_{\hat x}^{\ell_{1}}(H) =  \text{ad}_{\hat x}^{\ell_{1}}(H_0)$ commute with $\veps \nabla$, it follows that we can write  $[ \veps \nabla , \text{ad}^{j;\alpha}_{\hat{x}} (\sqrt{\omega})]$ as a linear combination of terms having the form
\[  R_H (z) M_1 R_H (z) M_2 R_{H}(z) \dots M_{r+1} R_H (z) \]
where, among the operators $M_1, \dots , M_{r+1}$, $r$ have the form $\text{ad}^{\ell_i}_{\hat{x}} (H)$ and one has the form $[\veps \nabla, H] = \veps \nabla V$. From here, the analysis is similar to what we did following (\ref{eq:EaEb}), moving resolvents to the left and distinguishing terms with at least $q_0+1$ resolvents at the left end of the string (but possibly other resolvents not at the leftmost place) and terms with all resolvents at the leftmost place. To move resolvents through the operator $\nabla V$, we use that $[H, \nabla V]$ is a bounded operator, and has bounded multiple commutators with $H$ and with $\hat{x}$. We omit further details.

%\subsection{Proof of Eq. (\ref{eq:sc5})} The proof of this bound uses the same integral representation of (\ref{eq:1-o}), and %follows the proof of (\ref{eq:sc4-}), (\ref{eq:sc4}); we omit the details.
%

\section{On the local semiclassical structure}\label{app:DG}

The goal of this section is to show that our Assumption \ref{ass:Weyl} is actually implied by other results in semiclassical analysis, such as the sharp, pointwise Weyl law. These results are well-known for non-relativistic Schr\"odinger operators in the semiclassical regime, and we believe that the techniques could be used to establish them also in our pseudo-relativistic setting. %In particular, it would be interesting to establish these results uniformly in the size of the classically confined region induced by $V_{\text{ext}}$.

To establish these equivalences, let us preliminarily discuss how the trace-norm bound (\ref{eq:localWeyl}) is implied by a similar estimate in Hilbert-Schmidt norm.

\subsection{Reduction to a Hilbert-Schmidt bound}

\begin{proposition} Let $m,n$ even and $m\leq n/2$. Then, the following bound holds true:
\begin{equation}\label{eq:tr2HS}
\Big\| \frac{1}{(\beta (H-\mu))^{2m} + 1} \mathcal{W}^{(n)}_{z} \Big\|_{\mathrm{tr}} \leq C_{m,n}\Big\| \frac{1}{(\beta (H-\mu))^{m} + 1} \mathcal{W}^{(n/2)}_{z} \Big\|_{\mathrm{HS}}^{2}\;.
\end{equation}
\end{proposition}
\begin{proof} We define:
\begin{equation*}
R_{\beta;H} := \frac{1}{\beta (H - \mu) + i}\;.
\end{equation*}
Observe that:
\begin{equation}\label{eq:resbd}
\frac{c_{m}}{(\beta (H-\mu))^{2m} + 1} \leq | R_{\beta;H} |^{2m} \leq \frac{C_{m}}{(\beta (H-\mu))^{2m} + 1}\;.
\end{equation}
Then, we have:
\begin{equation*}
\Big\| \frac{1}{(\beta (H - \mu))^{2m} + 1} \mathcal{W}_{z}^{(2n)} \Big\|_{\text{tr}} \leq K_{m} \Big\| R_{\beta;H}^{2m}  \mathcal{W}_{z}^{(2n)} \Big\|_{\text{tr}}\;.
\end{equation*}
We write, for $2k\leq m$: 
\begin{equation}\label{eq:d4}
\begin{split}
\Big\| R_{\beta;H}^{2m}  \mathcal{W}_{z}^{(2n)} \Big\|_{\text{tr}} &\leq \Big\| R_{\beta;H}^{2m}  \mathcal{W}_{z}^{(n)}  \mathcal{W}_{z}^{(n)} \Big\|_{\text{tr}} \\
&\leq \Big\| R_{\beta;H}^{2m}  \mathcal{W}_{z}^{(n)} R_{\beta;H}^{-2k} R_{\beta;H}^{2k}  \mathcal{W}_{z}^{(n)} \Big\|_{\text{tr}} \\
&\leq \Big\| R_{\beta;H}^{2m}  \mathcal{W}_{z}^{(n)} R_{\beta;H}^{-2k} \Big\|_{\text{HS}} \Big\| R_{\beta;H}^{2k}  \mathcal{W}_{z}^{(n)} \Big\|_{\text{HS}}\;.
\end{split}
\end{equation}
The second factor is of the desired form. Consider the first factor. We have:
\begin{equation*}
R_{\beta;H}^{2m}  \mathcal{W}_{z}^{(n)} R_{\beta;H}^{-2k} = R_{\beta;H}^{2m - 1}  \mathcal{W}_{z}^{(n)} R_{\beta;H}^{-2k + 1} + R_{\beta;H}^{2m - 1}  \Big[ R_{\beta;H}, \mathcal{W}_{z}^{(n)}  \Big] R_{\beta;H}^{-2k + 1}\;.
\end{equation*}
The commutator can be written as:
\begin{equation*}
\Big[ R_{\beta;H}, \mathcal{W}_{z}^{(n)}  \Big] = R_{\beta;H} [ \mathcal{W}_{z}^{(n)}, \beta H ] R_{\beta;H}\;,
\end{equation*}
which gives:
\begin{equation*}
R_{\beta;H}^{2m}  \mathcal{W}_{z}^{(n)} R_{\beta;H}^{-2k} = R_{\beta;H}^{2m - 1}  \mathcal{W}_{z}^{(n)} R_{\beta;H}^{-2k + 1} + R_{\beta;H}^{2m} \text{ad}_{\beta H}\big(\mathcal{W}_{z}^{(n)}\big) R_{\beta;H}^{-2k + 1}\;.
\end{equation*}
This formula can be iterated, until all resolvents on the right-hand side disappear. Thus, we find that:
\begin{equation}\label{eq:RRmeno1}
\Big\| R_{\beta;H}^{2m}  \mathcal{W}_{z}^{(n)} R_{\beta;H}^{-2k} \Big\|_{\text{HS}} \leq \sum_{j=0}^{2k} \sum_{\alpha: |\alpha| = 2k-j} C_{\alpha} \Big\| R_{\beta;H}^{2m - j} \text{ad}_{\beta H}^{\alpha}\big(\mathcal{W}_{z}^{(n)}\big) \Big\|_{\text{HS}}\;.
\end{equation}
Consider the multi-commutators of $H$ with $\W_z^{(n)}$. Observe that:
\begin{equation}\label{eq:never}
[ \beta H,  \mathcal{W}_{z}^{(n)}] =  [ \beta H_{0}, \mathcal{W}_{z}^{(n)} ] = \mathcal{W}_{z}^{(n)} [ \beta H_{0}, | \hat x - z |^{4n} ] \mathcal{W}_{z}^{(n)}\;,
\end{equation}
and, using that every commutator with $\hat x$ is bounded and that it introduces a factor $\varepsilon$, which compensates $\beta$:
\begin{equation*}
\Big\| O \mathcal{W}_{z}^{(n)} [ \beta H_{0}, | \hat x - z |^{4n} ] \mathcal{W}_{z}^{(n)} \Big\|_{\text{HS}} \leq C_{n} \Big\|  O \mathcal{W}_{z}^{(n)} \Big\|_{\text{HS}}\;.
\end{equation*}
More generally, to estimate the higher commutators, it is convenient to represent the right-hand side of (\ref{eq:never}) as a linear combination of:
\begin{equation}\label{eq:+comm}
\frac{1}{1 + | x - z |^{4n}} \text{ad}_{\hat x}^{\xi}(\beta H_{0}) \frac{(x-z)^{\gamma}}{1 + | x - z |^{4n}} =: \mathcal{W}_{z}^{(n)} \text{ad}_{\hat x}^{\xi}(\beta H_{0}) \mathcal{W}_{z}^{(n,\gamma)}
\end{equation}
where $|\gamma| + |\xi| = 4n$ and $|\xi| \geq 1$. Using the integral representation of the square root operator, we have, since $\beta \varepsilon = O(1)$:
\begin{equation}\label{eq:d12}
(\ref{eq:+comm}) = K_{\xi} \varepsilon^{|\xi|-1}  \int \frac{d\lambda}{\sqrt{\lambda}}\, \mathcal{W}_{z}^{(n)} \frac{(\varepsilon D)^{\xi}}{(\lambda + 1 - \varepsilon^{2} \Delta)^{|\xi|}}  \mathcal{W}_{z}^{(n,\gamma)}\;.
\end{equation}
Eq. (\ref{eq:d12}) is a good starting point for taking other commutators with $\beta H$. To this end, we observe that: the commutator of $\beta H$ with $\mathcal{W}_{z}^{(n)}$ reproduces the same structure we just obtained; the commutator of $\beta H$ with the differential operator has been already studied after (\ref{eq:Pk}), and it is given by a linear combination of terms of the form (\ref{eq:1com}); it remains to discuss the commutator of $\beta H$ with $\mathcal{W}_{z}^{(n,\gamma)}$. It is:
\begin{equation*}
\begin{split}
[ \beta H, \mathcal{W}_{z}^{(n,\gamma)} ] &= [ \beta H_{0}, \mathcal{W}_{z}^{(n,\gamma)} ] \\
&= [ \beta H_{0}, (x-z)^{\gamma} ] \mathcal{W}_{z}^{(n)} + (x-z)^{\gamma} [ \beta H_{0}, \mathcal{W}_{z}^{(n)}  ]\;.
\end{split}
\end{equation*}
The first term can be rewritten as a linear combination of:
\begin{equation*}
\text{ad}_{\hat x}^{\gamma_{1}}(\beta H_{0}) \mathcal{W}_{z}^{(n, \gamma_{2})}
\end{equation*}
with $\gamma_{1} + \gamma_{2} = \gamma$, while the second term can be written as a linear combination of:
\begin{equation*}
\mathcal{W}_{z}^{(n, \gamma)} \text{ad}_{\hat x}^{\gamma_{3}} (\beta H_{0})  \mathcal{W}_{z}^{(n, \gamma_{4})}
\end{equation*}
with $|\gamma_{3}| + |\gamma_{4}| = 4n$ and $|\gamma_{3}| \geq 1$. To summarize, we obtained that $\text{ad}_{\beta H}^{\alpha}\big(\mathcal{W}_{z}^{(n)}\big)$ with $|\alpha| = 2$ can be expressed as a string of operators of the form:
\begin{equation*}
\mathcal{W}_{z}^{(n)} O_{1} \cdots O_{L}\;,
\end{equation*}
where the only possibly unbounded operator arising is $O_{j} = D V$, while all the other operators are bounded. The unboundedness of $DV$ is controlled using the localization operator; this is done by observing that the commutators of $D V$ with all possible operators generated by the iteration are bounded. Thus, all in all we obtain that the generic term contributing to the right-hand side of (\ref{eq:RRmeno1}) is bounded as:
\begin{equation*}
\Big\| R_{\beta;H}^{2m - j} \text{ad}_{\beta H}^{\alpha}\big(\mathcal{W}_{z}^{(n)}\big) \Big\|_{\text{HS}} \leq \sum_{i=0}^{2k-j} C_{i} \Big\| R_{\beta;H}^{2m - j} \mathcal{W}_{z}^{(n)} (D V)^{i} \Big\|_{\text{HS}}\;;
\end{equation*}
hence, since $2k\leq m$ and $m\leq n/2$ we obtain:
\begin{equation*}
\Big\| R_{\beta;H}^{2m}  \mathcal{W}_{z}^{(n)} R_{\beta;H}^{-2k} \Big\|_{\text{HS}} \leq C_{m,n,k} \Big\| R_{\beta;H}^{m} \mathcal{W}_{z}^{(n/2)} \Big\|_{\text{HS}}\;.
\end{equation*}
Coming back to (\ref{eq:d4}), we have:
\begin{equation*}
\begin{split}
\Big\| R_{\beta;H}^{2m}  \mathcal{W}_{z}^{(2n)} \Big\|_{\text{tr}} &\leq C \Big\| R_{\beta;H}^{m} \mathcal{W}_{z}^{(n/2)} \Big\|_{\text{HS}}\Big\| R_{\beta;H}^{2k}  \mathcal{W}_{z}^{(n)} \Big\|_{\text{HS}} \\
&\leq C\Big\| R_{\beta;H}^{m}  \mathcal{W}_{z}^{(n/2)} \Big\|_{\text{HS}}^{2}\;.
\end{split}
\end{equation*}
Equivalently, by (\ref{eq:resbd}), choosing $m$ even:
\begin{equation*}
\begin{split}
\Big\| R_{\beta;H}^{2m}  \mathcal{W}_{z}^{(2n)} \Big\|_{\text{tr}} &\leq K \Big\| \Big(\frac{1}{(\beta (H-\mu))^{2m} + 1}\Big)^{1/2} \mathcal{W}_{z}^{(n/2)} \Big\|_{\text{HS}}^{2} \\
&\leq \Big\| \frac{1}{(\beta (H-\mu))^{m} + 1} \mathcal{W}_{z}^{(n/2)} \Big\|_{\text{HS}}^{2}\;.
\end{split}
\end{equation*}
This concludes the check of (\ref{eq:tr2HS}). 
\end{proof}

\subsection{Relation with the pointwise Weyl law}\label{sec:pointWeyl}
Let us define:
\begin{equation*}
\begin{split}
I_{V,\mu} &:= \varepsilon^{-3}\int dxdp\, (\mathcal{W}_{z}^{(n/2)}(x))^{2}  \mathbbm{1}\Big(\sqrt{1 + |p|^{2}} + V_{\text{ext}}(x) \leq \mu\Big) \\
&= \varepsilon^{-3} \frac{4\pi}{3} \int dx\, (\mathcal{W}_{z}^{(n/2)}(x))^{2} \big( ( \mu - V_{\text{ext}}(x) )^{2}_{+} - 1 \big)_{+}^{\frac{3}{2}}\;.
\end{split}
\end{equation*}
The next proposition establishes the connection between the local Weyl law and the validity of Assumption \ref{ass:Weyl}.
\begin{proposition}\label{eq:LWimplies} Suppose that there exists $0<\alpha<1$ such that for all $\nu \in [\mu - 2\varepsilon^{\alpha}; \mu + 2\varepsilon^{\alpha}]$:
\begin{equation}\label{eq:locWeyl}
\tr\, \mathcal{W}_{z}^{(n/2)} \mathbbm{1}(H \leq \nu) \mathcal{W}_{z}^{(n/2)} = I_{V,\nu} + R_{V,\nu}\;,
\end{equation}
with $|R_{V,\nu}|\leq K \varepsilon^{-2}$. Then, for $m, n$ large enough, and $\varepsilon$ small enough:
\begin{equation}\label{eq:HSbound}
\Big\| \frac{1}{(\beta (H-\mu))^{m} + 1} \mathcal{W}_{z}^{(n/2)} \Big\|_{\mathrm{HS}}^{2} \leq C_{m,n} (V_{\text{ext}}(z)^{2} + 1) \varepsilon^{-2}
\end{equation}
for a suitable constant $C_{m,n} \equiv C_{m,n}(K)$.
\end{proposition}
\begin{remark}
\begin{itemize}
\item[(i)] The bound (\ref{eq:HSbound}), combined with (\ref{eq:tr2HS}), shows that Assumption \ref{ass:Weyl} holds true with a constant $C\equiv C_{\kappa}$ for all $z$ such that $|V_{\text{ext}}(z)| \leq \kappa$.
\item[(ii)] It is not difficult to check (\ref{eq:locWeyl}) for $V_{\text{ext}} = 0$. More generally, this type of sharp asymptotics is known to hold for non-relativistic Schroedinger operators, see {\it e.g.} \cite{DG0}. It would be interesting to establish it for pseudo-relativistic Schroedinger operators with $V_{\text{ext}} \neq 0$, with a constant $K$ uniform in the size of the classically confined region associated with $V_{\text{ext}}$.
\end{itemize}
\end{remark}
\begin{proof} Let $\alpha > 0$ as in the assumption of the proposition. We start by estimating:
\begin{equation}\label{eq:minmag}
\begin{split}
&\Big\| \frac{1}{(\beta (H-\mu))^{m} + 1} \mathcal{W}_{z}^{(n/2)} \Big\|_{\text{HS}} \\
&\qquad \leq \Big\| \frac{ \mathbbm{1}(|H-\mu| > \varepsilon^{\alpha})}{(\beta (H-\mu))^{m} + 1} \mathcal{W}_{z}^{(n/2)} \Big\|_{\text{HS}} + \Big\| \frac{ \mathbbm{1}(|H-\mu| \leq \varepsilon^{\alpha})}{(\beta (H-\mu))^{m} + 1} \mathcal{W}_{z}^{(n/2)} \Big\|_{\text{HS}}\;.
\end{split}
\end{equation}
Consider the first term. Using that $\beta = O(\varepsilon^{-1})$, we have, for $m-q$ even:
\begin{equation}\label{eq:eq}
\Big\| \frac{ \mathbbm{1}(|H-\mu| > \varepsilon^{\alpha})}{(\beta (H-\mu))^{m} + 1} \mathcal{W}_{z}^{(n/2)} \Big\|_{\text{HS}} \leq C\varepsilon^{(1 - \alpha) q} \Big\| \frac{1}{(\beta (H-\mu))^{m-q} + 1} \mathcal{W}_{z}^{(n/2)} \Big\|_{\text{HS}}\;.
\end{equation}
Let us now bound the Hilbert-Schmidt norm on the right-hand side. We estimate, for $K$ large enough:
\begin{equation*}
\begin{split}
\Big\| \Big( \frac{1}{(\beta (H - \mu))^{m-q} + 1} \Big)  \mathcal{W}_{z}^{(n/2)} \Big\|_{\text{HS}}&\leq K  \Big\|  \Big( \frac{1}{(H-\mu)^{m-q} + K} \Big)  \mathcal{W}_{z}^{(n/2)} \Big\|_{\text{HS}} \\
&\leq \widetilde K  \Big\|  \Big( \frac{1}{H^{m-q} + K/2} \Big)\mathcal{W}_{z}^{(n/2)}\Big\|_{\text{HS}} \\
&\leq C \Big\| \frac{1}{(H + 1)^{m-q}} \mathcal{W}_{z}^{(n/2)} \Big\|_{\text{HS}}\;.
\end{split}
\end{equation*}
%
%Then, we rewrite the square of the norm in the right-hand side as:
%
%\begin{equation}\label{eq:tracexx}
%\tr\, \mathcal{W}^{(n/2)}_{z} \frac{1}{(H + 1)^{2(m-q)}} \mathcal{W}^{(n/2)}_{z} = \int dx\, \mathcal{W}^{(n/2)}_{z}(x)^{2} \frac{1}{(H + 1)^{2(m-q)}}(x;x)\;.
%\end{equation}
%
Using the representation
\begin{equation*}
 \frac{1}{(H + 1)^{2(m-q)}} = c_{2(m-q)} \int_{0}^{\infty} dt\, e^{-tH} e^{-t} t^{2(m-q)-1}\;,
\end{equation*}
and applying the Feynman-Kac formula, Proposition \ref{prop: Feynamn-Kac formula}, we get, proceeding as in Appendix \ref{sec: appen-apriori} to deal with the unboundedness of $V_{\text{ext}}$:
\begin{equation*}
\begin{split}
\tr\, \mathcal{W}^{(n/2)}_{z} \frac{1}{(H + 1)^{2(m-q)}} \mathcal{W}^{(n/2)}_{z} &= c_{2(m-q)} \int_{0}^{\infty} dt\, \langle \mathcal{W}^{(n/2)}_{z}, e^{-tH} \mathcal{W}^{(n/2)}_{z} \rangle e^{-t} t^{2(m-q)-1} \\
&\leq c_{2(m-q)} \int_{0}^{\infty} dt\, \langle \mathcal{W}^{(n/2)}_{z}, e^{-tH_{0}} \mathcal{W}^{(n/2)}_{z} \rangle e^{-t} t^{2(m-q)-1} \\
&= \tr\, \mathcal{W}^{(n/2)}_{z} \frac{1}{(H_{0} + 1)^{2(m-q)}} \mathcal{W}^{(n/2)}_{z}\;.
\end{split}
\end{equation*}
Therefore, for $m-q$ large enough:
\begin{equation}\label{eq:WFKbd}
\tr\, \mathcal{W}^{(n/2)}_{z} \frac{1}{(H + 1)^{2(m-q)}} \mathcal{W}^{(n/2)}_{z} \leq \tr\, \mathcal{W}^{(n/2)}_{z} \frac{1}{(H_{0} + 1)^{2(m-q)}} \mathcal{W}^{(n/2)}_{z} \leq C\varepsilon^{-3}\;,
\end{equation}
where the last inequality follows from an explicit computation. Thus, for $q$ large enough, from (\ref{eq:WFKbd}), (\ref{eq:eq}) we have:
\begin{equation}\label{eq:largeen}
\Big\| \frac{ \mathbbm{1}(|H-\mu| > \varepsilon^{\alpha})}{(\beta (H-\mu))^{m} + 1} \mathcal{W}_{z}^{(n/2)} \Big\|_{\text{HS}} \leq C\varepsilon^{-2}.
\end{equation}
Consider now the second term in (\ref{eq:minmag}). To estimate it, we use the local Weyl law (\ref{eq:locWeyl}) combined with a dyadic argument. We write:
\begin{equation*}
\begin{split}
\mathbbm{1}(|H-\mu| \leq \varepsilon^\alpha) &\leq \sum^{k_{*}}_{k\geq 1} \mathbbm{1}(2^{k-1}\varepsilon \leq |H-\mu| \leq 2^{k} \varepsilon) +  \mathbbm{1}(|H-\mu| \leq \varepsilon) \equiv \sum^{k_{*}}_{k\geq 0} f_{k}(H-\mu)\;,
\end{split}
\end{equation*}
where $k_{*}$ is an integer such that $\varepsilon^{\alpha} \leq 2^{k_{*}} \varepsilon \leq 2\varepsilon^{\alpha}$. Therefore, we have:
\begin{equation}\label{eq:afterdyadic}
\begin{split}
\Big\| \frac{ \mathbbm{1}(|H-\mu| \leq \varepsilon^{\alpha})}{(\beta (H-\mu))^{m} + 1} \mathcal{W}_{z}^{(n/2)} \Big\|_{\text{HS}} &\leq \sum_{k\geq 0}^{k_{*}} \Big\| \frac{ f_{k}(H-\mu)}{(\beta (H-\mu))^{m} + 1} \mathcal{W}_{z}^{(n/2)} \Big\|_{\text{HS}} \\
&\leq \sum_{k\geq 0}^{k_{*}} C^{m} 2^{-(k-1)m}  \Big\| f_{k}(H-\mu) \mathcal{W}_{z}^{(n/2)} \Big\|_{\text{HS}}\;.
\end{split}
\end{equation}
Consider the term with $k=0$. We have:
\begin{equation*}
\begin{split}
\Big\| f_{0}(H-\mu) \mathcal{W}_{z}^{(n/2)} \Big\|_{\text{HS}}^{2} &= \Big\| \mathbbm{1}(H \leq \mu + \varepsilon) \mathcal{W}_{z}^{(n/2)} \Big\|_{\text{HS}}^{2} -  \Big\| \mathbbm{1}(H < \mu -\varepsilon) \mathcal{W}_{z}^{(n/2)} \Big\|_{\text{HS}}^{2} \\
&= \tr\, \mathcal{W}_{z}^{(n/2)} \big( \mathbbm{1}(H \leq \mu + \varepsilon) - \mathbbm{1}(H \leq \mu - \varepsilon) \big) \mathcal{W}_{z}^{(n/2)}\;;
\end{split}
\end{equation*}
Using the bound
\begin{equation}\label{eq:I-I}
\begin{split}
I_{V,\mu+\varepsilon} - I_{V,\mu-\varepsilon} &\leq C\varepsilon^{-2} \int dx\, (\mathcal{W}_{z}^{(n/2)}(x))^{2} (V(x)^{2} + 1) \\
&\leq K \varepsilon^{-2} (V(z)^{2} + 1)\;,
\end{split}
\end{equation}
and applying the local Weyl law (\ref{eq:locWeyl}), we find 
%combined with the bound (\ref{eq:I-I}) for $\delta = \varepsilon$, we find:
%%
\begin{equation}\label{eq:f0bd}
\begin{split}
\Big\| f_{0}(H-\mu) \mathcal{W}_{z}^{(n/2)} \Big\|_{\text{HS}}^{2} &\leq I_{V,\mu+\varepsilon} - I_{V,\mu-\varepsilon} + C\varepsilon^{-2} \leq K \varepsilon^{-2} (V(z)^{2} + 1) + C\varepsilon^{-2}\;.
\end{split}
\end{equation}
Consider now the terms with $k\neq 0$ in (\ref{eq:afterdyadic}). Using that $f_{k}(H-\mu) = \mathbbm{1}(|H-\mu| \leq 2^{k}\varepsilon) - \mathbbm{1}(|H-\mu| < 2^{k-1}\varepsilon)$, we estimate:
\begin{equation*}
\begin{split}
\Big\| f_{k}(H-\mu) \mathcal{W}_{z}^{(n/2)} \Big\|_{\text{HS}} &\leq \Big\|  \mathbbm{1}(|H-\mu| \leq 2^{k}\varepsilon) \mathcal{W}_{z}^{(n/2)} \Big\|_{\text{HS}} + \Big\|  \mathbbm{1}(|H-\mu| < 2^{k-1}\varepsilon) \mathcal{W}_{z}^{(n/2)} \Big\|_{\text{HS}} \\
&\leq 2 \Big\|  \mathbbm{1}(|H-\mu| \leq 2^{k}\varepsilon) \mathcal{W}_{z}^{(n/2)} \Big\|_{\text{HS}}\;;
\end{split}
\end{equation*}
the right-hand side is bounded as in the case $k=0$, with the only difference that $\varepsilon$ is multiplied by $2^{k}$, with $k\leq k_{*}$. Thus, we obtained:
\begin{equation}\label{eq:fkbd}
\Big\| f_{k}(H-\mu) \mathcal{W}_{z}^{(n/2)} \Big\|_{\text{HS}} \leq 2 K \varepsilon^{-2} 2^{k} (V(z)^{2} + 1) + 2C\varepsilon^{-2}\;.
\end{equation}
In conclusion, plugging the bounds in (\ref{eq:f0bd}), (\ref{eq:fkbd}) in (\ref{eq:afterdyadic}), we find:
\begin{equation}\label{eq:smallen}
\begin{split}
\Big\| \frac{ \mathbbm{1}(|H-\mu| \leq \varepsilon^{\alpha})}{(\beta (H-\mu))^{m} + 1} \mathcal{W}_{z}^{(n/2)} \Big\|_{\text{HS}} &\leq \sum_{k\geq 0}^{k_{*}} C^{m} 2^{-(k-1)m} \big( 2 K \varepsilon^{-2} 2^{k} (V(z)^{2} + 1) + 2C\varepsilon^{-2}\big) \\
&\leq K_{m} (V(z)^{2} + 1) \varepsilon^{-2}\;.
\end{split}
\end{equation}
Putting together (\ref{eq:largeen}), (\ref{eq:smallen}),
\begin{equation*}
\Big\| \frac{1}{(\beta (H-\mu))^{m} + 1} \mathcal{W}_{z}^{(n/2)} \Big\|_{\text{HS}}^{2} \leq C_{m} (V(z)^{2} + 1) \varepsilon^{-2}\;,
\end{equation*}
which concludes the proof of (\ref{eq:HSbound}).
\end{proof}
\subsection{Relation with other local semiclassical estimates}

We conclude the appendix by showing that the bound (\ref{eq:HSbound}) is implied by analogous estimates, which can be found in the literature for non-relativistic Schroedinger operators \cite{DG}. 
\begin{proposition}\label{prp:fg} Let $f$ and $g$ be smooth compactly supported functions, and let $f_{\gamma}(\cdot) = f(\cdot / \gamma)$. Let $\mathcal{A}_{V,\mu}$ be the classically confined region:
\begin{equation*}
\mathcal{A}_{V,\mu} := \Big\{ x\in \mathbb{R} \, \big|\, V(x) - \mu < 0 \Big\}\;.
\end{equation*}
Suppose that the support of $f(\cdot)$ is inside the classically confined region. Let $z\in \mathbb{R}^{3}$ such that:
\begin{equation}\label{eq:zasin}
z\in \mathcal{A}_{V,\mu}\;,\qquad \mathrm{dist}(z, \mathcal{A}_{V,\mu}^{\text{c}}) \geq \varepsilon^{-\delta}
\end{equation}
with $\delta > 0$. Suppose that, for $\gamma \geq \varepsilon$:
\begin{equation}\label{eq:fg}
\| g(\hat x) f_{\gamma}(H - \mu) \|_{\mathrm{HS}}^{2} \leq C \varepsilon^{-3} \gamma\;.
\end{equation}
Then, for $n,m$ large enough:
\begin{equation*}
\Big\| \frac{1}{(\beta (H-\mu))^{m} + 1} \mathcal{W}_{z}^{(n/2)} \Big\|_{\mathrm{HS}}^{2} \leq C_{n,m}\varepsilon^{-2}\;.
\end{equation*}
\end{proposition}
\begin{remark} It is not difficult to check the bound (\ref{eq:fg}) for $V_{\text{ext}} = 0$. For $V_{\text{ext}} \neq 0$, the bound (\ref{eq:fg}) is available in the literature, for non-relativistic systems, \cite{DG}. It would be interesting to extend these estimates to the pseudo-relativistic case, with a constant $C$ that does not depend on the size of the classically confined region.
\end{remark}
\begin{proof}
As in Section \ref{sec:pointWeyl}, the proof is based on a dyadic argument, this time performed in energy and in space. The localization in energy is peformed as in Section \ref{sec:pointWeyl}; observe that in the dyadic argument of Section \ref{sec:pointWeyl} we could have replaced $\mathbbm{1}(\cdot)$ by $\chi(\cdot)$, with $\chi(\cdot)$ the smoothening of the characteristic function of the ball of radius $1$, such that $\chi(t) = 0$ for $t>2$. Thus, we are left with proving that the bound
\begin{equation}\label{eq:fWbd}
\| f_{k}(H-\mu) \mathcal{W}_{z}^{(n/2)} \|_{\text{HS}}^{2} \leq C 2^{k}\varepsilon^{-2} 
\end{equation}
is implied by (\ref{eq:fg}), with $f_k (x) = \chi(2^{k-1}\varepsilon \leq |x| \leq 2^{k} \varepsilon)$, for $k \geq 1$ and $f_0 (x) =  \chi(|x| \leq \varepsilon)$. Let $z\in \mathbb{R}^{3}$ as in (\ref{eq:zasin}). We write:
\begin{equation*}
\begin{split}
\mathcal{W}^{(n/2)}_{z}(x) &= \mathcal{W}^{(n/2)}_{z}(x) \chi(2\varepsilon^{\alpha} | x - z |) + \mathcal{W}^{(n/2)}_{z}(x) (1-\chi(2\varepsilon^{\alpha} | x - z |)) \\
&\equiv A_{z} + B_{z}\;.
\end{split}
\end{equation*}
Observe that, by the assumption on $z$, the function $x\mapsto \chi(2\varepsilon^{\alpha} | x - z |)$ is supported inside $\mathcal{A}_{V,\mu}$. We then estimate:
\begin{equation}\label{eq:WAB}
\begin{split}
&\Big\|  f_{k}(H-\mu) \mathcal{W}^{(n/2)}_{z}\Big\|_{\text{HS}} \\
&\qquad \leq \Big\| f_{k}(H-\mu) A_{z}\Big\|_{\text{HS}} + \Big\| \Big( f_{k}(H-\mu) B_{z} \Big\|_{\text{HS}}\;.
\end{split}
\end{equation}
Consider the second term. Using that $|x-z| \geq \varepsilon^{-\alpha}$, and choosing $n \equiv n(\alpha)$ large enough, we have $B_{z}(x) \leq \sqrt{\varepsilon} \mathcal{W}^{(k)}_{z}(x)$ for some $k<n/2$. Therefore:
\begin{equation*}
\begin{split}
\Big\| f_{k}(H-\mu) B_{z} \Big\|_{\text{HS}} &\leq \sqrt{\varepsilon} \Big\|  f_{k}(H-\mu) \mathcal{W}^{(k)}_{z}\Big\|_{\text{HS}} \\
&\leq  \sqrt{\varepsilon} \Big\|  f_{k}(H-\mu) \mathcal{W}^{(k)}_{z}\Big\|_{\text{HS}}\;.
\end{split}
\end{equation*}
Since $f_{k}(H-\mu) \leq C_{m} / (H + 1)^{m}$, we have:
\begin{equation*}
\Big\|  f_{k}(H-\mu) B_{z}\Big\|_{\text{HS}} \leq C \sqrt{\varepsilon} \Big\| \frac{1}{(H + 1)^{m}} \mathcal{W}^{(k)}_{z} \Big\|_{\text{HS}}\;.
\end{equation*}
Recalling the bound (\ref{eq:WFKbd}), we easily get:
\begin{equation}\label{eq:fBz}
\Big\|  f_{k}(H-\mu) B_{z}\Big\|_{\text{HS}}^{2} \leq K \varepsilon^{-2}\;.
\end{equation}
Consider now the first term in (\ref{eq:WAB}). We write:
\begin{equation*}
A_{z}(x) = \sum_{\ell=0}^{\ell_{*}} \chi_{\ell,z}(x) \mathcal{W}^{(n)}_{z}(x)\;,
\end{equation*}
where:
\begin{equation*}
\begin{split}
\chi_{\ell,z}(x) &= \chi(2^{-\ell} |x-z|) - \chi(2^{-(\ell-1)} |x-z|)\qquad \text{for $\ell>0$,} \\
\chi_{0,z}(x) &= \chi(|x - z|)\;,
\end{split}
\end{equation*}
where $\ell_{*}$ is the smallest integer such that $2^{-\ell_{*}} \leq 2\varepsilon^{\alpha}$. All the functions $\chi_{\ell,z}(x)$ are supported in $\mathcal{A}_{V,\mu}$. We estimate:
\begin{equation*}
\begin{split}
\Big\| f_{k}(H-\mu) A_{z} \Big\|_{\text{HS}} &\leq \sum_{\ell = 0}^{\ell_{*}} \Big\| f_{k}(H-\mu) \chi_{\ell,z} \mathcal{W}^{(n/2)}_{z}\Big\|_{\text{HS}} \\
&\leq \sum_{\ell=0}^{\ell_{*}} 2^{-2n\ell} \Big\| f_{k}(H-\mu) \chi_{\ell,z}\Big\|_{\text{HS}}\;,
\end{split}
\end{equation*}
where we used that $\mathcal{W}^{(n/2)}_{z} \simeq 2^{-2n\ell}$ in the support of $\chi_{\ell,z}$. It is convenient to further decompose $\chi_{\ell,z}$ into a sum of functions with compact support with volume of order $1$. To do this, it is useful to visualize $\chi_{\ell,z}$ as a smoothening of an annulus, with width of order $1$, and radius of order $2^{\ell}$. Thus, it it is clear that we can cover this domain with the union of $2^{2\ell}$ overlapping balls, of radius of order $1$. Correspondingly, we have:
\begin{equation*}
\chi_{\ell,z}(x) = \sum_{j=1}^{2^{2\ell}} \chi_{\ell,j,z}(x)\;,
\end{equation*}
where $\{\chi_{\ell,j,z}(x)\}$ are the smooth characteristic functions of such balls. We then get:
\begin{equation}\label{eq:AestW}
\Big\| f_{k}(H-\mu) A_{z} \Big\|_{\text{HS}} \leq \sum_{\ell=0}^{\ell_{*}} 2^{-2n\ell} \sum_{j=1}^{2^{2\ell}} \Big\| f_{k}(H-\mu) \chi_{\ell,j,z}\Big\|_{\text{HS}}\;.
\end{equation}
Thus, (\ref{eq:fg}) implies that, with our choice of compactly supported function:
\begin{equation*}
\begin{split}
\Big\| f_{k}(H-\mu) A_{z} \Big\|_{\text{HS}} &\leq C\sum_{\ell=0}^{\ell_{*}} 2^{-2n\ell} 2^{2\ell} 2^{k/2}\varepsilon^{-1} \\
&\leq K 2^{k/2}\varepsilon^{-1}\;.
\end{split}
\end{equation*}
Combined with (\ref{eq:fBz}), this proves (\ref{eq:fWbd}), and concludes the proof of Proposition \ref{prp:fg}.
\end{proof}


\begin{thebibliography}{10}

%\bibitem{Araki}
%H.~Araki. Quasifree States of CAR and Bogoliubov Automorphisms. \textit{Publ. RIMS Kyoto Univ.} \textbf{6} 384
%(1970/71).

\bibitem{AW}
H.~Araki and W.~Wyss. Representations of canonical anticommutation relations. \textit{Helv. Phys. Acta} \textbf{37} 136
(1964).

%\bibitem{APPP} A.~Athanassoulis, T.~Paul, F.~Pezzotti, M.~Pulvirenti. Strong Semiclassical Approximation of Wigner Functions for the Hartree Dynamics. {\it Rend. Lincei Mat.~Appl.~}{\bf 22}, 525--552 (2011).

%\bibitem{Bach} V.~{Bach}. {Error bound for the {H}artree--{F}ock energy of atoms and molecules}. \emph{Comm.~Math.~Phys.~}\textbf{147}, no.~3, 527--548 (1992).

\bibitem{BBPPT} V.~Bach, S.~Breteaux, S.~Petrat, P.~Pickl, T.~Tzaneteas. Kinetic energy estimates for the accuracy of the time-dependent Hartree--Fock approximation with Coulomb interaction. {\it J.~Math.~Pures Appl.~}{\bf 105}(1), 1--30 (2016).

\bibitem{BLS} V.~Bach, E.~H.~Lieb and J.~P. Solovej. Generalized Hartree-Fock theory and the Hubbard model. {\it J. Stat. Phys.} {\bf 76}, 3-89 (1994).
 
\bibitem{BGGM} C.~{Bardos}, F.~{Golse}, A.~D.~{Gottlieb} and N.~J.~{Mauser}. {Mean-field dynamics of fermions and the time-dependent {H}artree--{F}ock equation}. \emph{J.~Math.~Pures Appl.~(9)} \textbf{82}, no.~6, 665--683 (2003).

\bibitem{BJPSS} N.~{Benedikter}, V.~Jaksic, M.~{Porta}, C.~Saffirio and B.~{Schlein}. Mean-field Evolution of Fermionic Mixed States. {\it Commun.~Pur.~Appl.~Math.~}\textbf{69}, 2250--2303 (2016).

\bibitem{BNPSSa} N.~Benedikter, P.~T.~Nam, M.~Porta, B.~Schlein, R.~Seiringer. Optimal upper bound for the correlation energy of a Fermi gas in the mean-field regime. {\it Comm.~Math.~Phys.~}{\bf 374}, 2097--2150 (2020).

\bibitem{BNPSSb} N.~Benedikter, P.~T.~Nam, M.~Porta, B.~Schlein, R.~Seiringer. Correlation Energy of a Weakly Interacting Fermi Gas. {\it Invent.~Math.~}{\bf 225}, 885--979 (2021)

\bibitem{BPSScorr} N.~Benedikter, M.~Porta, B.~Schlein, R.~Seiringer. Correlation Energy of a Weakly Interacting Fermi Gas with Large Interaction Potential. {\it Arch. Rational Mech. Anal.} {\bf 247}, 65 (2023).

\bibitem{BNPSS} N.~{Benedikter}, P.~T.~Nam, M.~{Porta}, B.~{Schlein} and R.~Seiringer.
Bosonization of Fermionic Many-Body Dynamics. {\it Ann. H. Poincar\'e} {\bf 23}, 1725-1764 (2022).

\bibitem{BPSS} N.~{Benedikter}, M.~{Porta}, C.~Saffirio and B.~{Schlein}. {From the Hartree dynamics to the Vlasov equation.} {\it Arch.~Rational Mech.~Anal.~}{\bf 221}, 273--334 (2016)

\bibitem{BPS} N.~{Benedikter}, M.~{Porta} and B.~{Schlein}. {Mean-field evolution of fermionic systems}. \emph{Comm.~Math.~Phys.~}\textbf{331}, 1087--1131 (2014).

\bibitem{BPS2} N.~{Benedikter}, M.~{Porta} and B.~{Schlein}. {Mean-field dynamics of fermions with relativistic dispersion}. \emph{J.~Math.~Phys.~}\textbf{55}, 021901 (2014).

%\bibitem{BPSbook} N.~{Benedikter}, M.~{Porta} and B.~{Schlein}. {\it Effective evolution equations from quantum dynamics}. SpringerBriefs in Mathematical Physics {\bf 7} (2016).

%\bibitem{Bratteli}
%O.~Bratteli and D.W.~Robinson. \textit{Operator algebras and quantum statistical mechanics}, Volume 1-2, Springer (2002).

\bibitem{Carmona}
R.~Carmona. Path integrals for relativistic Schr\"odinger operators. In: Holden, H., Jensen, A. (eds) Schr\"odinger Operators. Lecture Notes in Physics, vol \textbf{345}. Springer, Berlin, Heidelberg.

\bibitem{CS}
R.~Carmona, W.C.~Masters and B.~Simon. Relativistic Schr\"odinger operators: Asymptotic behaviour of the eigenfunctions. \textit{J. Funct. Anal.} \textbf{91} (1990).

\bibitem{CHN} M.~R.~Christiansen, C.~Hainzl, P.~T.~Nam. The Random Phase Approximation for Interacting Fermi Gases in the Mean-Field Regime. {\it Forum of Mathematics, Pi} {\bf 11}, e32 (2023).

\bibitem{CHN2} M.~R.~Christiansen, C.~Hainzl, P.~T.~Nam. The Gell-Mann--Brueckner Formula for the Correlation Energy of the Electron Gas: A Rigorous Upper Bound in the Mean-Field Regime. {\it Commun. Math. Phys.} {\bf 401}, 1469-1529 (2023).

\bibitem{CHN3} M.~R.~Christiansen, C.~Hainzl, P.~T.~Nam. The Correlation Energy of the Electron Gas in the Mean-Field Regime. {\tt arXiv:2405.01386}

\bibitem{CLS} J.~J.~Chong, L.~Lafleche, C.~Saffirio. From many-body quantum dynamics to the Hartree--Fock and Vlasov equations with singular potentials. {\it J. Eur. Math. Soc.} {\bf 26}, 4923–5007 (2024).

%\bibitem{DFPP} 
%D.-A.~Deckert, J.~Fr\"ohlich, P.~Pickl, A.~Pizzo. Effective Dynamics of a Tracer Particle Interacting with an Ideal Bose Gas. {\it Comm.~Math.~Phys.~}{\bf 328} (2014), 597--624.

\bibitem{DG0}
A.~Deleporte, G.~Lambert.
Universality for free fermions and the local Weyl law for semiclassical Schrödinger operators.
{\it J. Eur. Math. Soc.} (2024), DOI 10.4171/JEMS/1447

\bibitem{DG}
A.~Deleporte, G.~Lambert. Widom's conjecture: variance asymptotics and entropy bounds for counting statistics of free fermions. Preprint arXiv:2405.07796.

\bibitem{DS} M. Dimassi and J. Sj\"ostrand. {\it Spectral Asymptotics in the Semi-Classical Limit}. Cambridge University Press (1999).

\bibitem{EESY} A.~{Elgart}, L.~{Erd{\H{o}}s}, B.~{Schlein} and H.-T.~{Yau}. {Nonlinear {H}artree equation as the mean-field limit of weakly coupled fermions}. \emph{J.~Math.~Pures Appl.~(9)} \textbf{83}, no.~10, 1241--1273 (2004).

\bibitem{FPS}
L.~Fresta, M.~Porta, B.~Schlein. Effective Dynamics of Extended Fermi Gases in the High-Density Regime. \textit{Comm.~Math.~Phys.} (2023).

\bibitem{FK} J.~{Fr{\"o}hlich} and A.~{Knowles}. {A microscopic derivation of the time-dependent {H}artree-{F}ock equation with {C}oulomb two-body interaction}, \emph{J.~Stat.~Phys.~}\textbf{145}, no.~1, 23--50 (2011).

%\bibitem{Golden}
%S. Golden. Lower Bounds for the Helmhotz Function, \textit{Phys. Rev.} \textbf{137} (1965).
% 
%\bibitem{GS} G.~M.~Graf and J.~P.~Solovej. A correlation estimate with applications to quantum systems with Coulomb interactions. \emph{Rev.~Math.~Phys.~}\textbf{6}, 977--997 (1994).

%\bibitem{HPR}  C.~Hainzl, M.~Porta, F.~Rexze. On the correlation energy of interacting fermionic systems in the mean-field regime. {\it Comm.~Math.~Phys.~}{\bf 374}, 485--524 (2020)

%\bibitem{LS} L.~Lafleche and C.~Saffirio. Strong semiclassical limit from Hartree and Hartree--Fock to Vlasov-Poisson equation. {\it Analysis and PDE} (to appear). 

%\bibitem{LeSa} M.~Lewin and J.~Sabin. The Hartree and Vlasov equations at positive density. {\it Comm.~Part.~Differ.~Equat.~}{\bf 45}, 1702--1754 (2020).

%\bibitem{LiSa} E.~H.~Lieb and R.~Seiringer. {\it The Stability of Matter in Quantum Mechanics.} Cambridge University Press (2010). 

%\bibitem{LP} P.-L.~Lions and T.~Paul. Sur les mesures de Wigner. {\it Rev.~Mat.~Iberoamericana} {\bf 9}, 553--618 (1993).

\bibitem{Hiroshima}
J.~L\"orinczi, F.~Hiroshima, and V.~Betz. Feynman-Kac-Type Formulae and Gibbs Measures, Berlin, Boston: De Gruyter, (2020).

%\bibitem{MM} P.~A.~Markowich and N.~J.~Mauser. The Classical Limit of a Self-Consistent Quantum Vlasov Equation. {\it Math.~Models Methods Appl.~Sci.~}{\bf 3}, no. 1, 109--124 (1993).

%\bibitem{MP} D.~Mitrouskas, P.~Pickl. Effective pair interaction between impurity particles induced by a dense  Fermi gas. {\tt arXiv:2105.02841}
 
\bibitem{NS} H.~{Narnhofer} and G.~L.~{Sewell}. {Vlasov hydrodynamics of a quantum mechanical model}. \emph{Comm.~Math.~Phys.~}\textbf{79}, no.~1, 9--24 (1981).

\bibitem{PP} S.~{Petrat} and P.~{Pickl}. {A new method and a new scaling for deriving fermionic mean-field dynamics}. {\it Math.~Phys.~Anal.~Geom.~}{\bf 19}: 3, (2016).

\bibitem{PRSS} M.~Porta, S.~Rademacher, C.~Saffirio and B.~Schlein. Mean Field Evolution of Fermions with Coulomb Interaction. \textit{J.~Stat.~Phys.~}\textbf{166}, 1345--1364 (2017).

%\bibitem{EP} E.~Presutti. {\it Scaling Limits in Statistical Mechanics and Microstructures in Continuum Mechanics}. Theoretical and Mathematical Physics, Springer (2008). 

\bibitem{Solovej} J.~P.~Solovej. {{Many Body Quantum Mechanics}}. Lecture Notes. Summer 2007. 
%Available at www.mathematik.uni-muenchen.de/$\widetilde
%$sorensen/Lehre/SoSe2013/MQM2/skript.pdf.

\bibitem{Sp} H.~{Spohn}. {On the {V}lasov hierarchy}, \emph{Math.~Methods Appl.~Sci.~}\textbf{3}, no.~4, 445--455 (1981).

%\bibitem{Thompson}
%C. J. Thompson. Inequality with applications in statistical mechanics. \textit{Journal of Mathematical Physics}, \textbf{6} (1965).
 
\end{thebibliography}
\end{document}